\documentclass[11pt]{amsart}
\usepackage{enumerate}
\usepackage{hyperref}
\usepackage{amsmath,amssymb,amsthm}
\usepackage{tikz}
\usepackage{algorithm}%
\usepackage{algpseudocode}
\usepackage[margin=1.3in]{geometry}
\usepackage{tikz}

\usetikzlibrary{shapes.geometric, arrows,matrix,backgrounds}

\hypersetup{
    pdftitle = {The ``art of trellis decoding'' is fixed-parameter tractable},
   pdfauthor=  {Jisu Jeong, Eun Jung Kim, and Sang-il Oum}
   }

\newcommand\abs[1]{\lvert #1\rvert}
\newcommand\spn[1]{\langle #1 \rangle}
\newtheorem{THM}{Theorem}[section]
\newtheorem{LEM}[THM]{Lemma}
\newtheorem*{THMMAIN}{Theorem \ref{thm:mainthm}}
\newtheorem*{THMRWD}{Theorem \ref{thm:lrw}}
\newtheorem*{CORCWD}{Corollary~\ref{cor:cwd}}
\newtheorem*{THMSUBSP}{Theorem \ref{thm:subspace}}
\newtheorem{COR}[THM]{Corollary}
\newtheorem{PROP}[THM]{Proposition}

\newtheorem{QUE}{Question}
\theoremstyle{remark}

\theoremstyle{definition}

\newcommand\rank{\operatorname{rank}}

\newcommand\F{\mathbb F}
\newcommand\FF{\mathcal F}
\newcommand\FS{\operatorname{FS}_k}
\renewcommand\L{\mathcal L}
\newcommand\V{\mathcal V}
\newcommand\B{\mathcal B}
\newcommand\sle\precsim
\newcommand\sge\succsim
\newcommand\seq\sim
\newcommand\tle\preccurlyeq
\newcommand\tge\succcurlyeq
\newcommand\teq\cong
\newcommand\ple\sqsubseteq
\newcommand\pge\sqsupseteq

\newcommand \splus \oplus
\newcommand \tplus \oplus

\newcommand \fplus {\operatorname{\oplus}}

\newcommand\poly{\operatorname{poly}}
\newcommand\up{\operatorname{\sf up_k}}

\newcommand\Hlineny{Hlin{\v e}n{\'y}}

\begin{document}
\title{The ``art of trellis decoding'' is fixed-parameter tractable}
\author{Jisu Jeong}
\author{Eun Jung Kim}
\author{Sang-il Oum}
\address[Kim]{CNRS, LAMSADE, Place du Marechal de Lattre de Tassigny, 75775 Paris cedex 16, France}
\address[Jeong, Oum]{Department of Mathematical Sciences, KAIST, 291 Daehak-ro
  Yuseong-gu Daejeon, 34141 South Korea}
\email{jjisu@kaist.ac.kr}
\email{eunjungkim78@gmail.com}
\email{sangil@kaist.edu}
\thanks{The first and last authors are supported by Basic Science Research
  Program through the National Research Foundation of Korea (NRF)
  funded by  the Ministry of Science, ICT \& Future Planning
  (2011-0011653).}
\thanks{This paper was presented in part at the Twenty-Seventh Annual ACM-SIAM Symposium on Discrete Algorithms
(SODA 2016) \cite{JKO2016a}.}
\date{\today}
\begin{abstract}

Given $n$ subspaces of a finite-dimensional vector space over a fixed finite field $\mathbb F$, we wish to find a linear layout $V_1,V_2,\ldots,V_n$ of the subspaces such that $\dim((V_1+V_2+\cdots+V_i) \cap (V_{i+1}+\cdots+V_n))\le k$ for all $i$;
such a linear layout is said to have width at most $k$.
When restricted to $1$-dimensional subspaces, this problem is equivalent to
computing the trellis-width (or minimum trellis state-complexity) of  a linear code in coding theory
and  computing the path-width of an $\mathbb F$-represented matroid in matroid theory.

  We present a fixed-parameter tractable algorithm to construct a linear layout of width at most $k$,  if it exists, for input subspaces of a finite-dimensional vector space over $\mathbb F$. As corollaries, we obtain a fixed-parameter tractable algorithm to produce a path-decomposition of width at most $k$ for an input $\mathbb F$-represented matroid of path-width at most $k$, and a fixed-parameter tractable algorithm to find a linear rank-decomposition of width at most $k$ for an input graph of linear rank-width at most $k$. In both corollaries, no such algorithms were known previously.
Our approach is based on dynamic programming combined with the idea developed by Bodlaender and Kloks (1996) for their work on path-width and tree-width of graphs.

It was previously known that a fixed-parameter tractable algorithm exists for the decision version of the problem for matroid path-width; 
a theorem by Geelen, Gerards, and Whittle~(2002) implies that for each fixed finite field $\mathbb F$, there are finitely many forbidden $\mathbb F$-representable minors for the class of matroids of path-width at most $k$. 
An algorithm by Hlin{\v e}n{\'y} (2006) can detect a minor in an input $\mathbb F$-represented matroid of bounded branch-width. 
However, this indirect approach would not produce an actual path-decomposition.
Our algorithm is the first one to construct such a path-decomposition and does not depend on the finiteness of forbidden minors.
\end{abstract}
\maketitle

\section{Introduction}\label{sec:intro}

In telecommunication, a message sent through a communication channel is exposed to noise which corrupts the original message. Decoding the original message from a received signal stream is a major issue in coding theory. A linear code is an error-correcting code, for which the alphabet is an element of a finite field $\F$ and its codewords form a subspace of a fixed vector space over $\F$. 
Many algorithms for decoding a linear (block or convolutional) code rely on a graphic interpretation of a linear code called the trellis.  A \emph{trellis} for a linear code $C$ of length $n$ is a layered directed graph such that the following hold:  

\begin{itemize}
\item The vertex set is partitioned into $V_0,V_1,\ldots, V_n$.
\item Every edge starts at a vertex in $V_{i-1}$ and ends at a vertex in $V_{i}$ for some $i\in \{1,2,\ldots,n\}$.
\item Every vertex is on a path from a vertex in $V_0$ to a vertex in $V_n$.
\item Every edge is labeled by an element of $\F$.
\item Each path from $V_0$ to $V_n$ corresponds to a codeword in $C$.
\end{itemize}
The weight of an edge is defined by the square of the error, 
that is the difference between the label on the edge and the received signal at position $i$. Many decoding algorithms, such as the monumental algorithm by Viterbi~\cite{Viterbi67}, for a linear (block or convolutional) code compute a minimum-weight path from $V_0$ to $V_n$ on the trellis. Such a path corresponds to the most likely original message.

Obtaining a small trellis is crucial in reducing the complexity of a trellis-based decoding algorithm, because its space complexity depends on (log of) $\max_{0\le i\le n}\abs{V_i}$. 
By permuting the coordinates of a linear code, one may produce an equivalent code with a smaller trellis. 
Massey~\cite{Massey1978} asks whether one can find a good permutation of the coordinates to make a smallest trellis, calling it the ``art of trellis decoding'' of a linear code.
The example in Figure~\ref{fig:trellis} from the excellent survey of Vardy~\cite{Vardy1998} on trellises depicts two equivalent %
linear codes. The first one corresponds to the linear code $C$ generated by the basis $(1,0,0,0,0,1)$, $(0,1,0,1,0,0)$, and $(0,0,1,0,1,0)$ and the second code $C'$ is generated by $(1,1,0,0,0,0)$, $(0,0,1,1,0,0)$, and $(0,0,0,0,1,1)$,  obtained by a permutation $\pi=(2,3,6)$ on coordinates.
The \emph{trellis-width}, also called the \emph{trellis-state complexity}, of a linear code measures the minimum bit size of memory registers required for the decoding after reordering the coordinates. %

\begin{figure}
	\begin{center}
		\includegraphics[scale=0.6]{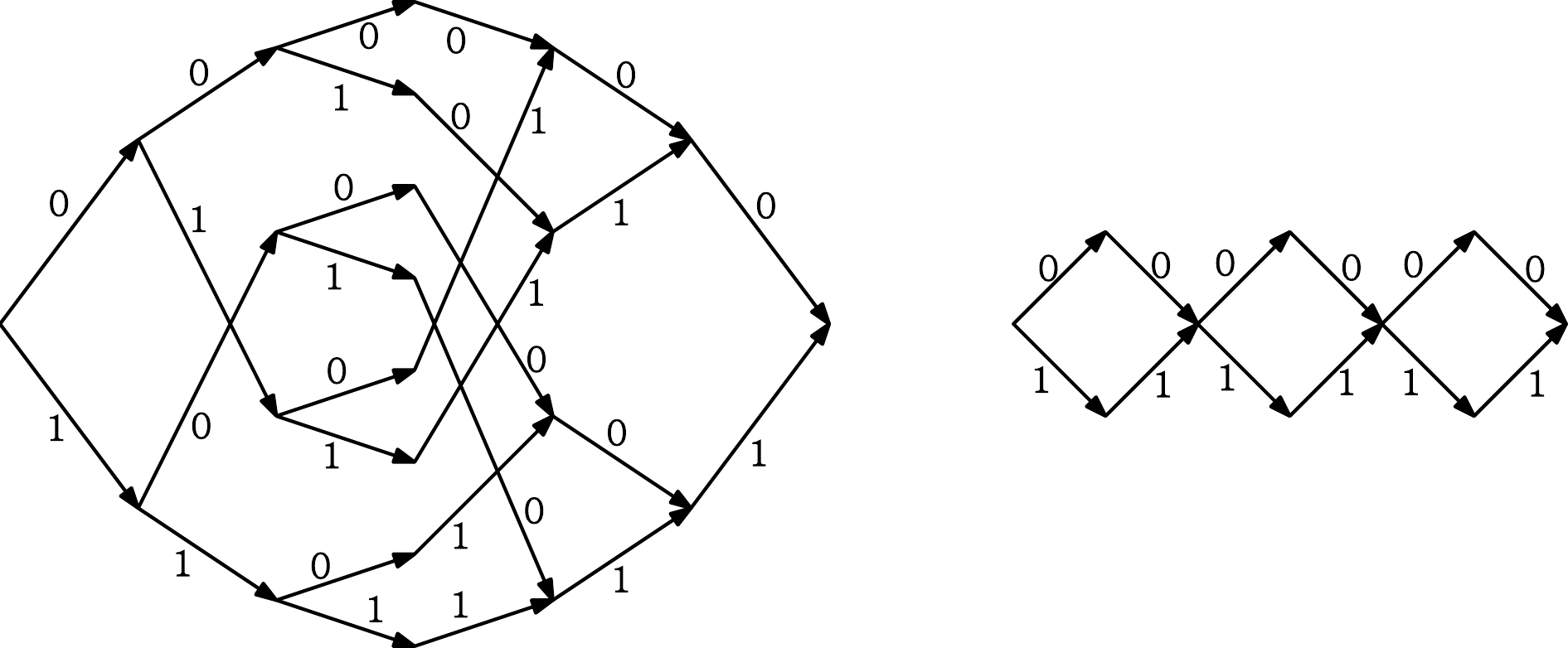}
	\end{center}
  \caption{Equivalent linear codes can produce trellises of different sizes~\cite{Vardy1998}.}
  \label{fig:trellis}
\end{figure}

The \emph{trellis-width of a linear code} can be stated equivalently in terms of vectors over a finite field; if a linear code is generated by a matrix $\left( v_1 ~ v_2 ~ \cdots ~ v_n\right)$, then
its trellis-width can be determined by finding a desired linear layout of the vectors $v_1$, $v_2$, $\ldots$, $v_n$ as follows.
\begin{description}
\item [Input] $n$ vectors $v_1,\ldots,v_n$ in a vector space over a fixed finite field $\F$, an integer $k$.
\item [Parameter] $k$.
\item [Problem] Determine whether there exists a permutation $\sigma$ of $\{1,2,\ldots,n\}$ such that 
\[\dim \spn{ v_{\sigma(1)},v_{\sigma(2)},\ldots,v_{\sigma(i)}} \cap
 \spn{ v_{\sigma(i+1)},v_{\sigma(i+2)},\ldots,v_{\sigma(n)}} \le k\]
for all $1\le i\le n-1$
and find a linear layout $v_{\sigma(1)}$, $v_{\sigma(2)}$, $\ldots$, $v_{\sigma(n)}$ if one exists.
\end{description}
The minimum $k$ having such a linear layout is called the \emph{trellis-width} of a linear code generated by a matrix $\left( v_1 ~ v_2 ~ \cdots ~ v_n\right)$.
This problem and its equivalent variants have been studied under various names in the coding theory literature, such as \textsc{Maximum Partition Rank Permutation}~\cite{HK1996}, \textsc{Maximum Width}~\cite{JMV1998}, and \textsc{Trellis State-Complexity}~\cite{Vardy1998}.
In matroid theory, such a linear layout is called a \emph{path-decomposition} (of width at most $k$) of a matroid $M$ represented by vectors $v_1$, $v_2$, $\ldots$, $v_n$ and the minimum such $k$ is called the \emph{path-width} of the matroid $M$.

Kashyap~\cite{Navin2008} proved that it is NP-complete to decide whether such a linear layout exists when $k$ is a part of the input. To prove its hardness, he first observed that the trellis-width is actually equal to the path-width of a matroid represented by the same set of vectors. Path-width has been widely studied in structural graph theory, notably by Robertson and Seymour~\cite{RS1983}, and is also investigated in the context of matroid theory, for instance, in \cite{HOS2007}.
This link allowed him to deduce the NP-completeness of trellis-width from the NP-completeness of the path-width of graphs.

When $k$ is a fixed constant, then the problem is solvable in polynomial time, thanks to a simple and general algorithm by Nagamochi~\cite{Nagamochi2012}. His algorithm can decide in time $O(n^{ck})$ whether such a linear layout exists, even if the input matroid is given by an independence oracle.

\subsection*{Our contribution}
We present a \emph{fixed-parameter tractable} algorithm to decide whether such a linear layout exists when $k$ is a parameter.\footnote{In 2012, van Bevern, Downey, Fellows, Gaspers, and Rosamond~\cite{BDFGR2012} 
posted a manuscript on arXiv claiming such a fixed-parameter tractable algorithm.
However, their reduction to the problem of \emph{hypergraph cutwidth},
which is proven to be fixed-parameter tractable, had a subtle bug.
So their journal paper does not mention trellis-width~\cite{BDFGR2015}.
}
In other words, we have an algorithm that runs in time $f(k)n^{c}$ for some function $f$ and a constant $c$ when $n$ is the number of elements. 
Our algorithm does not only decide whether such a linear layout exists, but also output such a linear layout if one exists.

More precisely we will prove the following theorem.

\begin{THM}
  Let $\F$ be a fixed finite field.   There exists, for some function $f$, an $O(f(k) n^3)$-time algorithm  that, for input vectors $v_1,v_2,\ldots,v_n$ in $\F^r$ for some $r\le n$ and a parameter $k$, either  finds a linear layout $v_{\sigma(1)},v_{\sigma(2)},\ldots,v_{\sigma(n)}$ for a permutation $\sigma$ on $\{1,2,\ldots,n\}$ 
such that 
\[\dim \spn{ v_{\sigma(1)},v_{\sigma(2)},\ldots,v_{\sigma(i)}} \cap
 \spn{ v_{\sigma(i+1)},v_{\sigma(i+2)},\ldots,v_{\sigma(n)}} \le k\]
for all $i=1,2,\ldots,n-1$ or confirms that no such linear layout exists.
\end{THM}
Equivalently we have the following.
\begin{THMMAIN}%
  Let $\F$ be a fixed finite field.
  There is an algorithm that, for an input $n$-element matroid given by its matrix representation over $\F$ having at most $n$ rows and  a parameter $k$, decides in time
  $O(f(k)n^3)$ for some function $f$ whether its path-width is at most $k$ 
and if so, outputs a path-decomposition of width at most $k$.
\end{THMMAIN}

The existence of a fixed-parameter tractable algorithm for the decision version of the problem is not new, but no prior algorithm was able to construct a linear layout of width at most $k$ even if one exists. For the existence of a decision algorithm we use branch-width. Branch-width of a graph or a matroid is a width parameter, introduced by Robertson and Seymour~\cite{RS1991},  that measures how easy it is to decompose a graph or a matroid into a tree-like structure by cutting through a separation of small order. 
In particular, the path-width of a matroid is always greater than or equal to its branch-width.
Geelen, Gerards, and Whittle~\cite{GGW2002,GGW2002cor} proved that for a fixed finite field $\F$, $\F$-representable matroids of bounded branch-width are well-quasi-ordered by the matroid  minor relation. (Recently they announced a stronger theorem that does not require bounded branch-width~\cite{GGW2014}.)
As a corollary, we deduce that there are finitely many $\F$-representable minor obstructions for the class of  matroids of path-width at most $k$, because all those minor obstructions have branch-width at most $k+1$.
\Hlineny~\cite{Hlineny2004} proved that for a fixed matroid $N$, it is possible to decide  in time $f(k)n^3$ whether an input $n$-element matroid given by a matrix representation over $\F$ having branch-width at most $k$ has  a minor isomorphic to $N$. As path-width is always greater than or equal to branch-width, by checking the existence of obstructions as a minor, 
we can decide whether the path-width of an input matroid given by a matrix representation over $\F$ 
is at most $k$. But this approach has the following weaknesses.
\begin{itemize}
\item There are a lot of minor obstructions. Koutsonas, Thilikos, and Yamazaki~\cite{KTY2011} proved that there are at least $(k!)^2$ minor obstructions for matroid path-width at most $k$ in any field $\F$.
Furthermore there is no algorithm known to generate all $\F$-representable minor obstructions.%
\footnote{Kant\'e  and Kwon have claimed that they proved an upper bound of the size of each minor obstruction as a function of $k$ and~$\abs{\F}$ in a manuscript posted on arXiv on December 2014, but later we were told by the authors that the proof is incorrect and not yet fixed [private communication, 2015].}

\item Even if we know the complete list of minor obstructions, it is still non-trivial to construct a linear layout of width at most $k$ when it exists.\footnote{\Hlineny{}~\cite{Hlineny2016} claimed that this is possible by the self-reduction technique.}
 This is because non-existence of forbidden minors in the input matroid does not provide any hint of how to construct a linear layout.
\end{itemize}

The current situation for matroid path-width is somewhat similar to the status for path-width of graphs 25 years ago. %
For a fixed constant $k$, the problem of deciding whether a graph has path-width at most $k$ has been studied by various researchers. In 1983, Ellis, Sudborough, and Turner~\cite{EST1983} presented an $O(n^{2k^2+4k+8})$-time algorithm. Robertson and Seymour~\cite{RS1995} proved the existence of a fixed-parameter tractable $O(f(k)n^2)$-time algorithm based on the finiteness of minor obstructions, thus only solving the decision problem. Fellows and Langston~\cite{FL1994} developed a \emph{self-reduction} technique to convert such a decision algorithm into a construction algorithm and therefore a path-decomposition of a graph witnessing small path-width can be found in time $f(k)n^c$. However $f(k)$ depends on the number of minor obstructions, which is only known to be finite.
For matroid path-width, we do not know any self-reduction algorithm that converts a decision algorithm into a construction algorithm.%
\footnote{For matroid branch-width, \Hlineny{} and Oum~\cite{HO2006} devised a self-reduction algorithm to convert a decision algorithm to a construction algorithm, but this was based on a lemma on \emph{titanic} sets, originated by Robertson and Seymour~\cite{RS1991}, which works well with branch-width, not path-width.}
Bodlaender and Kloks~\cite{BK1996} proved the first constructive algorithm for path-width and tree-width of graphs based on dynamic programming, developing important techniques to be used in other papers \cite{BFT2009,BT2004,MZ2009,Thilikos2000,TSB2005,TSB2005a}.

\subsection*{Extension to subspace arrangements}

In fact, we present an algorithm for a more general problem as described below. 

\begin{description}
\item [Input] $n$ subspaces $V_1,V_2,\ldots,V_n$ of a vector space $\F^r$ over a fixed finite field $\F$, an integer $k$.
\item [Parameter] $k$.
\item [Problem] Determine whether there exists a permutation $\sigma$ of $\{1,2,\ldots,n\}$ such that 
\[\dim \spn{ V_{\sigma(1)},V_{\sigma(2)},\ldots,V_{\sigma(i)}} \cap
 \spn{ V_{\sigma(i+1)},V_{\sigma(i+2)},\ldots,V_{\sigma(n)}} \le k\]
for all $1\le i\le n-1$
and find a linear layout $V_{\sigma(1)},V_{\sigma(2)},\ldots,V_{\sigma(n)}$ if one exists.
\end{description}

A set of subspaces is called a \emph{subspace arrangement} and let us call minimum such $k$ the \emph{path-width} of a subspace arrangement. %
Clearly if we restrict input subspaces to have dimension at most $1$, then this problem is equivalent to the problem on vectors.
The following theorem is our main result.

\begin{THMSUBSP}
  Let $\F$ be a fixed finite field.
  Given, as an input, $n$ subspaces of $\F^r$ for some $r$ and a parameter $k$, in time $O(f(k)n^3+rm^2)$ for some function $f$, 
we can either find a linear layout $V_1,V_2,\ldots,V_n$ of the subspaces such that
\[
\dim (V_1+V_2+\cdots+V_i)\cap (V_{i+1}+V_{i+2}+\cdots+V_n)\le k
\]
for all $i=1,2,\ldots,n-1$, or confirm that no such linear layout exists,
where each $V_i$ is given by its spanning set of $d_i$ vectors and $m=\sum_{i=1}^n d_i$.
\end{THMSUBSP}

\subsection*{Our contribution to rank-width of graphs}
The linear rank-width of a graph is a linearized variant of the rank-width introduced by Oum and Seymour~\cite{OS2004}. The linear rank-width of an $n$-vertex graph is the smallest $k$ such that the vertices can be arranged in a linear order $v_1$, $v_2$, $\ldots$, $v_n$ so that the rank of the $i\times (n-i)$ $0$-$1$ matrix (over the binary field) representing the adjacency between $\{v_1,v_2,\ldots,v_i\}$ and $\{v_{i+1},v_{i+2},\ldots,v_{n}\}$ is at most $k$ for all $i=1,2,\ldots,n-1$. Deciding whether the linear rank-width of a graph is at most $k$ is NP-complete when $k$ is a part of the input; Proposition 3.1 in \cite{Oum2004} shows that the problem of determining the linear rank-width of a bipartite graph is equivalent to the problem of determining the path-width of a binary matroid,  which is shown to be NP-complete by Kashyap~\cite{Navin2008}. 
However, if $k$ is a fixed parameter, then the following two theorems imply that there exists an $O(f(k)n^3)$-time algorithm to decide whether the linear rank-width is at most~$k$:
\begin{enumerate}[(1)]
\item  Graphs of bounded rank-width are well-quasi-ordered by the vertex-minor relation, shown by Oum~\cite{Oum2004a}%
  \footnote{This theorem is a generalization of the well-quasi-ordering theorem of Geelen, Gerards, and Whittle~\cite{GGW2002} for $\F$-representable matroids of bounded branch-width where $\F$ is a fixed finite field.}.
  Suppose that $F_k$ is a minimal set of graphs such that a graph $G$ has
  linear rank-width  at most $k$ if and only if no graph in $F_k$ is
  isomorphic to a vertex-minor of $G$. Then it can be easily seen that
  graphs in $F_k$ have linear rank-width exactly $k+1$, implying that
  their rank-width are at most $k+1$.
  By the well-quasi-ordering theorem of graphs of bounded rank-width,
  we deduce that $F_k$ is finite.
\item If $F$ is a fixed graph, then one can decide whether the input $n$-vertex graph has a vertex-minor isomorphic to $F$ or confirm that the rank-width is larger than $k$ in time $O(g(k)n^3)$ implied by Courcelle and Oum~\cite{CO2004}.
\end{enumerate}
Jeong, Kwon, and Oum~\cite{JKO2014} showed that $\abs{F_k}\ge 2^{\Omega(3^k)}$ but there is no known upper bound on the size of graphs in $F_k$.
Furthermore, there is no known algorithm to generate $F_k$ and even if we know the list $F_k$, it still does not produce a linear rank-decomposition.

The following theorem presents the first fixed-parameter tractable algorithm that provides a linear rank-decomposition of width at most $k$ for graphs of linear rank-width at most $k$.
It will be proved by reducing this problem to the problem on path-width for a subspace arrangement consisting of  $2$-dimensional subspaces arising from a graph, discussed in Section~\ref{sec:linearrankwidth}.

\begin{THMRWD}
For an input $n$-vertex graph and a parameter $k$,
we can 
decide in time $O(f(k)n^3)$ for some function $f$ whether its linear rank-width is at most $k$
and if so, find a linear rank-decomposition of width at most~$k$.
\end{THMRWD}

We would like to mention that the above theorem has an application to linear clique-width. The linear clique-width of a graph is a variant of clique-width introduced by Courcelle and Olariu~\cite{CO2000}, which was the main motivation to introduce rank-width of graphs. In Subsection~\ref{subsec:linearcwd}, we will present definitions of the linear clique-width and the linear $k$-expression, which is a certificate showing a graph has linear clique-width at most $k$. 
Fellows, Rosamond, Rotics, and Szeider~\cite{FRRS2009} showed that it is NP-complete to decide whether a graph has linear clique-width at most $k$ 
when $k$ is a part of the input. 
However, when $k$ is a fixed parameter, then by Theorem~\ref{thm:lrw}, we obtain the following approximation algorithm for linear clique-width, which was not known previously.
\begin{CORCWD}
For an input $n$-vertex graph $G$ and a parameter $k$, 
we can find a linear $(2^k+1)$-expression of $G$ confirming that $G$ has linear clique-width at most $2^k+1$
or certify that $G$ has linear clique-width larger than $k$ in time $O(f(k)n^3)$ for some function $f$.
\end{CORCWD}

\subsection*{Computing path-width exactly when branch-width is bounded}

We also present a polynomial-time algorithm that can
\begin{itemize}
\item compute the path-width of an $\F$-represented matroid of bounded branch-width,
\item compute the linear rank-width of a graph of bounded rank-width.
\end{itemize}
This is analogous to the polynomial-time algorithm to compute path-width of a graph of bounded tree-width by Bodlaender and Kloks~\cite{BK1996} and 
generalizes polynomial-time algorithms to compute the path-width of the cycle matroid of outerplanar graphs by Koutsonas, Thilikos, and Yamazaki~\cite{KTY2011} 
and compute the linear rank-width of graphs of rank-width $1$ by Adler, Kant\'e, and Kwon~\cite{AMK2016}.

\subsection*{Proof overview}

We assume that we are given a branch-decomposition of small width and then use dynamic programming to solve the problem based on the tree-like structure given by the branch-decomposition. Courcelle's Theorem~\cite{COU1990} reveals what typically makes a graph problem amenable for dynamic programming on tree-like structure. In his seminal work, Courcelle proved that every graph property that can be expressed as a monodic second-order formula $\phi$ admits a finite number of {\sl states} in the following sense: when two graphs $G$ and $H$ are {\sl glued} along $w$ common vertices, 
we can decide 
whether the glued graph satisfies the formula or not
by looking at their respective {\sl behavior} on the $w$ vertices with respect to the formula. As the number of possible behaviors over all graphs, or {\sl states}, on $w$ vertices with respect to $\phi$ is bounded by a fixed function $f(w,\abs{\phi})$, remembering a small set of states available for $G$ on $w$ vertices is sufficient to decide whether $G$ glued with $H$ shall satisfy $\phi$ or not. Especially when an input graph can be recursively decomposed along vertex separators of size at most $w$, a fixed graph property $\phi$ can be tested on the input in linear time. Courcelle's  Theorem also  applies to graph optimization problems. While Courcelle's Theorem provides a powerful meta-algorithm, the downside is that the function $f$ that plays as a hidden constant in $O(n)$ is gigantic in general, or the algorithm might not be constructible. The novelty of non-trivial dynamic programming algorithm lies in devising an encoding scheme of a behavior of a graph (with a partial solution) on a small separator. %

Our proof strategy is similar  to the one employed by Bodlaender and Kloks~\cite{BK1996}. The algorithm of Bodlaender and Kloks is based on dynamic programming on tree-decomposition of small width. %
We also use the idea of \emph{typical sequences} introduced by Bodlaender and Kloks~\cite{BK1996} to encode and compress the amount of information needed to track partial solutions.%
\footnote{A similar attempt to adapt the method of Bodlaender and Kloks to a wide class of  width parameters was made by Berthom\'e et al.~\cite{BBMNS2013} (see also the Ph.D. thesis of Soares~\cite[Theorem 11]{Soares2013}). One may ask whether their algorithm implies our theorems.
 However, their theorem requires a ``nice decomposition'' of small width ``compatible'' with the partition function corresponding to branch-width of matroids or subspace arrangements
to be given as an input in order to run dynamic programming. Since their definition of ``compatible nice decompositions'' is very strong, the existence of a compatible nice decomposition is not guaranteed, even if we have bounded path-width for matroids or subspace arrangements. Thus it is unlikely that their theorem contains ours.
}
One of our key contribution is to come up with an encoding of partial solutions with respect to a fixed subspace. In a nutshell, we propose an encoding scheme of a linear layout of a subspace arrangement $\V$ with respect to a subspace $B$ as follows: for every linear layout $\sigma$ of a subspace arrangement ${\mathcal W}$ with $\operatorname{span}(\V)\cap \operatorname{span}({\mathcal W}) \subseteq B$,
\begin{quote}
if two linear layouts $\sigma_1$ and $\sigma_2$ of $\V$ have the same encoding $\Gamma$ with respect to $B$, then the width of a layout $\pi_1$ of $\V \dot\cup \mathcal W$ extending $\sigma_1$ and $\sigma$ is equal to the width of a layout $\pi_2$ of $\V \dot\cup \mathcal W$ extending $\sigma_2$ and $\sigma$, as long as $\pi_1$ and $\pi_2$ {\sl mix} $\sigma_i$ and $\sigma$  the same way for both $i=1,2$ in relation to $\Gamma$.
\end{quote}
This property\footnote{We do not formally prove this property, but it can be easily deduced from the definition of $B$-trajectory and Lemma~\ref{lem:dimension}.} allows us to safely forget the specific linear layout of $\V$ that gives rise to an encoding $\Gamma$, and remember $\Gamma$ on $B$ only. 
Also, due to this property, it is enough to consider the possibilities of mixing two encodings instead of considering all possible ways of mixing two linear layouts: the set of achievable widths remains the same in both cases. What we mean by {\sl mixing} shall become clear in Subsection~\ref{subs:sum} with the notion of \emph{lattice path}.
Unlike usual encodings of partial solutions on graphs, our encoding needs to handle subspaces and keep track of how subspaces interact. 
We develop an encoding scheme, operations and relations on encodings to this aim. 
We also present a framework to perform dynamic programming on branch-decompositions of a subspace arrangement.%
\footnote{\Hlineny~\cite{Hlineny2004} presented such a framework for represented matroids.}

The paper is organized as follows. Section~\ref{sec:prelim} will review basic concepts and a framework to perform dynamic programming on branch-decompositions.
Section~\ref{sec:traj} will introduce $B$-trajectories, which will be a key concept for our algorithm.
Section~\ref{sec:fullset} will define full sets consisting of $B$-trajectories and show several propositions on combining two full sets coming from child nodes of a node to obtain the full set. 
We will present our algorithm in Section~\ref{sec:algorithm} based on the operations discussed in Section~\ref{sec:fullset}.
Since our algorithm is based on dynamic programming, up to Section~\ref{sec:algorithm} we assume that some branch-decomposition is given. Section~\ref{sec:approx} will discuss two methods of providing this branch-decomposition.
Section~\ref{sec:matroid} presents a corollary to path-width of matroids
and Section \ref{sec:linearrankwidth} will discuss corollaries to linear rank-width and linear clique-width of graphs.
Section~\ref{sec:exact} will present an exact polynomial-time algorithm to compute path-width of a  matroid when the input has small branch-width and its corollary to linear rank-width of graphs.
Section~\ref{sec:disc} concludes this paper by presenting some open problems on matroids.

\section{Preliminaries}\label{sec:prelim}

We assume readers are familiar with linear algebra.
For a set $X$ of vectors, we write $\spn{X}$ to denote the \emph{span} of $X$, which is the set of all vectors that are linear combinations of vectors in $X$.
We say that $X$ is a \emph{spanning set} of a vector space $S$ if $\spn{X}=S$.
For simplicity, we also write $\spn{v_1,v_2,\ldots,v_n}=\spn{\{v_1,v_2,\ldots,v_n\}}$ for a sequence of vectors $v_1$, $v_2$, $\ldots$, $v_n$.

A \emph{subspace arrangement} $\V$ over a field $\F$ is a finite set of subspaces $V_1$, $V_2$, $\ldots$, $V_n$ of a finite-dimensional vector space over $\F$. 
Formally a subspace arrangement is %
an indexed family
$\V=\{V_i\}_{i\in E}$ of subspaces for a finite set $E$. %
For two subspaces $A$, $B$, we write $A+B=\{a+b: a\in A, b\in B\}$.
For a subspace arrangement $\V$, let us write $\spn{\V}$ to denote $\sum_{V\in \V} V$.

\subsection{Path-width and branch-width of a subspace arrangement}

Let $\V$ be a subspace arrangement over a field $\F$.
A \emph{linear layout} of $\V$ is a permutation $\sigma=V_1,V_2,\ldots,V_n$ of $\V$. The width of a linear layout $\sigma=V_1,V_2,\ldots,V_n$ is \[ \max_{1\le i<n} \dim (V_1+V_2+\cdots+V_i)\cap (V_{i+1}+V_{i+2}+\cdots+V_n).\]
The \emph{path-width} of $\V$ is the minimum width of all possible linear layouts of $\V$. (If $\abs{\V}\le 1$, then the width of its trivial linear layout of $\V$ is $0$ and the path-width is $0$.)

We now define the branch-width of $\V$.
We say a tree is \emph{subcubic} if all its internal vertices have degree $3$.
A \emph{branch-decomposition} of $\V$ is a pair $(T,\L)$ of a subcubic tree $T$ and a bijection $\L$ from the set of leaves of $T$ to $\V$. 
For a set $X$ of leaves of $T$, let $\L(X)=\{\L(x): ~x\in X\}$.
Each edge $e$ induces a partition $(A_e,B_e)$ of the leaves given by $T-e$ and the \emph{width} of $e$ is defined to be \[\dim \spn{\L(A_e)} \cap \spn{\L(B_e)}.\]
The \emph{width} of a branch-decomposition $(T,\L)$ is the maximum width of all edges of $T$. The \emph{branch-width} of $\V$ is the minimum width of all possible branch-decompositions of $\V$. (If $\abs{\V}\le 1$, then there is no branch-decomposition and we define the branch-width of $\V$ to be $0$.)
\begin{PROP}\label{prop:submodular}
  Let $\V$ be a subspace arrangement. For a subset $X$ of $\V$, let $f(X)=\dim \spn{X}\cap \spn{\V-X}$. Then $f$ satisfies the following.
  \begin{enumerate}[(i)]
  \item $f(X)=f(\V-X)$ for all $X\subseteq \V$. (\emph{symmetric})
\item
  $ f(X)+f(Y)\ge f(X\cap Y)+f(X\cup Y)$ for all $X,Y\subseteq \V$. (\emph{submodular})
  \end{enumerate}
\end{PROP}
\begin{proof}
  (i) is trivial. Let us prove (ii). It is enough to prove the following.
  \begin{multline*}
    \dim \spn{X}+\dim\spn{\V-X}+
    \dim \spn {Y}+\dim\spn{\V-Y}\\
    \ge \dim \spn{X\cap Y}+\dim\spn{\V-(X\cap Y)}
    + \dim \spn{X\cup Y}+\dim\spn{\V-(X\cup Y)}.
  \end{multline*}
  This follows from the facts that $\dim (\spn{A}+\spn{B})=\dim \spn{A\cup B}$,
  $\dim (\spn{A}\cap \spn{B})\ge \dim \spn{A\cap B}$,
  and $\dim \spn{A}+\dim \spn{B} =\dim (\spn{A}+\spn{B})+\dim \spn{A}\cap \spn{B}$
 for all $A,B\subseteq \V$.
\end{proof}
\begin{PROP}\label{prop:pwbw}
  If  $\V$ is a subspace arrangement of path-width at most $k$,
  then \[\dim X\cap \spn{\V-\{X\}} \le 2k\] for all $X\in \V$ and the branch-width of $\V$ is at most $2k$.
\end{PROP}
\begin{proof}
  Let $V_1, V_2,\ldots, V_n$ be a linear layout  of $\V$ of width at most $k$.
  Clearly $\dim V_1\cap (\V-\{V_1\})\le k$ and $\dim V_n\cap (\V-\{V_n\})\le k$.
  For  $i\in\{2,3,\ldots,n-1\}$,
  \begin{align*}
    \dim (V_1+\cdots+V_i)\cap (V_{i+1}+\cdots+V_{n})\le k,\\
    \dim (V_i+V_{i+1}+\cdots+V_n)\cap (V_{1}+\cdots+V_{i-1})\le k.
  \end{align*}
  By Proposition~\ref{prop:submodular}, we deduce that
  $\dim  {V_i}\cap \spn{\V-\{V_i\}}\le 2k$ because $\{V_1,V_2,\ldots,V_i\}\cap \{V_i,V_{i+1},\ldots,V_{n}\}=\{V_i\}$.
  This proves that $\dim X\cap \spn{\V-\{X\}}\le 2k$ for all $X\in \V$.

  Let $T$ be a tree obtained from the path graph $P$ on $n$ vertices by attaching one leaf to each internal node of the path. Note that $T$ has $n$ leaves.
  Let $\L$ be a function that maps all leaves of $T$ to $V_1$, $V_2$, $\ldots$, $V_n$ bijectively in the order of the path.
  Then $(T,\L)$ is a branch-decomposition of $\V$. 
  For every edge $e$ in $P$, the width of $e$ is clearly at most $k$ by the assumption and so $(T,\L)$ has width at most $2k$.
\end{proof}

\subsection{Sequence of integers}
Bodlaender and Kloks~\cite{BK1996} introduced typical sequences for their works on path-width and tree-width.
For an integer $n>0$, let $\alpha=(a_1,a_2,\ldots,a_n)$ be a finite sequence of integers of length $n$.

For a sequence $\alpha$, we define its \emph{typical sequence} $\tau(\alpha)$ to be the sequence obtained by iterating the following \emph{compression} operations until no further operations can be applied:
\begin{itemize}
\item If two consecutive entries are equal, then remove one.
\item If there exist $i$ and $j$ such that $j-i>1$  and
$a_i\le a_k\le a_j$ for all $k\in \{i+1,i+2,\ldots,j-1\}$
or
$a_i\ge a_k\ge a_j$ for all $k\in \{i+1,i+2,\ldots,j-1\}$,
then remove all entries from the $(i+1)$-th entry to the $(j-1)$-th entry.
\end{itemize}
A sequence $\alpha$ is called \emph{typical} if $\tau(\alpha)=\alpha$.
For example, if $\alpha=(1,3,2,5,2,2,4,4,3)$, then $\tau(\alpha)=(1,5,2,4,3)$.

Bodlaender and Kloks~\cite{BK1996} not only showed that $\tau(\alpha)$ is uniquely defined but also proved
that the length of typical sequences of bounded integers and
the number of distinct typical sequences of bounded integers
are bounded as follows. These lemmas are fundamental to our applications.

\begin{LEM}[Bodlaender and Kloks~{\cite[Lemma 3.3]{BK1996}}]\label{lem:lentypical}
 The length of a typical sequence of integers consisting of $\{0,1,\ldots,k\}$
 is at most $2k+1$.
\end{LEM}

The following lemma was originally proved for the number of typical sequences only. In fact, the proof is constructive and it is not difficult to see that all sequences can be generated with polynomial delay. Here, we will base our algorithm on a weaker result.
\begin{LEM}[Bodlaender and Kloks~{\cite[Lemma 3.5]{BK1996}}]\label{lem:typical}
  There are at most $\frac83 2^{2k}$ distinct typical sequences of integers consisting of  $\{0,1,\ldots,k\}$. Furthermore, the set of all typical sequences can be generated in $\poly(k)\cdot 2^{2k}$ steps.
\end{LEM}

\subsection{Manipulating subspaces}\label{subsec:linear}
We present the data structure to store subspaces and review algorithms performing elementary operations on subspaces.

Let $\B$ be an ordered basis of a subspace $S\subseteq \F^r$ of dimension $d$. Let $M_S\in \F^{r\times d}$ be an $r\times d$ matrix whose column vectors are vectors in $\B$.

In our application, $d$ will be usually a lot smaller than $r$ and so we will use the coordinate system with respect to $\B$ in order to represent a vector in $S$. In other words, for each vector $w\in S\subseteq \F^r$, we write $[w]_{\B}$ to denote the coordinate vector with respect to $\B$ as  a column vector.
It is easy to obtain the coordinates $w\in \F^r$ by multiplying $M_S$ with $[w]_{\B}$.

By using this coordinate system, a $d'$-dimensional subspace $S'$ of $S$ can be represented by a $d\times d'$ matrix $M_{S,S'}\in \F^{d\times d'}$ in such a way that the range of $M_S M_{S,S'}$ is exactly $S'$. If so, then we say that $M_{S,S'}$ \emph{represents} $S'$.

It is well known that one can solve a system of linear equations in time $O(n^3)$
by Gauss-Jordan elimination
when the number of variables and the number of equations are at most $n$.
We summarize the time complexity of computing a matrix representing a subspace as follows.

\begin{LEM}\label{lem:subspace1}
Let $S$ be a $d$-dimensional subspace of $\F^r$ and $S'$ be a subspace of $S$.
Let $\B$ be an ordered basis of $S$.
Given a matrix $M_{S,S'}$ representing $S'$ %
and a vector $x$ in $S$, when $x$ is given as a coordinate vector $[x]_{\B}$ in $\F^d$ with respect to $\B$,
we can decide $x\in S'$ in time $O(d^3)$. 
\end{LEM}
\begin{proof}
We can decide $x\in S'$ by 
checking whether 
$[x]_{\B}$ is a linear combination of columns in $M_{S,S'}$.
\end{proof}

\begin{LEM}\label{lem:subspace2}
Let $S$ be a $d$-dimensional subspace of $\F^r$.
Given, as an input, 
matrices $M_{S,S_i}$ representing subspaces $S_i$ of $S$ for $i=1,2$, 
we can 
\begin{enumerate}
\item
   decide whether $S_1\subseteq S_2$ in time $O(d^3)$,

\item    compute a matrix $M_{S,S'}$ representing $S'=S_1+S_2$ in time $O(d^3)$, and 

\item  compute a matrix $M_{S,S''}$ representing $S''=S_1\cap S_2$ in time $O(d^3)$. 
\end{enumerate}
\end{LEM}

\begin{proof}
(1) This can be answered by solving a matrix equation $M_{S,S_2}X=M_{S,S_1}$; 
this is feasible if and only if $S_1\subseteq S_2$. 
(2) This can be done by finding %
a column basis of $\left( M_{S,S_1} \quad M_{S,S_2}\right)$.
(3)
We find a matrix $\binom{M_x}{M_y}$ to represent the space of all vectors $\binom{x}{y}$ satisfying
$\left( M_{S,S_1} \quad -M_{S,S_2}\right) \binom{x}{y} =0$ where $x\in \F^{\dim S_1}$ and $y\in \F^{\dim S_2} $.
Then $M_{S,S''}$ is obtained from $M_{S,S_1}M_x$ by deleting linearly dependent column vectors.
\end{proof}

\subsection{Dynamic programming on branch-decompositions}\label{subsec:dyn}

For many discrete objects, tree-like structures allow us to view the given input structure as a composition of simple objects in a tree-like structure, making it very efficient to run algorithms based on dynamic programming. 
Tree-decompositions of graphs have been widely used as a decomposition tree for the dynamic programming algorithm.

For matroids, \Hlineny~\cite{Hlineny2004} introduced ``matroid parse trees'' for the purpose of dynamic programming on matroids of bounded branch-width, represented over a fixed finite field.
We are going to provide a similar yet different framework for the same purpose that is easier for us to handle.

Let $(T,\L)$ be a branch-decomposition of a subspace arrangement $\V$ having width at most $\theta$.
We may assume that $T$ is a rooted binary tree by picking an arbitrary edge $e$ and subdividing $e$ to create a degree-$2$ vertex called the \emph{root}.
By orienting each edge of $T$ towards the root, we may further assume that $T$ is a rooted binary directed tree.
For a node $v$ of $T$, let $\V_v$ be the set of all subspaces in $\V$ associated with $v$ and its descendants by $\L$. %
We define the \emph{boundary space} $B_v$ as $\spn{\V_v} \cap \spn{\V-\V_v}$.
From now on, when we call $(T,\L)$ a branch-decomposition, $T$ is always a rooted binary directed tree.

The following easy lemma is essential for our dynamic programming.
It is well known that $\dim(X)+\dim(Y)=\dim(X+Y)+\dim(X\cap Y)$
for finite-dimensional subspaces $X$, $Y$.
\begin{LEM}\label{lem:subspace}
  Let $v$ be an internal node of $T$ and $w_1$, $w_2$ be two children of $v$.
  Let $B$, $B_1$, $B_2$ be the boundary space of $v$, $w_1$, $w_2$, respectively.
  Then
  \begin{align*}
    \spn{\V_{w_1}}\cap B_{1} = B_{1}&= \spn{\V_{w_1}}\cap (B_{1}+B_{2}),\\
    \spn{\V_{w_2}}\cap B_{2} = B_{2}&= \spn{\V_{w_2}}\cap (B_{1}+B_{2}),\\
    B&\subseteq B_{1}+B_{2},\\
    (\spn{\V_{w_1}}+B_{1}+B_{2})\cap (\spn{\V_{w_2}}+B_{1}+B_{2})&=B_{1}+B_{2}.
  \end{align*}
\end{LEM}
\begin{proof}
  Let $M_w=\spn{V_w}$ for each node $w$ of $T$.
 Since $M_{w_1}=\spn{\V_{w_1}}$ and $B_1=\spn{\V_{w_1}}\cap\spn{\V-\V_{w_1}}$, trivially
 $M_{w_1}\cap B_1=B_1\subseteq M_{w_1}\cap (B_1+B_2)$.
 Note that $\V_{w_1}\subseteq \V-\V_{w_2}$ and $\V_{w_2}\subseteq \V-\V_{w_1}$. Then
 \begin{align*}
   \dim M_{w_1}\cap (B_1+B_2)
   &= \dim M_{w_1}+\dim (B_1+B_2)-\dim (M_{w_1}+B_2)\\
   &= \dim B_1-\dim B_1\cap B_2 +\dim M_{w_1}\cap B_2\\
   &= \dim B_1
\end{align*}
because  $M_{w_1}+B_1+B_2 = M_{w_1}+B_2$ and $M_{w_1}\cap B_2=M_{w_1}\cap M_{w_2}=B_1\cap B_2$.
Thus, $B_1= M_{w_1}\cap (B_1+B_2)$.

 Now we claim that $B\subseteq B_1+B_2$.
 Let $b\in B=\spn{\V_v}\cap\spn{\V-\V_v}$.
 Clearly $b\in\spn{\V_v}=\spn{\V_{w_1}}+\spn{\V_{w_2}}$.
 Let $x\in\spn{\V_{w_1}}$, $y\in\spn{\V_{w_2}}$ be vectors such that $x+y=b$.
 Since $b\in\spn{\V-\V_v}\subseteq\spn{\V-\V_{w_1}}$ and $y\in\spn{\V_{w_2}}\subseteq\spn{\V-\V_{w_1}}$,
 we deduce that  $x\in \spn{\V-\V_{w_1}}$. Similarly, $y\in \spn{\V-\V_{w_2}}$.
 Thus, $b=x+y$ is in $B_1+B_2$.

  Note that $(M_{w_1}+B_{1}+B_{2})\cap (M_{w_2}+B_{1}+B_{2})=(M_{w_1}+B_2)\cap (M_{w_2}+B_1)$ and $B_1\cap B_2=M_{w_1}\cap M_{w_2}$. Thus
  \begin{align*}
    \dim&(M_{w_1}+B_2)\cap (M_{w_2}+B_1)\\
        &=\dim(M_{w_1}+B_2)+\dim(M_{w_2}+B_1)-\dim(M_{w_1}+B_1+M_{w_2}+B_2)\\
        &=\dim(M_{w_1})+\dim(B_2)-\dim(M_{w_1}\cap B_2)\\
        &\quad+\dim(M_{w_2})+\dim(B_1)-\dim(M_{w_2}\cap B_1)-\dim(M_{w_1}+M_{w_2})
    \displaybreak[3]\\
        &=\dim(M_{w_1})+\dim(B_2)-\dim M_{w_1}\cap M_{w_2}\cap \spn{\V-\V_{w_2}}\\
        &\quad+\dim(M_{w_2})+\dim(B_1)-\dim M_{w_2}\cap M_{w_1}\cap \spn{\V-\V_{w_1}}\\
    &\quad-\dim(M_{w_1}+M_{w_2}) \displaybreak[3]\\
    &= \dim(M_{w_1})+\dim(M_{w_2})+\dim(B_1)+\dim(B_2)\\&\quad-2\dim(M_{w_1}\cap M_{w_2})
      - \dim(M_{w_1}+M_{w_2}) \displaybreak[3]\\
    &= \dim(B_1)+\dim(B_2)-\dim(M_{w_1}\cap M_{w_2})\\
    &= \dim(B_1+ B_2).
  \end{align*}
  As $B_1+B_2\subseteq (M_{w_1}+B_1+B_2)\cap (M_{w_2}+B_1+B_2)$, we deduce that  $(M_{w_1}+B_{1}+B_{2})\cap (M_{w_2}+B_{1}+B_{2})=B_{1}+B_{2}$.
\end{proof}

\section{Trajectories}\label{sec:traj}
Given a branch-decomposition $(T,\L)$ of a subspace arrangement $\V$ over $\F$ whose width is at most $\theta$, let us assume that the boundary space $B_v$ has been computed for every node $v$ of $T$. Consider $\V_v$ at a node $v$ of $T$.
To design an algorithm based on dynamic programming, we need to keep a bounded (in terms of $\theta$, $k$, and $\abs{\F}$) amount of information %
that encodes all ``good'' linear layouts of $\V_v$.
Suppose that $\sigma=V_1,V_2, \ldots,V_m$ is a linear layout of $\V_v$. To be a ``good'' partial solution, we need $\dim (V_1+\cdots+V_{i-1})\cap (V_{i}+\cdots+V_m)\le k$ for all $2\le i\le m$.
What information do we need to keep in order to determine whether this can be extended 
to a linear layout of $\V$ of width at most $k$?
Lemmas~\ref{lem:join-key} and \ref{lem:dimension} will ensure that we only need to keep \begin{align*}L_i&:=(V_1+\cdots+V_{i-1})\cap B_v,\\ R_i&:=(V_{i}+\cdots+V_m)\cap B_v\end{align*}
and the extra dimension of the intersection not shown in $B_v$, that is \[\lambda_i:=\dim (V_1+\cdots+V_{i-1})\cap (V_{i}+\cdots+V_m)-\dim (V_1+\cdots+V_{i-1})\cap (V_{i}+\cdots+V_m)\cap B_v\] for each $2 \le i \le m$.
This motivates the definition of a \emph{statistic} that is a triple $(L,R,\lambda)$ of two subspaces $L$, $R$ of bounded dimension and one nonnegative integer $\lambda$.
For instance, each $i\in\{1,2,\ldots,m+1\}$ induces a statistic \[a_i:=(L_i,R_i,\lambda_i)\]
where we define $L_1=\{0\}=R_{m+1}$, $R_1=(V_1+\cdots+V_m)\cap B_v=L_{m+1}$, and $\lambda_1=\lambda_{m+1}=0$.
Then the linear layout $\sigma=V_1,V_2,\ldots,V_m$ induces a sequence of statistics 
\[\Gamma=a_1,a_2,\ldots,a_{m+1},\]
which we call the \emph{canonical} $B$-trajectory of $\sigma$.

While the number of $B$-trajectories can be arbitrarily large, it can be shown that the number of \emph{compact} ones is bounded in terms of $\theta$, $k$, and $\abs{\F}$. Hence, at each node $v$ of $T$, we carry over \emph{only} compact $B_v$-trajectories for dynamic programming.

\subsection{$B$-trajectories}
For a vector space $B$, a \emph{statistic} is a triple  $a=(L,R,\lambda)$ of two subspaces $L$, $R$ of $B$ and a nonnegative integer $\lambda$. We write $L(a)=L$, $R(a)=R$, and $\lambda(a)=\lambda$. For two statistics $a$ and $b$, we write $a\le b$ if
\[
L(a)=L(b), \qquad R(a)=R(b), \qquad \lambda(a)\le\lambda(b).
\]

A \emph{$B$-trajectory} of length $n\ge 1$ is a sequence $\Gamma=a_1,a_2,\ldots,a_n$ of statistics 
such that 
$R(a_1)=L(a_{n})$, $L(a_i)\subseteq L(a_{i+1})\subseteq B$ and $B\supseteq R(a_i)\supseteq R(a_{i+1})$ for all $i=1,2,\ldots,n-1$.
The \emph{width} of $\Gamma$ is $\max_{1\le i \le n}\lambda(a_i)$. %
We write $\abs{\Gamma}$ to denote the length $n$ of $\Gamma$.
We write $\Gamma(i)$ to denote the $i$-th statistic $(L(a_i),R(a_i),\lambda(a_i))$.
A $B$-trajectory $\Gamma^*$ is an \emph{extension} of a $B$-trajectory $\Gamma$ if $\Gamma^*$ is obtained from $\Gamma$ by repeating some of its entries.
The set of all extensions of $\Gamma$ is denoted as $E(\Gamma)$.

The \emph{compactification} $\tau(\Gamma)$ of $\Gamma=a_1,a_2,\ldots,a_n$ is a $B$-trajectory obtained from $\Gamma$ by applying the following operations until no further operation can be applied.
\begin{itemize}
\item Remove an entry $a_i$  if $a_{i-1}=a_i$.
\item If there exist $i$ and $j$ such that $j-i>1$, $L(a_i)=L(a_j)$, $R(a_i)=R(a_j)$ and
$\lambda(a_i)\le \lambda(a_k)\le \lambda(a_j)$ for all $k\in \{i+1,i+2,\ldots,j-1\}$ or $\lambda(a_i)\ge \lambda(a_k)\ge \lambda(a_j)$ for all $k\in \{i+1,i+2,\ldots,j-1\}$,
then remove $a_{i+1},a_{i+2},\ldots,a_{j-1}$ from the sequence.
\end{itemize}
It is clear that the width of $\tau(\Gamma)$ and $\Gamma$ are equal.
We say that $\Gamma$ is \emph{compact} if $\tau(\Gamma)=\Gamma$.

As the notion of $B$-trajectories encapsulates sequences of integers, 
the notion of compact $B$-trajectories captures typical sequences of integers.
For example, when $B=\{0\}$, a $B$-trajectory $\Gamma$ is of the form 
$(\{0\},\{0\},\lambda_1)$, $(\{0\},\{0\},\lambda_2)$, $\ldots$, $(\{0\},\{0\},\lambda_n)$.
Then $\Gamma$ is compact if and only if the sequence $\lambda_1$, $\lambda_2$, $\ldots$, $\lambda_n$ is typical.
Due to~\cite{BK1996}, 
we can show that 
if $B$ is a vector space over $\F$ of dimension $\theta$,
then 
the length of compact $B$-trajectories of width at most $k$
and the number of distinct compact $B$-trajectories of width at most $k$ are bounded from above as follows.

\begin{LEM}\label{lem:lenLR}
Let $B$ be a vector space over $\F$ of dimension $\theta$. If $\Gamma=a_1,a_2,\ldots,a_m$ is a compact $B$-trajectory, then the number of distinct pairs $(L(a_i),R(a_i))$ is at most $2\theta+1$.
\end{LEM}
\begin{proof}
Let $(L_1,R_1), (L_2,R_2),\ldots,(L_p,R_p)$ be distinct pairs of subspaces of $B$ such that $L_1\subseteq L_2\subseteq\cdots\subseteq L_p$ and $R_1\supseteq R_2\supseteq \cdots\supseteq R_p$.
  Then $\dim(L_1)\le \cdots\le\dim(L_p)\le \theta$
  and $\theta\ge \dim(R_1)\ge \cdots\ge \dim(R_p)$.
  As they are all distinct, $\dim(L_{i+1})>\dim(L_i)$ or $\dim(R_{i+1})<\dim(R_i)$ for all $i$. It follows that
\[
  -\theta \le \dim(L_1)-\dim(R_1)< \dim(L_2)-\dim(R_2)
<\cdots< \dim(L_p)-\dim(R_p)\le \theta.
\]
  This implies that $p\le 2\theta+1$.
\end{proof}

\begin{LEM}\label{lem:lentraj}
Let $B$ be a vector space over $\F$ of dimension $\theta$.
Every compact $B$-trajectory of width $k$ has length at most $(2\theta+1)(2k+1)$.
\end{LEM}
\begin{proof}
Let $\Gamma=a_1,a_2,\ldots,a_m$ be a compact $B$-trajectory.
By Lemma~\ref{lem:lenLR}, there are at most $2\theta+1$ distinct pairs $(L(a_i),R(a_i))$ in $\Gamma$. Each maximal consecutive subsequence $\lambda(a_i),\lambda(a_{i+1}),\ldots ,\lambda(a_j)$ with $L(a_i)=L(a_j)$ and $R(a_i)=R(a_j)$ forms a typical sequence of integers, whose length is bounded by $2k+1$ by Lemma~\ref{lem:lentypical}. This proves the claim.
\end{proof}

For a vector space $B$ over $\F$, let $U_k(B)$ be the set of all compact $B$-trajectories of width at most $k$. 
The following lemma will find an upper bound of $\abs{U_k(B)}$.
Its proof uses lattice paths.
A \emph{lattice path} from $(1,1)$ to $(x,y)$ is a path $P=v_1,v_2,\ldots,v_n$ on the plane such that $v_1=(1,1)$, $v_n=(x,y)$, $v_i\in \mathbb Z\times\mathbb Z$ and $v_{i+1}-v_{i}\in \{(1,0),(0,1)\}$ for all $i=1,2,\ldots,n-1$.

\begin{LEM}\label{lem:numtra}
Let $B$ be a vector space over $\F$ of dimension $\theta$. Then $U_k(B)$ contains at most $2^{9\theta+2}\abs{\F}^{\theta(\theta-1)}2^{2(2\theta+1)k}$ elements. Furthermore, the set $U_k(B)$ can be generated in $\poly(\theta, \abs{\F}, k)\cdot \abs{U_k(B)}$ steps.
\end{LEM}
\begin{proof}

Let $q=\abs{\F}$. First let us assume that $\theta>0$.
The number of length-$(\theta+1)$ \emph{chains} $V_0\subset V_1\subset \cdots\subset V_\theta$ of subspaces of $B$ is precisely
\[
\binom{\theta}{\theta-1}_q \binom{\theta-1}{\theta-2}_q
\binom{\theta-2}{\theta-3}_q\cdots \binom{1}{0}_q
=\prod_{d=1}^{\theta-1}\binom{d+1}{d}_q,
\]
where $\binom{d+1}{d}_q=\frac{q^{d+1}-1}{q-1}$, known as the Gaussian binomial coefficients. Since $\binom{d+1}{d}_q=q^d+q^{d-1}+\cdots+1\le 2q^d$, we deduce that the number of chains of length $\theta+1$ is at most $2^{\theta-1} q^{\theta(\theta-1)/2}$.

Each $B$-trajectory $\Gamma=a_1,a_2,\ldots,a_m$
induces a set $\{(L(a_i),R(a_i)):i=1,2,\ldots,m\}$ and there exist length-$(\theta+1)$ chains $L_0\subset L_1\subset L_2\subset \cdots\subset L_{\theta}$, $R_0\supset R_1\supset R_2\supset \cdots\supset R_{\theta}$ of subspaces of $B$ such that $L(a_i)\in\{L_0,L_1,\ldots,L_\theta\}$
and $R(a_i)\in \{R_0,R_1,\ldots,R_\theta\}$.
Furthermore there exists a lattice path $(x_0,y_0)$, $(x_1,y_1)$,
$\ldots$, $(x_{2\theta},y_{2\theta})$  from $(0,0)$ to $(\theta,\theta)$ such that
$(L(a_i),R(a_i))\in \{(L_{x_i},R_{y_i}):i=0,1,\ldots,2\theta\}$.

By the previous observation, there are  at most $2^{\theta-1} q^{\theta(\theta-1)/2}$ ways of choosing each of $\{L_i\}$ and $\{R_i\}$
and there are $\binom{2\theta}{\theta}$ ways of selecting a lattice path $P=v_0,v_1,v_2,\ldots,v_{2\theta}$ from $v_0=(0,0)$ to $v_{2\theta}=(\theta,\theta)$ and let $(x_i,y_i)=v_i$.
Since $\{(L(a_i),R(a_i)):i=1,2,\ldots,m\}$ is a subset of $\{(L_{x_i},R_{y_i}):i=0,1,\ldots,2\theta\}$, there are at most
\[
(2^{\theta-1} q^{\theta(\theta-1)/2} )^2\binom{2\theta}{\theta}2^{2\theta+1}
\le 2^{6\theta-1} q^{\theta(\theta-1)}
\]
distinct sets $\{(L(a_i),R(a_i):i=1,2,\ldots,m\}$ induced by $B$-trajectories $\Gamma$. It is not difficult to see that all such sets $\{(L(a_i),R(a_i):i=1,2,\ldots,m\}$ can be generated in $\poly(k)\cdot 2^{6\theta-1} q^{\theta(\theta-1)}$ for some fixed polynomial function in $k$.

  Each maximal consecutive subsequence $a_{i},\ldots,a_{j}$ of a $B$-trajectory $\Gamma=a_1,\ldots,a_n$ with $L(a_i)=L(a_j)$ and $R(a_i)=R(a_j)$ induces a `typical sequence' $\lambda(a_i),\ldots,\lambda(a_j)$. 
There are at most $\frac83 2^{2k}$ distinct such typical sequences and 
they can be generated in $\poly(k)\cdot 2^{2k}$ steps by Lemma~\ref{lem:typical}.

  Therefore, the number of compact $B$-trajectories of width at most $k$ is at most
  \[2^{6\theta-1}q^{\theta(\theta-1)}\left(\frac{8}{3} 2^{2k}\right)^{2\theta+1}
  \le 2^{9\theta+1}q^{\theta(\theta-1)}2^{2(2\theta+1)k}. \]

If $\theta=0$, then the number of compact $B$-trajectories is precisely the number of typical sequences, that is at most $\frac83 2^{2k}\le 2^2 2^{2k}$. The running time for producing $U_k(B)$ follows immediately. 
\end{proof}

Let $\Gamma_1$ and $\Gamma_2$ be two $B$-trajectories. We write $\Gamma_1\le\Gamma_2$ if $\Gamma_1$ and $\Gamma_2$ have the same length, say $n$, and $\Gamma_1(i)\le \Gamma_2(i)$ for all $1\le i\le n$.
We say that $\Gamma_1 \tle \Gamma_2$ if
there are extensions $\Gamma_1^*\in E(\Gamma_1)$ and $\Gamma_2^*\in E(\Gamma_2)$ such that $\Gamma_1^*\leq \Gamma_2^*$.

\begin{LEM}\label{lem:width}
  Let  $\Gamma$, $\Delta$ be two $B$-trajectories. If $\Gamma\tle\Delta$, then the width of $\Gamma$ is less than or equal to the width of $\Delta$.
\end{LEM}
\begin{proof}
This is trivial from definitions.
\end{proof}
A \emph{lattice path with diagonal steps} from $(1,1)$ to $(x,y)$ on $\mathbb Z\times \mathbb Z$ is a sequence $v_1,v_2,\ldots,v_m$ of distinct points in $\mathbb Z\times \mathbb Z$ such that $v_1=(1,1)$, $v_m=(x,y)$, and $v_{i+1}-v_i\in \{(1,0),(0,1),(1,1)\}$ for all $i=1,2,\ldots,m-1$.
Obviously the number of lattice paths with diagonal steps from $(1,1)$ to $(x,y)$ depends only on $x$ and $y$.
The following lemma indicates that comparing two $B$-trajectories can be done without going through infinitely many extensions.

\begin{LEM}\label{lem:comparison}
  Let $\Gamma$, $\Delta$ be two $B$-trajectories. Then $\Gamma\tle \Delta$ if and only if there exists a lattice path $P=v_1,v_2,\ldots,v_m$ with diagonal steps from $(1,1)$ to $(\abs{\Gamma},\abs{\Delta})$ such that $\Gamma(x_i)\le \Delta(y_i)$ when $v_i=(x_i,y_i)$ for each $i=1,2,\ldots,m$.
\end{LEM}
\begin{proof}
  Suppose that there exist $\Gamma^*\in E(\Gamma)$ and $\Delta^*\in E(\Delta)$ such that $\Gamma^*\le \Delta^*$. We may assume that $m=\abs{\Gamma^*}(=\abs{\Delta^*})$ is chosen to be minimum. Let $1= x_1\le x_2\le \cdots\le x_m=\abs{\Gamma}$ be integers such that $\Gamma^*(i)=\Gamma(x_i)$ for all $i\in\{1,2,\ldots,m\}$ and $x_{i+1}-x_i\le 1$ for all $i\in \{1,2,\ldots,m-1\}$. Similarly choose integers $1= y_1\le y_2\le\cdots\le y_m=\abs{\Delta}$ so that $\Delta^*(i)=\Delta(y_i)$ for all $i\in\{1,2,\ldots,m\}$ and $y_{i+1}-y_i\le 1$ for all $i\in\{1,2,\ldots,m-1\}$. Let $v_i=(x_i,y_i)$ for $1\le i\le m$.

  If $x_{i+1}=x_i$ and $y_{i+1}=y_i$, then we can remove $\Gamma^*(i+1)$ and $\Delta^*(i+1)$ from $\Gamma^*$ and $\Delta^*$, respectively, contradicting our assumption that $m$ is chosen to be minimum. Thus we conclude that $v_{i+1}-v_i\in\{(1,0),(0,1),(1,1)\}$ for all $i$. This proves the forward implication.

  To prove the converse, let us define $\Gamma^*(i)=\Gamma(x_i)$ and $\Delta^*(i)=\Delta(y_i)$ for all $i=1,2,\ldots,m$. Then $\Gamma^*\in E(\Gamma)$, $\Delta^*\in E(\Delta)$ and $\Gamma^*\le \Delta^*$ and therefore $\Gamma\tle \Delta$.
\end{proof}

We say that $\Gamma_1\teq \Gamma_2$ if $\Gamma_1\tle \Gamma_2$ and $\Gamma_2\tle \Gamma_1$.

\begin{LEM}\label{lem:transitive}
The binary relations $\le$ and $\tle$ on $B$-trajectories are transitive.
\end{LEM}
\begin{proof}
 By definition, it is easy to see that $\le$ is transitive.

  Suppose $\Gamma_1\tle \Gamma_2$ and $\Gamma_2\tle \Gamma_3$.
  Then there are extensions $\Gamma_1^*\in E(\Gamma_1)$ and $\Gamma_2^*\in E(\Gamma_2)$ such that $\Gamma_1^*\le \Gamma_2^*$.
  It is easy to see that $\Gamma_2^*\tle\Gamma_3$ and therefore there exist extensions $\Gamma_2^{**}\in E(\Gamma_2^*)$ and $\Gamma_3^*\in E(\Gamma_3)$ such that $\Gamma_2^{**}\le \Gamma_3^*$. We can easily find an extension $\Gamma_1^{**}$ of $\Gamma_1$ such that $\Gamma_1^{**}\le \Gamma_2^{**}\le \Gamma_3^*$ and therefore $\Gamma_1\tle \Gamma_3$.
\end{proof}

\begin{LEM}\label{lem:bothways}
For a $B$-trajectory $\Gamma$, there exist $\Gamma_1,\Gamma_2 \in E(\tau(\Gamma))$ such that \[\Gamma_1\leq \Gamma \leq \Gamma_2.\]
\end{LEM}
\begin{proof}
  We proceed by induction on $\abs{\Gamma}$.
  Let $\Gamma=a_1,a_2,\ldots,a_m$.
  If $\tau(\Gamma)=\Gamma$, then this is trivial.
  If $a_{i-1}=a_i$ for some $i$, then we apply the induction hypothesis to $\Gamma'=a_1,a_2,\ldots,a_{i-1},a_{i+1},a_{i+2},\ldots,a_m$ and deduce this lemma.
  Thus, we may assume that there exist $i$ and $j$ such that $i+1<j$, $L(a_i)=L(a_j)$, $R(a_i)=R(a_j)$ and $\lambda(a_i)\le \lambda(a_k)\le \lambda(a_j)$ for all $k\in \{i+1,i+2,\ldots,j-1\}$ or $\lambda(a_i)\ge\lambda(a_k)\ge \lambda(a_j)$ for all $k\in \{i+1,i+2,\ldots,j-1\}$. Note that $L(a_i)=L(a_k)=L(a_j)$ and $R(a_i)=R(a_k)=R(a_j)$ for all $k\in \{i+1,i+2,\ldots,j-1\}$.

  Let $\Gamma'=a_1,a_2,\ldots,a_i,a_j,a_{j+1},a_{j+2},\ldots,a_m$. By definition, $\tau(\Gamma')=\tau(\Gamma)$ and therefore by the induction hypothesis, there exist $\Gamma_1',\Gamma_2'\in E(\tau(\Gamma))$ such that $\Gamma_1'\le \Gamma'\le \Gamma_2'$.
  Let \begin{align*}
        \Gamma_1'&=b_1,b_2,\ldots,b_i,b_j,b_{j+1},b_{j+2},\ldots,b_m \text{ and }\\
        \Gamma_2'&=c_1,c_2,\ldots,c_i,c_j,c_{j+1},c_{j+2},\ldots,c_m.
      \end{align*}
  (We intentionally skip some indices for the convenience of the proof.)

  If $\lambda(a_i)\le \lambda(a_j)$, then
  define $b_{i+1}$, $b_{i+2}$, $\ldots$, $b_{j-1}$ to be equal to $b_i$ and
  define $c_{i+1}$, $c_{i+2}$, $\ldots$, $c_{j-1}$ to be equal to $c_j$. Let $\Gamma_1=b_1,b_2,\ldots,b_m$ and $\Gamma_2=c_1,c_2,\ldots,c_m$. Clearly $\Gamma_1,\Gamma_2\in E(\tau(\Gamma))$.
  For all $i<k<j$, $b_k=b_i\le a_i\le a_k \le a_j\le c_j= c_k$. This proves that $\Gamma_1\le \Gamma\le \Gamma_2$. %

  Similarly, we can  find appropriate $\Gamma_1$, $\Gamma_2$ if $\lambda(a_i)>\lambda(a_j)$.
\end{proof}

The following corollary is an immediate consequence of Lemma~\ref{lem:bothways}.
\begin{COR}\label{cor:tauequivsame}
For all $B$-trajectories $\Gamma$, $\tau(\Gamma)\teq \Gamma$.
\end{COR}

The next corollary follows from Corollary~\ref{cor:tauequivsame} and the transitivity of $\tle$.

\begin{COR}\label{cor:tautrans}
For two $B$-trajectories $\Gamma_1$ and $\Gamma_2$, $\Gamma_1\tle \Gamma_2$ if and only if $\tau(\Gamma_1)\tle \tau(\Gamma_2)$.
\end{COR}

\subsection{Realizable $B$-trajectories}\label{subsec:realizable}

For a linear layout $\sigma=V_1,V_2,\ldots, V_m$ of a subspace arrangement $\V$,
the \emph{canonical $B$-trajectory of $\sigma$} is the $B$-trajectory $\Gamma_{\sigma}=a_1,a_2,\ldots,a_{m+1}$
such that 
\begin{align*}
  L(a_i)&=(V_1+V_2+\cdots+V_{i-1})\cap B,\\
  R(a_i)&=(V_{i}+V_{i+1}+\cdots+V_m)\cap B,\\
  \lambda(a_i)&=\dim (V_1+V_2+\cdots+V_{i-1})\cap (V_{i}+V_{i+1}+\cdots+V_m)\\
  &\quad\quad-\dim (V_1+V_2+\cdots+V_{i-1})\cap (V_{i}+V_{i+1}+\cdots+V_m)\cap B
\end{align*}
for all $i=1,2,\ldots,m+1$. Note that $L(a_1)=\{0\}$, $R(a_{m+1})=\{0\}$, $\lambda(a_1)=\lambda(a_{m+1})=0$ and $\lambda(a_i)\ge 0$.

We say that a $B$-trajectory $\Gamma$ is \emph{realizable in $\V$} 
if some extension of $\Gamma$ is the canonical $B$-trajectory of some linear layout of $\V$.

\subsection{Projection}
For subspaces $B$ and $B'$, we will define how to obtain a $B'$-trajectory from a $B$-trajectory.
For a statistic $a=(L,R,\lambda)$, we write \[a|_{B'}=(L\cap B',R\cap B', \lambda+\dim (L\cap R)-\dim (L\cap R\cap B')).\]
For a $B$-trajectory $\Gamma=a_1,a_2,\ldots,a_n$, its \emph{projection} $\Gamma|_{B'}$ is the $B'$-trajectory  $a_1|_{B'}$, $a_2|_{B'}$, $\ldots$, $a_n|_{B'}$.
Observe that the width of $\Gamma$ is less than or equal to the width of $\Gamma|_{B'}$.

\begin{LEM}\label{lem:projstat}
  Let $B$ be a subspace of $\F^r$. If $a$, $b$ are statistics such that $a\le b$, then $a|_B\le b|_B$.
\end{LEM}
\begin{proof}
  Let $a'=a|_B$ and $b'=b|_B$. It is easy to see that %
  \begin{align*}
    L(a')&=L(a)\cap B= L(b)\cap B=L(b'), \\
    R(a')&=R(a)\cap B= R(b)\cap B=R(b'), \\
\lambda(a')&=\lambda(a)+\dim L(a)\cap R(a)-\dim L(a)\cap R(a)\cap B\\
			&\leq \lambda(b) +\dim L(b)\cap R(b)-\dim L(b)\cap R(b)\cap B  \\
			&= \lambda(b').
	\end{align*}
    Therefore, $a|_B\le b|_B$.
\end{proof}
\begin{LEM}\label{lem:projcomp}
  Let $B$, $B'$ be subspaces of $\F^r$ and let $\Gamma$, $\Delta$ be $B$-trajectories.
  If $\Gamma\tle \Delta$, then $\Gamma|_{B'}\tle \Delta|_{B'}$.
\end{LEM}
\begin{proof}
  Without loss of generality, we may assume that $\Gamma\le \Delta$. Lemma~\ref{lem:projstat} easily implies that $\Gamma|_{B'}\le \Delta|_{B'}$.
\end{proof}

\begin{LEM}\label{lem:projcan}
  Let $\V$ be a subspace arrangement of subspaces of $\F^r$ %
  and let $\sigma$ be a linear layout of $\V$.
  Let $B$ be a subspace of $\F^r$ and let $B'$ be a subspace of $B$.
  If $\Gamma$ is the canonical $B$-trajectory of $\sigma$, then
  $\Gamma|_{B'}$ is the canonical $B'$-trajectory of $\sigma$.
\end{LEM}
\begin{proof}
  Let $\sigma= V_1, V_2, \ldots, V_m$.
  Let $\Gamma'$ be the canonical $B'$-trajectory of $\sigma$.
  For all $1\le i\le m+1$, we observe that
  \begin{align*}
    L(\Gamma'(i))&=(V_1+V_2+\cdots+V_{i-1})\cap B'= L(\Gamma(i))\cap B',\\
    R(\Gamma'(i))&=(V_i+V_{i+1}+\cdots+V_m)\cap B'= R(\Gamma(i))\cap B',
\end{align*}
and 
$\lambda(\Gamma'(i))
=\dim (V_1+V_2+\cdots+V_{i-1})\cap (V_i+V_{i+1}+\cdots+V_m)  - \dim L(\Gamma'(i))\cap R(\Gamma'(i)) 
= \lambda(\Gamma(i))+\dim L(\Gamma(i))\cap R(\Gamma(i))-   \dim L(\Gamma(i))\cap R(\Gamma(i))\cap B'$.
Therefore, $\Gamma'(i)=\Gamma(i)|_{B'}$ and so $\Gamma'=\Gamma|_{B'}$.
\end{proof}

\begin{LEM}\label{lem:proj}
  Let $\V$ be a subspace arrangement of subspaces of $\F^r$.
  Let $B$ be a subspace of $\F^r$ and let $B'$ be a subspace of $B$.
  If $\Gamma$ is a realizable $B$-trajectory in $\V$, then $\Gamma|_{B'}$ is realizable in $\V$.
\end{LEM}
\begin{proof}
  Let $\Gamma^*\in E(\Gamma)$ be the canonical $B$-trajectory of a linear layout $\sigma$ of $\V$.
  By Lemma~\ref{lem:projcan}, $\Gamma^*|_{B'}$ is the canonical $B'$-trajectory of $\sigma$.
  Observe now that $\Gamma^*|_{B'}$ is an extension of $\Gamma|_{B'}$
  and therefore $\Gamma|_{B'}$ is realizable in $\V$.  
\end{proof}

\subsection{Sum of two $B$-trajectories}\label{subs:sum}
For two statistics $a$ and $b$,  let $a*b$ denote the statistic $(L(a)+L(b), R(a)+R(b), \lambda(a)+\lambda(b))$.
For a statistic $a$ and a  nonnegative integer $n$, let $a+n$ be the statistic $(L(a),R(a),\lambda(a)+n)$.

For two $B$-trajectories $\Gamma_1$, $\Gamma_2$ and a lattice path $P=v_1,v_2,\ldots,v_{\abs{\Gamma_1}+\abs{\Gamma_2}-1}$ from $(1,1)$ to $(\abs{\Gamma_1},\abs{\Gamma_2})$,  we define the \emph{sum} $\Gamma_1 +_P \Gamma_2$ as the $B$-trajectory $\Gamma$ of length $\abs{\Gamma_1}+\abs{\Gamma_2}-1$ such that
 if $v_i=(x_i,y_i)$, then 
\begin{align*}
\Gamma(i)=\Gamma_1(x_i)*\Gamma_2(y_i)+(&\dim R(\Gamma_1(1))\cap R(\Gamma_2(1))\\
  &- \dim (L(\Gamma_1(x_i))+R(\Gamma_1(x_i)))\cap (L(\Gamma_2(y_j)+R(\Gamma_2(y_j)))))
\end{align*}
for all $1\le i\le \abs{\Gamma_1}+\abs{\Gamma_2}-1$.
The \emph{sum set} $\Gamma_1\tplus \Gamma_2$ of two $B$-trajectories $\Gamma_1$ and $\Gamma_2$ is the set of all sums $\Gamma_1+_P \Gamma_2$ for all possible lattice paths $P$ from $(1,1)$ to $(\abs{\Gamma_1},\abs{\Gamma_2})$. %

\begin{LEM}\label{lem:sumwidth}
Let $\Gamma_1$, $\Gamma_2$ be two $B$-trajectories and $\Gamma\in\Gamma_1\tplus\Gamma_2$.
Then the width of $\Gamma_1$ and $\Gamma_2$ are at most the width of $\Gamma$.
\end{LEM}

\begin{proof}
Let $P=(x_1,y_1), (x_2,y_2), \ldots, (x_n,y_n)$ be a lattice path 
from $(1,1)$ to $(\abs{\Gamma_1},\abs{\Gamma_2})$ such that $\Gamma=\Gamma_1+_{P}\Gamma_2$.
It is enough to show that $\lambda(\Gamma_1(x_i))\le\lambda(\Gamma(i))$ for all $1\le i\le n$.

We will show 
$\dim R(\Gamma_1(1))\cap R(\Gamma_2(1)) - \dim (L(\Gamma_1(x_i))+R(\Gamma_1(x_i)))\cap (L(\Gamma_2(y_i))+R(\Gamma_2(y_i)))\ge 0$ for all $1\le i\le n$.
By the definition of $B$-trajectories,
$R(\Gamma_1(1))=L(\Gamma_1(\abs{\Gamma_1}))\supseteq L(\Gamma_1(x_i))$
and 
$R(\Gamma_1(1))\supseteq R(\Gamma_1(x_i))$.
Thus, $R(\Gamma_1(1))\supseteq (L(\Gamma_1(x_i))+R(\Gamma_1(x_i)))$ 
and similarly $R(\Gamma_2(1))\supseteq(L(\Gamma_2(y_i))+R(\Gamma_2(y_i)))$.
Therefore, we conclude 
$\dim R(\Gamma_1(1))\cap R(\Gamma_2(1)) - \dim (L(\Gamma_1(x_i))+R(\Gamma_1(x_i)))\cap (L(\Gamma_2(y_i))+R(\Gamma_2(y_i)))\ge 0$.
Since $\lambda(\Gamma_2(y_i))$ is a nonnegative integer, by the definition of the sum,
$\lambda(\Gamma_1(x_i))\le\lambda(\Gamma(i))$ for all $1\le i\le n$.
\end{proof}

\begin{LEM}\label{lem:joinext}
  Let $\Gamma_1,\Gamma_2$ be $B$-trajectories and $\Gamma_1^*\in E(\Gamma_1)$ and $\Gamma_2^*\in E(\Gamma_2)$. Then for every lattice path $P$ from $(1,1)$ to $(\abs{\Gamma_1^*},\abs{\Gamma_2^*})$, there is a lattice path $Q$ from $(1,1)$ to $(\abs{\Gamma_1},\abs{\Gamma_2})$ such that $\Gamma_1^* +_P \Gamma_2^*$ is an extension of  $\Gamma_1 +_Q \Gamma_2 $.
\end{LEM}
\begin{proof}
  By induction, it is enough to prove the statement when $\abs{\Gamma_1^*}=\abs{\Gamma_1}+1$ and $\Gamma_2=\Gamma_2^*$.
  Suppose that $\Gamma_1(i)=\Gamma_1^*(i)$ for all $1\le i\le p$
  and $\Gamma_1(i-1)=\Gamma_1^*(i)$ for all $p+1\le i\le \abs{\Gamma_1^*}$.

  Let $P=v_1,v_2,\ldots,v_m$ with $m=\abs{\Gamma_1^*}+\abs{\Gamma_2}-1$. Let $(x_i,y_i)$ be the coordinate of $v_i$.
  Let $s$ be the maximum such that $x_s=p$.
  Let $t$ be the minimum such that $x_t=p+2$.
  Let us consider the subpath $P'=v_s, v_{s+1},\ldots , v_t$ of $P$.
  Then the points on $P'$ are
  \[ (p, y_s), (p+1,y_{s+1}), (p+1,y_{s+2}),\ldots,(p+1,y_{t-1}), (p+2,y_{t})\]
  where $y_s=y_{s+1}<y_{s+2}<\cdots<y_{t-2}<y_{t-1}=y_t$.

  Let us construct $Q=w_1,w_2,\ldots,w_{m-1}$ as follows:
  \[
  w_i=
  \begin{cases}
    v_i &\text{if } i\le s,\\
    v_{i+1} -(1,0) & \text{if } i>s.
  \end{cases}
  \]
  Let $(x_i',y_i')$ be the coordinate of $w_i$. Then for $i>s$, $x_i'=x_{i+1}-1\ge p$ and $y_i'=y_{i+1}$.
  Let $\Gamma=\Gamma_1+_Q \Gamma_2$ and $\Gamma'=\Gamma_1^*+_P \Gamma_2$.
  For $i\le s$, $\Gamma'(i)=\Gamma(i)$.
  For $i>s$, 
  \begin{align*}
  \Gamma'(i)&=\Gamma_1^*(x_i)*\Gamma_2(y_i)+(\dim R(\Gamma_1^*(1))\cap R(\Gamma_2(1)) \\
  &\quad- \dim (L(\Gamma_1^*(x_i))+R(\Gamma_1^*(x_i)))\cap(L(\Gamma_2(y_i))+R(\Gamma_2(y_i))))  \\
  &=\Gamma_1(x_i-1)*\Gamma_2(y_i)+(\dim R(\Gamma_1(1))\cap R(\Gamma_2(1))\\
  &\quad- \dim (L(\Gamma_1(x_i-1))+R(\Gamma_1(x_i-1)))\cap(L(\Gamma_2(y_i))+R(\Gamma_2(y_i))))  \\
  &= \Gamma_1(x_{i-1}')*\Gamma_2(y_{i-1}')+(\dim R(\Gamma_1(1))\cap R(\Gamma_2(1)) \\
  &\quad- \dim (L(\Gamma_1(x_{i-1}'))+R(\Gamma_1(x_{i-1}')))\cap(L(\Gamma_2(y_{i-1}'))+R(\Gamma_2(y_{i-1}'))))  \\
  & =\Gamma(i-1).
  \end{align*}
  Thus we conclude that $\Gamma'$ is an extension of $\Gamma$.
\end{proof}

\subsection{Join}

We first introduce three lemmas on subspaces, which are essential for join operations. 

\begin{LEM}\label{lem:join-key}
  Let $\V_1$ and $\V_2$ be subspace arrangements of subspaces of $\F^r$ and let $B$ be a subspace of $\F^r$.
  If $(\spn{\V_1}+B)\cap (\spn{\V_2}+B)=B$, then \[(X_1\cap B)+(X_2\cap B)=(X_1+X_2)\cap B\] for each subspace $X_1$ of $\spn{\V_1}$ and each subspace $X_2$ of $\spn{\V_2}$.
\end{LEM}
\begin{proof}
  It is enough to prove that if $(X_1+B)\cap (X_2+B)=B$, then $(X_1\cap B)+(X_2\cap B)=(X_1+X_2)\cap B$. Clearly, $(X_1\cap B)+(X_2\cap B)\subseteq (X_1+X_2)\cap B$.

  Note that $X_1\cap X_2\subseteq B$ because $(X_1\cap X_2)\subseteq (X_1+B)\cap (X_2+B)=B$.
Now,
  \begin{align*}
    &\dim(X_1\cap B+X_2\cap B)\\
    &= \dim X_1\cap B+\dim X_2\cap B-\dim X\cap X_2 \displaybreak[3]\\
    &= \dim(X_1)+\dim(B)+\dim(X_2)+\dim(B)\\&\quad-\dim(X_1+B)-\dim(X_2+B)-\dim X_1\cap X_2\\
    &= \dim(X_1)+\dim(B)+\dim(X_2)+\dim(B)\\&\quad-\dim(X_1+X_2+B)-\dim(X_1+B)\cap (X_2+B)-\dim X_1\cap X_2 \displaybreak[3]\\
    &= \dim(X_1)+\dim(B)+\dim(X_2)-\dim(X_1+X_2+B)-\dim X_1\cap X_2\\
    &= \dim(X_1+X_2)+\dim(B)-\dim(X_1+X_2+B)\\
    &= \dim (X_1+X_2)\cap B.\qedhere
  \end{align*}
\end{proof}

\begin{LEM}\label{lem:dim-join}
For finite-dimensional vector spaces $X_1$, $X_2$, $Y_1$, and $Y_2$,
\begin{align*}
&\dim{(X_1+X_2)\cap\left(Y_1+Y_2\right)}\\
&=\dim{X_1\cap Y_1}+\dim{X_2\cap Y_2}
+\dim{\left(X_1+Y_1\right)\cap\left(X_2+Y_2\right)}
-\dim{X_1\cap X_2}-\dim{Y_1\cap Y_2}.
\end{align*}
\end{LEM}
\begin{proof}
Recall that for two finite-dimensional vector spaces $X$ and $Y$,
$\dim(X)+\dim(Y)=\dim(X+Y)+\dim(X\cap Y)$.
Using this fact, we deduce as follows.
\begin{align*}
 &\dim{(X_1+X_2)\cap(Y_1+Y_2)}\\
&=\dim{(X_1+X_2)}+\dim{(Y_1+Y_2)}
-\dim{(X_1+X_2+Y_1+Y_2)}\\
&=\dim{X_1}+\dim{X_2}-\dim{X_1\cap X_2}+\dim{Y_1}+\dim{Y_2}-\dim{Y_1\cap Y_2}\\
&\quad-\dim{(X_1+X_2+Y_1+Y_2)}\\
&=\dim{(X_1+Y_1)}+\dim{X_1\cap Y_1}
+\dim{(X_2+Y_2)}+\dim{X_2\cap Y_2}\\
&\quad-\dim{X_1\cap X_2}
-\dim{Y_1\cap Y_2}
-\dim{(X_1+X_2+Y_1+Y_2)}\\
&=\dim{(X_1+Y_1)\cap(X_2+Y_2)}
+\dim{X_1\cap Y_1}
+\dim{X_2\cap Y_2}\\
&\quad-\dim{X_1\cap X_2}
-\dim{Y_1\cap Y_2}.\qedhere
\end{align*}
\end{proof}

\begin{LEM}\label{lem:dimension}
Let $\V_1$ and $\V_2$ be subspace arrangements of subspaces of $\F^r$ and $B$ be a subspace of $\F^r$.
If $(\spn{\V_1}+B)\cap (\spn{\V_2}+ B)=B$, then 
\begin{align*}
\lefteqn{ \dim (X_1+X_2) \cap (Y_1+Y_2) - \dim (X_1+X_2)\cap (Y_1+Y_2)\cap B}\\
&= \dim X_1\cap Y_1 -\dim X_1\cap Y_1\cap B + \dim X_2\cap Y_2-\dim X_2\cap Y_2\cap B\\
&\quad  + \dim (X_1+Y_1)\cap (X_2+Y_2)\cap B - \dim (X_1\cap B+Y_1\cap B)\cap (X_2\cap B+Y_2\cap B)
\end{align*}
for subspaces $X_1,Y_1$ of $\spn{\V_1}$ and $X_2,Y_2$ of $\spn{\V_2}$.
\end{LEM}
\begin{proof}
From Lemma~\ref{lem:join-key}, observe that
\begin{multline*}
\dim (X_1+X_2) \cap (Y_1+Y_2) - \dim (X_1+X_2) \cap (Y_1+Y_2)\cap B\\
=\dim (X_1+X_2) \cap (Y_1+Y_2) - \dim (X_1\cap B + X_2\cap B) \cap (Y_1 \cap B+Y_2\cap B).
\end{multline*}
We apply Lemma~\ref{lem:dim-join} to $X_1,X_2,Y_1,Y_2$ and $X_1\cap B,X_2\cap B,Y_1\cap B,Y_2\cap B$, respectively, which yields
\begin{align*}
\dim(X_1+X_2)\cap (Y_1+Y_2) &= \dim X_1\cap Y_1 + \dim X_2\cap Y_2\\
&\quad-\dim X_1\cap X_2 - \dim Y_1\cap Y_2 + \dim (X_1+Y_1)\cap (X_2+Y_2),
\end{align*}
and 
\begin{align*}
&\dim(X_1\cap B+X_2\cap B)\cap (Y_1\cap B+ Y_2\cap B) \\
			&= \dim X_1 \cap Y_1\cap B + \dim X_2 \cap Y_2\cap B-\dim X_1\cap X_2\cap B - \dim Y_1 \cap Y_2\cap B \\ 
			&\quad+ \dim (X_1\cap B+Y_1\cap B)\cap (X_2\cap B+Y_2\cap B).
\end{align*}
Since $(X_1+Y_1)\cap (X_2+Y_2) \subseteq (X_1+Y_1+B)\cap (X_2+Y_2+B) \subseteq  (\spn{\V_1}+B)\cap (\spn{\V_2}+ B)= B$, we have 
$(X_1+Y_1)\cap (X_2+Y_2)=(X_1+Y_1)\cap (X_2+Y_2)\cap B$.
Similarly, $X_1\cap X_2=X_1\cap X_2\cap B$, $Y_1\cap Y_2=Y_1\cap Y_2\cap B$.
Thus the equality of the statement follows.
\end{proof}

\begin{PROP}\label{prop:join-one}
 Let $\V_1,\V_2$ be subspace arrangements of subspaces of $\F^r$ and $B$ be a subspace of $\F^r$
 such that $(\spn{\V_1}+B)\cap (\spn{\V_2}+B)= B$.
    If $\Gamma$ is realizable in $\V_1\dot\cup \V_2$, then there exist $B$-trajectories $\Gamma_1$, $\Gamma_2$ and a lattice path $P$ from $(1,1)$ to $(\abs{\Gamma_1},\abs{\Gamma_2})$  such that
  \begin{enumerate}
  \item $\Gamma_i$ is realizable in $\V_i$ for $i=1,2$,
  \item $\Gamma_1+_P \Gamma_2\tle \Gamma$.
  \end{enumerate}
\end{PROP}
\begin{proof}
  Let $\Gamma^*\in E(\Gamma)$ be the canonical $B$-trajectory 
  of a linear layout $\sigma=V_1,V_2,\ldots,V_m$ of $\V_1\dot\cup\V_2$.
  For $i=1,2$, let $\sigma_i=W_1^i, W_2^i, \ldots, W_{m_i}^i$ be a subsequence of $\sigma$
  by taking all $V_j\in \V_i$
  and let $\Gamma_i$ be the canonical $B$-trajectory of $\sigma_i$.
  Clearly, $m_1+m_2=m$, and $\sigma_i$ is a linear layout of $\V_i$
  and so $\Gamma_i$ is realizable in $\V_i$.
  Also, observe that $W_j^i\subseteq \spn{\V_i}$ for all $i$ and $j$.

  Let $P=v_1,v_2,\ldots,v_{m}$ be a lattice path from $(1,1)$ to $(m_1+1,m_2+1)$ such that for $i=1,2,\ldots,m$,
  \[
  v_{i+1}=
  \begin{cases}
    v_i+(1,0)  &\text{if  } V_i\in\V_1,\\ %
    v_i+(0,1)  &\text{if  } V_i\in\V_2. %
  \end{cases}
  \]
  We claim that $\Gamma_1+_P\Gamma_2= \Gamma^*$, which implies that $\Gamma_1+_P\Gamma_2\tle \Gamma$. 
  As $\Gamma^*$ has length $m+1$ and $\Gamma_i$ has length $m_i+1$ for $i=1,2$, both sides have the same length. Let $a_j$ be the $j$-th statistic of $\Gamma_1+_P\Gamma_2$ for $B$. It remains to prove that $\Gamma(j)=a_j$ for all $1\le j\le m+1$.

  Let $v_j=(x_j,y_j)$.
  Let $L_1=W^1_1+W^1_2+\cdots+W^1_{x_j-1}$, $L_2=W^2_1+W^2_2+\cdots+W^2_{y_j-1}$,
  $R_1=W^1_{x_j}+W^1_{x_j+1}+\cdots+W^1_{m_1}$, and $R_2= W^2_{y_j}+W^2_{y_j+1}+\cdots+W^2_{m_2}$.   
  By the construction and Lemma~\ref{lem:join-key}, 
  $x_j-1+y_j-1=j-1$ and $L(a_j)=L(\Gamma_1(x_j))+L(\Gamma_2(y_j)) =L_1\cap B + L_2\cap B=(L_1+L_2)\cap B$.
  Since $L_1+L_2=V_1+V_2+\cdots +V_{j-1}$, we deduce that $L(a_j)=L(\Gamma(j))$. 
  Similarly, $R(a_j)=R(\Gamma(j))$.

  It remains to prove that $\lambda(\Gamma(j))=\lambda(a_j)$.
First observe that $\lambda(\Gamma_1(x_j))=\dim L_1\cap R_1-\dim L_1\cap R_1\cap B$ and
$\lambda(\Gamma_2(y_j))=\dim L_2\cap R_2-\dim L_2\cap R_2\cap B$. 
Also, $R(\Gamma_1(1))=(L_1+R_1)\cap B$ and $R(\Gamma_2(1))=(L_2+R_2)\cap B$.
By definition,
$L_1$ and $R_1$ are subspaces of $\spn{\V_1}$, and $L_2$ and $R_2$ are subspaces of $\spn{\V_2}$.
  Therefore, by Lemma~\ref{lem:dimension}, 
  \begin{align*}
  	&\lambda(\Gamma(j))\\&=\dim(L_1+L_2)\cap (R_1+R_2)-\dim(L_1+L_2)\cap (R_1+R_2)\cap B \\
  	&=\dim L_1\cap R_1 - \dim L_1\cap R_1 \cap B + \dim L_2\cap R_2 - \dim L_2\cap R_2\cap B \\
  	&\quad +\dim(L_1+R_1)\cap (L_2 +R_2)\cap B - \dim (L_1\cap B+R_1\cap B)\cap (L_2\cap B + R_2 \cap B) \\
  	&=\lambda(  \Gamma_1(x_j))+\lambda(\Gamma_2(y_j))+\dim R(\Gamma_1(1))\cap R(\Gamma_2(1)) \\
  	&\quad -\dim (L(\Gamma_1(x_j))+R(\Gamma_1(y_j)))\cap (L(\Gamma_2(x_j))+R(\Gamma_2(y_j))) \\
  	&=\lambda(a_j). \qedhere
  \end{align*}
\end{proof}

\begin{PROP}\label{prop:composition}
  Let $\Gamma_1,\Gamma_1',\Gamma_2,\Gamma_2'$ be $B$-trajectories.
  If $\Gamma_1'\tle\Gamma_1$ and $\Gamma_2'\tle\Gamma_2$, then for all $\Gamma\in \Gamma_1\tplus \Gamma_2$, there exists $\Gamma'\in\Gamma_1'\tplus \Gamma_2'$ such that $\Gamma'\tle \Gamma$.
\end{PROP}
To prove this proposition, we use the following lemmas.
\begin{LEM}\label{lem:latticejoin}
  Let $\Gamma_1,\Gamma_1',\Gamma_2$ be $B$-trajectories such that $\Gamma_1'\le \Gamma_1$. Let $P$ be a lattice path from $(1,1)$ to $(\abs{\Gamma_1},\abs{\Gamma_2})$. Then $\Gamma_1 ' +_P \Gamma_2 \le \Gamma_1+_P \Gamma_2$.
\end{LEM}
\begin{proof}
  Observe that $\Gamma_1'+_P\Gamma_2$ is well defined because $\abs{\Gamma_1'}=\abs{\Gamma_1}$. Let $\Gamma'=\Gamma_1'+_P \Gamma_2$ and $\Gamma=\Gamma_1+_P \Gamma_2$. It is enough to prove that for $1\le j\le \abs{\Gamma}=\abs{\Gamma_1}+\abs{\Gamma_2}-1$, we have $\Gamma'(j)\le \Gamma(j)$.
  Let $(x_j,y_j)$ be the $j$-th point of $P$. Let $a=\Gamma_1(x_j)$, $a'=\Gamma_1'(x_j)$, and $b=\Gamma_2(y_j)$.
  First, observe that $L(\Gamma'(j))=L(a')+L(b)
  = L(a)+L(b)=L(\Gamma(j))$   because $L(a')=L(a)$ and similarly $R(\Gamma'(j))= R(\Gamma(j))$.
  Finally,  %
 $\lambda(\Gamma'(j))= \lambda(a')+\lambda(b) 
 + \dim R(\Gamma_1'(1))\cap R(\Gamma_2(1))-\dim (L(a')+R(a'))\cap (L(b)+R(b)) 
 \le \lambda(a)+\lambda(b) 
  + \dim R(\Gamma_1(1))\cap R(\Gamma_2(1))-\dim (L(a)+R(a))\cap (L(b)+R(b))=\lambda(\Gamma(j))$. 
  This proves the lemma.
\end{proof}

\begin{LEM}\label{lem:naturalext}
Let $\Gamma_1$ and $\Gamma_2$ be two $B$-trajectories and let $P$ be a lattice path from $(1,1)$ to $(\abs{\Gamma_1}, \abs{\Gamma_2})$. 
For two extensions $\Gamma'_1\in E(\Gamma_1)$, $\Gamma_2'\in E(\Gamma_2)$, 
there exists a lattice path $P'$ from $(1,1)$ to $(\abs{\Gamma'_1},\abs{\Gamma'_2})$ such that $\Gamma'_1 +_{P'} \Gamma'_2$ is an extension of $\Gamma_1 +_P\Gamma_2$.
\end{LEM}
\begin{proof}
It is enough to prove the statement when $\abs{\Gamma'_1}=\abs{\Gamma_1}+1$ and $\Gamma'_2=\Gamma_2$, from which we can inductively deduce our main statement. The lattice path $P$ is denoted as $v_1,v_2, \ldots , v_m$ with $v_i=(x_i,y_i)$. Let $a$ be a statistic of $\Gamma_1$ such that $\Gamma'_1$ is obtained by repeating $a$ once. We arbitrarily choose $1\leq p \leq m$ such that $\Gamma_1(x_p)=a$. Notice that $\Gamma'_1(i)=\Gamma_1(i-1)$ for $i > x_p$. Take $P'=v'_1,v'_2,\ldots , v'_{m+1}$ with $v'_i=(x'_i,y'_i)$ as follows:
  \[
  v'_i=
  \begin{cases}
    (x_i,y_i) &\text{if  } i\leq p,\\
    (x_{i-1}+1,y_{i-1})  &\text{if  } i \geq p+1.
  \end{cases}
  \]
We claim that $\Gamma'=\Gamma'_1 +_{P'} \Gamma_2$ is an extension of $\Gamma_1+_P \Gamma_2$. 
For $1\leq i\leq p$, 
\begin{align*}
\Gamma'(i)
&=\Gamma'_1(x'_i)*\Gamma_2(y'_i)+(\dim R(\Gamma'_1(1))\cap R(\Gamma_2(1)) \\
&\quad - \dim (L(\Gamma'_1(x'_i))+R(\Gamma'_1(x'_i)))\cap (L(\Gamma'_2(y'_i))+R(\Gamma'_2(y'_i))))\\
&=\Gamma_1(x_i)*\Gamma_2(y_i)+(\dim R(\Gamma_1(1))\cap R(\Gamma_2(1))\\
&\quad - \dim (L(\Gamma_1(x_i))+R(\Gamma_1(x_i)))\cap (L(\Gamma_2(y_i))+R(\Gamma_2(y_i)))). 
\end{align*}
For $p+1\leq i\leq m+1$, 
\begin{align*}
\Gamma'(i)
&=\Gamma'_1(x_{i-1}+1)*\Gamma_2(y_{i-1})+(\dim R(\Gamma'_1(1))\cap R(\Gamma_2(1))\\
&\quad -\dim (L(\Gamma'_1(x_{i-1}+1))+R(\Gamma'_1(x_{i-1}+1)))\cap (L(\Gamma_2(y_{i-1}))+R(\Gamma_2(y_{i-1})))) \\
&=\Gamma_1(x_{i-1})*\Gamma_2(y_{i-1})+(\dim R(\Gamma_1(1))\cap R(\Gamma_2(1))\\
&\quad-\dim (L(\Gamma_1(x_{i-1}))+R(\Gamma_1(x_{i-1})))\cap (L(\Gamma_2(y_{i-1}))+R(\Gamma_2(y_{i-1})))). 
\end{align*}
Therefore, $\Gamma'$ can be obtained from $\Gamma_1+_P \Gamma_2$ by repeating the statistic $\Gamma_1(x_p)+\Gamma_2(y_p)$.
\end{proof}

\begin{proof}[Proof of Proposition~\ref{prop:composition}]
  Since $\Gamma_1\tplus \Gamma_2=\Gamma_2\tplus \Gamma_1$ and $\tle$ is transitive by Lemma~\ref{lem:transitive}, it is enough to prove it for the case that $\Gamma_2'=\Gamma_2$.  Let $\Gamma_1''\in E(\Gamma_1')$ and $\Gamma_1^*\in E(\Gamma_1)$ such that $\Gamma_1''\le \Gamma_1^*$.  Let $P$ be a lattice path from $(1,1)$ to $(\abs{\Gamma_1},\abs{\Gamma_2})$ such that $\Gamma=\Gamma_1+_P \Gamma_2$. From Lemma~\ref{lem:naturalext}, there is a lattice path $P'$ from $(1,1)$ to $(\abs{\Gamma_1^*},\abs{\Gamma_2})$ such that $\Gamma_1^*+_{P'} \Gamma_2$ is an extension of $\Gamma$. Note that $\Gamma_1^*+_{P'} \Gamma_2\tle \Gamma$. By Lemma~\ref{lem:latticejoin}, $\Gamma_1''+_{P'} \Gamma_2\le \Gamma_1^*+_{P'} \Gamma_2$.  By Lemma~\ref{lem:joinext}, there is $\Gamma'\in \Gamma_1'\tplus\Gamma_2$ such that $\Gamma_1''+_{P'}\Gamma_2$ is an extension of $\Gamma'$.
\end{proof}

\begin{PROP}\label{prop:compreal}
Let $\V_1$, $\V_2$ be subspace arrangements of subspaces of $\F^r$
and let $B$ be a subspace of $\F^r$ such that $(\spn{\V_1}+B)\cap (\spn{\V_2}+B)=B$.
Let $\Gamma_1$ and $\Gamma_2$ be two $B$-trajectories realizable in $\V_1$ and $\V_2$, respectively. Then every $\Gamma \in \Gamma_1\tplus \Gamma_2$ is realizable in $\V_1\dot\cup \V_2$.
\end{PROP}
Let us delay proving Proposition~\ref{prop:compreal} by presenting a helpful lemma for the proof first.
\begin{LEM}\label{lem:Psumreal}
Let $\V_1$, $\V_2$ be subspace arrangements of subspaces of $\F^r$
and let $B$ be a subspace of $\F^r$ such that $(\spn{\V_1}+B)\cap (\spn{\V_2}+B)=B$.
For $i=1,2$, let $\Gamma_i$ be the canonical $B$-trajectory of some linear layout of $\V_i$.
For every lattice path $P$ from $(1,1)$ to $
(\abs{\Gamma_1},\abs{\Gamma_2})$,
there exists 
a linear layout of $\V_1\dot\cup\V_2$ whose canonical $B$-trajectory  is $\Gamma_1+_P\Gamma_2$.
\end{LEM}
\begin{proof}
Let $P=v_1,\ldots , v_{\abs{\Gamma_1}+\abs{\Gamma_2}-1}$ be a lattice path from $(1,1)$ to $(\abs{\Gamma_1},\abs{\Gamma_2})$ with $v_j=(x_j,y_j)$.
Let $\sigma_1=S_1,S_2,\ldots, S_{\abs{\Gamma_1}-1}$ be the linear layout of $\V_1$ whose canonical $B$-trajectory is $\Gamma_1$
and $\sigma_2=T_1,T_2,\ldots, T_{\abs{\Gamma_2}-1}$ be the linear layout of $\V_2$ whose canonical $B$-trajectory is $\Gamma_2$.
We build a linear layout $\sigma$ from $\sigma_1$ and $\sigma_2$ as follows: 
for $i=1,2,\ldots,\abs{\Gamma_1}+\abs{\Gamma_2}-2$,
\[
\sigma[i]=
\begin{cases}
S_{x_i} &\text{if } v_{i+1}-v_i=(1,0),\\
T_{y_i} &\text{if } v_{i+1}-v_i=(0,1),
\end{cases}
\]
where $\sigma[i]$ is the $i$-th subspace in $\sigma$.
It is clear that $\sigma$ is a linear layout of $\V_1\dot\cup \V_2$. 
Let $\Gamma$ be the canonical $B$-trajectory of $\sigma$.

We claim that $\Gamma=\Gamma_1+_P\Gamma_2$.
Let $i\in \{1,2,\ldots,\abs{\Gamma}\}$. We shall show that $\Gamma(i)=(\Gamma_1+_P\Gamma_2)(i)$.
Let $L_1=S_1+\cdots +S_{x_i-1}$, $R_1=S_{x_i}+\cdots +S_{\abs{\Gamma_1}-1}$, $L_2=T_1+\cdots +T_{y_i-1}$ and $R_2=T_{y_i}+\cdots +T_{\abs{\Gamma_2}-1}$. Note that $S_1,S_2,\ldots,S_{\abs{\Gamma_1}-1}\subseteq \spn{\V_1}$ and $T_1,T_2,\ldots,T_{\abs{\Gamma_2}-1}\subseteq\spn{\V_2}$, and thus $L_1,R_1\subseteq \spn{\V_1}$ and $L_2,R_2\subseteq \spn {\V_2}$. This allows us to use Lemma~\ref{lem:join-key}. %
We deduce from Lemma~\ref{lem:join-key} that
\begin{multline*}
  L(\Gamma(i))=(L_1+L_2)\cap B
              =L_1\cap B + L_2\cap B\\
              =L(\Gamma_1(x_i))+L(\Gamma_2(y_i))
                =L(\Gamma_1(x_i)+\Gamma_2(y_i))
\end{multline*}
where the last equality follows from the definition of the sum $\Gamma_1 +_P\Gamma_2$. 
Similarly,
$R(\Gamma(i))=R(\Gamma_1(x_i)+\Gamma_2(y_i))=R((\Gamma_1+_P\Gamma_2)(i)).$

From $(\spn{\V_1}+B)\cap (\spn{\V_2}+B)=B$, it holds that for every subspace $X_1$ of $\spn{\V_1}$ and subspace $X_2$ of $\spn{\V_2}$, $X_1\cap X_2\subseteq (X_1+B)\cap(X_2+B)\subseteq B$. 
Since $R(\Gamma_1(1))=(L_1+R_1)\cap B$ and  $R(\Gamma_2(1))=(L_2+R_2)\cap B$,
by Lemma~\ref{lem:dimension}, we have
\begin{align*}
	\lambda(\Gamma(i)) &= \dim ( L_1+L_2 )\cap (R_1+R_2) -\dim (L_1+L_2)\cap (R_1+R_2)\cap B\\
	&=\dim L_1\cap R_1-\dim L_1\cap R_1\cap B+\dim L_2\cap R_2-\dim L_2\cap R_2\cap B\\
	&\quad + \dim (L_1+R_1)\cap (L_2+R_2)\cap B \\
	&\quad -\dim (L_1\cap B+R_1\cap B) \cap (L_2\cap B+R_2\cap B) \\
	&=\lambda(\Gamma_1(x_i))+\lambda(\Gamma_2(y_i))+\dim R(\Gamma_1(1)) \cap R(\Gamma_2(1))\\
	&\quad 	-\dim (L(\Gamma_1(x_i))+R(\Gamma_1(x_i)))\cap (L(\Gamma_2(y_i))+R(\Gamma_2(y_i)))\\
	&=\lambda((\Gamma_1+_P\Gamma_2)(i)).
\end{align*}
Hence, we conclude that $\Gamma=\Gamma_1+_P\Gamma_2$.
\end{proof}

Now we are ready to finish the proof of Proposition~\ref{prop:compreal}.
\begin{proof}[Proof of Proposition~\ref{prop:compreal}]
For $i=1,2$,
let $\Gamma_i^*\in E(\Gamma_i)$ be the canonical $B$-trajectory of some linear layout of $\V_i$.
Let $P$ be a lattice path from $(1,1)$ to $(\abs{\Gamma_1},\abs{\Gamma_2})$ 
such that $\Gamma=\Gamma_1+_P\Gamma_2$.
By Lemma~\ref{lem:naturalext}, there exists a lattice path $P'$ from $(1,1)$ to $(\abs{\Gamma_1^*},\abs{\Gamma_2^*})$ 
such that $\Gamma_1^* +_{P'} \Gamma_2^*$ is an extension of $\Gamma$. 
Lemma~\ref{lem:Psumreal} implies that 
$\Gamma_1^* +_{P'} \Gamma_2^*$ is the canonical $B$-trajectory of a linear layout of $\V_1\dot\cup \V_2$.
Therefore, $\Gamma$ is realizable in $\V_1\dot\cup\V_2$.
\end{proof}

\section{The full set for dynamic programming} \label{sec:fullset}

Let $\V$ be a subspace arrangement of subspaces of $\F^r$ and $B$ be a subspace of $\F^r$.
The \emph{full set} of $\V$ of width $k$ with respect to $B$, denoted by $\FS(\V,B)$, is the set of all compact $B$-trajectories $\Gamma$ of width at most $k$ such that there exists a $B$-trajectory $\Delta$ realizable in $\V$ with $\Delta\tle\Gamma$.

For a set $\mathcal R$ of $B$-trajectories and $B\subseteq B'$, we define $\up(\mathcal R,B')$ as the set of all compact $B'$-trajectories $\Gamma$ of width at most $k$ such that there exists $\Delta\in\mathcal{R}$ with $\Delta\tle\Gamma$ (considering $\Delta$ as a $B'$-trajectory).
For subspaces $B'\subseteq B$ and a set $\mathcal{R}$ of $B$-trajectories, let $\mathcal{R}|_{B'}$ be
the set of all $B'$-trajectories $\Gamma|_{B'}$ where $\Gamma\in \mathcal{R}$.
For two sets $\mathcal{R}_1$, $\mathcal{R}_2$ of $B$-trajectories, we define $\mathcal{R}_1\fplus\mathcal{R}_2=\bigcup_{\Gamma_1\in \mathcal{R}_1,\Gamma_2 \in \mathcal{R}_2}\Gamma_1\tplus \Gamma_2$.

We give an overview of this section. 
Each subsection describes how we are going to compute a full set.
Subsection~\ref{subsec:single} is about the case that $\abs{\V}=1$.
Subsection~\ref{subsec:expand} discusses how to obtain $\FS(\V,B')$ from a given $\FS(\V,B)$ for a subspace $B$ of $B'$ such that $\spn{\V}\cap B'\subseteq B$, which we call \emph{expanding}.
Subsection~\ref{subsec:shrink} is the opposite; it presents the construction of $\FS(\V,B')$ from a given $\FS(\V,B)$ when $B'$ is a subspace of $B$. This process is called \emph{shrinking}.
Subsection~\ref{subsec:join} explains how we obtain $\FS(\V,B)$ when we are given $\FS(\V_1,B)$ and $\FS(\V_2,B)$ for some subspace arrangements $\V_1$ and $\V_2$ such that $\V$ is the disjoint union of $\V_1$ and $\V_2$. This process is called \emph{joining}.

\subsection{A full set for a single subspace}\label{subsec:single}

\begin{PROP}\label{prop:single}
  Let $B, V$ be subspaces in $\F^r$ such that $B\subseteq V$.
  Then %
  $\FS(\{V\},B)=\up(\{\Delta\},B)$ 
  where $\Delta$ is a $B$-trajectory $(\{0\},B,0), (B,\{0\},0)$.
\end{PROP}

\begin{proof}
  Since $\{V\}$ has the unique linear layout
  and $V\cap B=B$,
  $\Delta$ is the unique canonical $B$-trajectory of the linear layout.
Thus, for a compact $B$-trajectory $\Gamma$ of width at most $k$, 
 $\Gamma\in \FS(\{V\},B)$ if and only if 
$\Delta\tle \Gamma$.
\end{proof}

\subsection{A full set for an expanding operation}\label{subsec:expand}

\begin{PROP}\label{prop:expand}
Let $\V$ be a subspace arrangement of subspaces of $\F^r$.
Let $B$, $B'$ be subspaces of $\F^r$ such that $B\subseteq B'$.
If $\spn{\V}\cap B'\subseteq B$, then $\FS(\V,B')=\up(\FS(\V,B),B')$.
\end{PROP}
\begin{proof}
Because $\spn{\V}\cap B'\subseteq B$, for every subspace $W$ of $\spn{\V}$, we have $W\cap B'\subseteq B$.
This implies that for every linear layout $\sigma$ of $\V$, 
the canonical $B$-trajectory of $\sigma$ is also the canonical $B'$-trajectory of $\sigma$ and vice versa. 
Thus, $\Gamma$ is a realizable $B$-trajectory if and only if it is a realizable $B'$-trajectory.  
The conclusion follows easily, as $\tle$ is transitive by Lemma~\ref{lem:transitive}.
\end{proof}

\subsection{A full set for a shrinking operation}\label{subsec:shrink}

\begin{PROP}\label{prop:shrink}
Let $\V$ be a subspace arrangement of subspaces of $\F^r$ and $B$, $B'$ be subspaces of $\F^r$.
If $B'\subseteq B$, then $\FS(\V,B')=\up(\FS(\V,B)|_{B'},B')$.
\end{PROP}
\begin{proof}
We first prove $\FS(\V,B')\subseteq\up(\FS(\V,B)|_{B'},B')$.
Let $\Gamma\in \FS(\V,B')$ and
$\Delta$ be a $B'$-trajectory realizable in $\V$ such that $\Delta \tle \Gamma$.
Let $\sigma$ be a linear layout of $\V$ 
whose canonical $B'$-trajectory is an extension of $\Delta$.
Let $\Delta'$ be the canonical $B$-trajectory of $\sigma$.
Since $\Gamma$ is a compact $B'$-trajectory of width at most $k$,
it is enough to show that (i) $\tau(\Delta')|_{B'}\tle \Gamma$ and (ii) $\tau(\Delta')$ is in $\FS(\V,B)$.
For (i),  $\Delta'|_{B'}$ is an extension of $\Delta$ because it is the canonical $B'$-trajectory of $\pi$.
Thus, $\Delta'|_{B'}\tle\Delta\tle\Gamma$.
By Corollary~\ref{cor:tauequivsame} and Lemma~\ref{lem:projcomp}, $\tau(\Delta')|_{B'}\tle\Delta'|_{B'}$ and
we conclude that $\tau(\Delta')|_{B'}\tle \Gamma$ by Lemma~\ref{lem:transitive}.
For (ii),
by Lemma~\ref{lem:width},
we know that $(\text{the width of }\tau(\Delta'))=(\text{the width of }\Delta')\le(\text{the width of }\Delta'|_{B'})\le(\text{the width of }\Gamma)\le k$.
By Corollary~\ref{cor:tauequivsame}, we have $\Delta'\tle\tau(\Delta')$.
Since $\Delta'$ is realizable in $\V$,  $\tau(\Delta')$ is in $\FS(\V,B)$.

Next, let us show $\FS(\V,B')\supseteq\up(\FS(\V,B)|_{B'},B')$. Choose an arbitrary $\Gamma$ from $\up(\FS(\V,B)|_{B'},B')$ and let $\Delta\in \FS(\V,B)$ such that $\Delta|_{B'}\tle \Gamma$. By the definition of the full set $\FS(\V,B)$, there exists a $B$-trajectory $\Delta'$ that is realizable in $\V$ and $\Delta'\tle \Delta$. By Lemma~\ref{lem:proj}, $\Delta'|_{B'}$ is realizable in $\V$. By Lemma~\ref{lem:projcomp}, $\Delta'|_{B'}\tle \Delta|_{B'} \tle \Gamma$. Recall that $\Gamma$ is compact and of width at most $k$, and thus $\Gamma \in \FS(\V,B')$.
\end{proof}

\subsection{A full set for a join operation}\label{subsec:join}

\begin{PROP}\label{prop:join}
Let $\V_1,\V_2$ be subspace arrangements of subspaces of $\F^r$ and let $B$ be a subspace of $\F^r$. If
$(\spn{\V_1}+B)\cap (\spn{\V_2}+B)=B,$ then 
\[\FS(\V_1\dot\cup \V_2,B)=\up(\FS(\V_1,B)\fplus \FS(\V_2,B),B).\]
\end{PROP}
\begin{proof}
We prove the ``$\supseteq$"-direction first. Let $\Gamma\in\up(\FS(\V_1,B)\fplus \FS(\V_2,B),B)$.
By definition, there exists $\Delta \in \Delta_1\tplus \Delta_2$
for some $\Delta_1\in \FS(\V_1,B)$, $\Delta_2 \in \FS(\V_2,B)$ satisfying $\Delta \tle \Gamma$.
Again by definition, for $i=1,2$,
there exists a $B$-trajectory $\Delta'_i$ realizable in $\V_i$
such that $\Delta'_i\tle \Delta_i$.
By applying Proposition~\ref{prop:composition} to $\Delta_1,\Delta_2,$ and $\Delta'_1,\Delta'_2$, 
we deduce that there exists $\Delta'\in \Delta'_1\tplus \Delta'_2$ such that $\Delta'\tle \Delta$. Proposition~\ref{prop:compreal} implies that all elements of $\Delta'_1\tplus \Delta'_2$ are realizable in $\V_1\dot\cup \V_2$, and so is $\Delta'$. From the fact that $\Delta'$ is realizable in $\V_1\dot\cup \V_2$ and $\Delta'\tle \Delta\tle \Gamma$, we conclude that $\Gamma$ belongs to the set $\FS(\V_1\dot\cup \V_2,B)$.

Now we prove the ``$\subseteq$"-direction. Consider an arbitrary $\Gamma$ from $\FS(\V_1\dot\cup \V_2, B)$. By the definition of the full set, there exists a $B$-trajectory $\Delta$ realizable in $\V_1\dot\cup \V_2$ such that $\Delta\tle \Gamma$. By Proposition~\ref{prop:join-one}, there exist two $B$-trajectories $\Delta_1$, $\Delta_2$ and a lattice path $P$ from $(1,1)$ to $(|\Delta_1|,|\Delta_2|)$ such that
 \begin{enumerate}
  \item $\Delta_i$ is realizable in $\V_i$ for $i=1,2$,
  \item $\Delta_1+_P \Delta_2\tle \Delta$.
  \end{enumerate}
We apply Proposition~\ref{prop:composition} to $B$-trajectories $\Delta_1,\Delta_2$ and $\tau(\Delta_1),\tau(\Delta_2)$.
Since $\tau(\Delta_1)\tle \Delta_1$ and $\tau(\Delta_2)\tle \Delta_2$, there exists $\Delta'\in \tau(\Delta_1)\tplus \tau(\Delta_2)$ such that $\Delta'\tle \Delta_1 +_P \Delta_2$, and thus $\Delta'\tle \Delta\tle \Gamma$ by Lemma~\ref{lem:transitive}.
Since $\Delta_1+_P \Delta_2\tle \Delta \tle \Gamma$ and 
$\Gamma$ has width at most $k$, 
the width of $\Delta_1$ and $\Delta_2$ are at most $k$
by Lemmas~\ref{lem:width} and~\ref{lem:sumwidth}.
So it follows that both $\tau(\Delta_1)$ and $\tau(\Delta_2)$ have width at most $k$ by Corollary~\ref{cor:tauequivsame} and Lemma~\ref{lem:width}.
Since $\Delta_i$ is realizable in $\V_i$ and $\Delta_i\tle \tau(\Delta_i)$, we have $\tau(\Delta_i) \in \FS(\V_i,B)$ for $i=1,2$.
To summarize, we have $\Delta'\in \tau(\Delta_1)\tplus \tau(\Delta_2)$ for some $\tau(\Delta_i)\in \FS(\V_i,B)$, $i=1,2$, satisfying $\Delta' \tle \Gamma$ and therefore, $\Gamma \in \up(\FS(\V_1,B)\fplus \FS(\V_2,B),B)$.
\end{proof}

\section{An algorithm}\label{sec:algorithm}
We are ready to describe our main algorithm to solve the following problem. 
We first describe how subspaces are given as an input. For an $r\times m$ matrix $M$ and $X\subseteq\{1,2,\ldots,r\}$, $Y\subseteq \{1,2,\ldots, m\}$, let us write $M[X,Y]$ to denote the submatrix of $M$ induced by the rows in $X$ and the columns in $Y$ and $M[Y]:=M[\{1,2,\ldots,r\}, Y]$.

\begin{description}
\item [Input]
An $r\times m$ matrix $M$ over a fixed finite field $\F$ with a partition $\mathcal I=\{I_1,I_2,\ldots,I_n\}$ of $\{1,2,\ldots,m\}$, an integer $k$.
\item [Parameter] $k$.
\item [Problem] For each $1\le i\le n$, let $V_i$ be the column space of $M[I_i]$ and let $\V=\{V_i:1\le i\le n\}$ be a subspace arrangement.

Decide whether there exists a permutation $\sigma$ of $\{1,2,\ldots,n\}$ such that
\[
\dim (V_{\sigma (1)}+V_{\sigma(2)}+\cdots+V_{\sigma(i)})\cap (V_{\sigma(i+1)}+V_{\sigma(i+2)}+\cdots+V_{\sigma(n)})\le k
\]
for all $i=1,2,\ldots,n-1$ and if it exists, then output such a permutation $\sigma$.
\end{description}
If such a permutation $\sigma$ exists, then we say that $(M,\mathcal I,k)$ is a YES instance. Otherwise it is a NO instance.
In this case, a sequence $V_{\sigma (1)}$, $V_{\sigma(2)}$, $\ldots$, $V_{\sigma(n)}$ 
is a linear layout of $\V$ having width at most $k$.

In this section, we assume that a branch-decomposition $(T,\L)$ of $\V$ of width at most $\theta$ is given as a part of the input.
By Proposition~\ref{prop:pwbw}, there exists a branch-decomposition of width $\theta\le 2k$ if the input instance has path-width at most $k$. 
Finding such a branch-decomposition and the corresponding running time shall be discussed in detail in Section~\ref{sec:approx}. 
Let us assume that $T$ is a rooted binary tree by picking an arbitrary edge $e$ and subdividing $e$ to create a degree-$2$ vertex called the root.

\subsection{Preprocessing}\label{subsec:preprocess}
Our aim is to obtain a uniform bound on the running time of our algorithm in terms of $n$ for each fixed $k$. However, it is possible that $m$ or $r$ is very large or $V_i$ has huge dimension. Thus, our first step is to preprocess the input so that both $r$ and $m$ are bounded by functions of $k$ and $n$.

Let us say that an $r\times m$ matrix $M$ is \emph{of the standard form} if the submatrix induced by the first $r$ columns is the $r\times r$ identity matrix.  Clearly in every matrix of the standard form, row vectors are linearly independent.
\begin{LEM}[Row Reduction Lemma]\label{lem:rowreduction}
  Let $\F$ be a finite field.
  Given an $r\times m$ matrix $M$ over $\F$ with a partition $\mathcal I=\{I_1,I_2,\ldots,I_n\}$ of $\{1,2,\ldots,m\}$,
  we can find, in time $O(rm^2)$,  an $r'\times m$ matrix $M'$ over $\F$ with a partition $\mathcal I'=\{I_1',I_2',\ldots,I_n'\}$  of $\{1,2,\ldots,m\}$ such that
  \begin{enumerate}[(i)]
  \item $M'$ is of the standard form,
  \item $r'\le r$ and $\abs{I_i'}=\abs{I_i}$ for all $i\in\{1,2,\ldots,n\}$,
  \item For all $k$, $(M,\mathcal I, k)$ is a YES instance with a permutation $\sigma$ of $\{1,2,\ldots,n\}$ if and only if $(M', \mathcal I', k)$ is a YES instance with $\sigma$.
  \end{enumerate}
\end{LEM}
\begin{proof}
  Let $r'$ be the rank of $M$.
  This is easily achieved by applying elementary row operations and removing dependent rows to make the $r'\times r'$ identity submatrix and then permuting columns of $M$ and adjusting $\mathcal I$ accordingly so that the first $r'$ columns form the $r'\times r'$ identity matrix. Note that $\abs{I_i}=\abs{I_i'}$ for all $i$.
\end{proof}

\begin{LEM}[Column Reduction Lemma]\label{lem:columnreduction}
  Let $\F$ be a finite field and let $\theta$ be an integer.
  Let $M$ be an $r\times m$ matrix $M$ over $\F$ of the standard form 
  with a partition $\mathcal I=\{I_1,I_2,\ldots,I_n\}$ of $\{1,2,\ldots,m\}$. 
  Let $V_i$ be the column space of $M[I_i]$ and let $\V=\{V_1,V_2,\ldots,V_n\}$.

  In time $O(\theta rmn)$, we can either find an $r\times m'$ matrix over $\F$ and a partition $\mathcal I'=\{I_1',I_2',\ldots,I_n'\}$ of $\{1,2,\ldots,m'\}$ such that
  \begin{enumerate}[(i)]
  \item the column vectors of $M'[I_i']$ are linearly independent for every $i=1,2,\ldots,n$,
  \item $\abs{I_i'}\le \min(\abs{I_i},\theta)$ for all $i=1,2,\ldots,n$,
  \item For all $k$, $(M,\mathcal I, k)$ is a YES instance with a permutation $\sigma$ of $\{1,2,\ldots,n\}$ if and only if $(M', \mathcal I', k)$ is a YES instance with $\sigma$.
  \end{enumerate}
  or find $V_i\in\V$ such that $\dim (V_i\cap \spn{\V-\{V_i\}})>\theta$.
\end{LEM}
\begin{proof}
  Let $Z=\{1,2,\ldots,r\}$.   Let $\overline{V_i}=\sum_{j\neq i, j\in \{1,2,\ldots,n\}} V_j$
  and $\overline{I_i}=\{1,2,\ldots,m\}-I_i$.

    For this algorithm, we will either find $i$ such that $\dim V_i\cap \overline{V_i}>\theta$ or reduce the instance $M$ and $\mathcal I$ to a smaller instance $M'$ and $\mathcal I'$ whose subspaces $V_i'$ (the column space of $M'[I_i']$) have dimension at most $\theta$ by replacing $V_i$ with $V_i\cap \overline{V_i}$.

  Let $V_i'=V_i\cap \overline{V_i}$ and let $\V'=\{V_1',V_2',\ldots,V_n'\}$.
  We first claim that for a subset $X\subseteq\{1,2,\ldots,n\}$ and $\overline{X}=\{1,2,\ldots,n\}-X$,
  \[ \left(\sum_{i\in X} V_i\right) \cap \left(\sum_{j\in\overline X} V_j\right)
  = \left(\sum_{i\in X} V_i'\right) \cap \left(\sum_{j\in\overline X} V_j'\right).\]
  As $V_i'\subseteq V_i$ for all $i$, it is enough to show that
  $ \left(\sum_{i\in X} V_i\right) \cap \left(\sum_{j\in\overline X} V_j\right)
  \subseteq \left(\sum_{i\in X} V_i'\right) \cap \left(\sum_{j\in\overline X} V_j'\right)$.
  Suppose $a\in  \left(\sum_{i\in X} V_i\right) \cap \left(\sum_{j\in\overline X} V_j\right)$.
  Let $a_i\in V_i$ so that $a=\sum_{i\in X} a_i = \sum_{j\in \overline{X} }a_j$.
  If $i\in X$, then $a_i=\sum_{j\in \overline{X}}a_j -\sum_{j\in X-\{i\}} a_j \in \overline{V_i}$ and therefore $a_i\in V_i'$. Similarly $a_i\in V_i'$ for $i\in \overline{X}$. Thus $a\in \left(\sum_{i\in X} V_i'\right) \cap \left(\sum_{j\in\overline X} V_j'\right)$.
  This claim proves that for each permutation $\sigma$ of $\{1,2,\ldots,n\}$, $V_{\sigma(1)}$, $V_{\sigma(2)}$, $\ldots$, $V_{\sigma(n)}$ is a linear layout of width at most $k$ if and only if
  $V_{\sigma(1)}'$, $V_{\sigma(2)}'$, $\ldots$, $V_{\sigma(n)}'$  is a linear layout of width at most $k$.

  Now we describe how to output an $r\times m'$ matrix $M'$ with a partition $\mathcal I'=\{I_1',I_2',\ldots,I_n'\}$ of $\{1,2,\ldots,m'\}$ such that $V_i'$ is the column space of $M'[I_i']$ for each $i\in \{1,2,\ldots,n\}$ or find $i$ such that $\dim V_i'>\theta$.

  For each $i\in \{1,2,\ldots,n\}$, we find the column basis $X_i\subseteq \overline{I_i}-Z$ of $M[Z\cap I_{i}, \overline{I_i}-Z]$ and the column basis $Y_i\subseteq I_i-Z$ of $M[Z-I_i, I_i-Z]$ and let
  \[
  M_i'=
  \begin{pmatrix}
    M[Z\cap I_i, X_i] &  0 \\
    0 & M[Z-I_i, Y_i]
  \end{pmatrix}.
  \]
  Let
  \[
  M'=\bordermatrix{
     & I_1' & I_2' &\cdots & I_n'\cr
& M_1' & M_2' & \cdots &M_n'  }.
  \]

  We claim that the column space of $M_i'$ is exactly $V_i'=V_i\cap \overline{V_i}$. This can be seen from the following representation of $M$ after permuting rows and columns:
  \[
  M=\bordermatrix{
    & Z\cap I_i &Z-I_i &  I_i- Z & \overline{I_i}-Z \cr
    Z\cap I_i &
    \begin{smallmatrix}
      1 &\\
       & 1\\
       &&\ddots \\
       &&& 1
    \end{smallmatrix}
    &0& \mathbf{A} &  \mathbf{B} \cr
    Z-I_i& 0 &
    \begin{smallmatrix}
      1 &\\
       & 1\\
       &&\ddots \\
       &&& 1
    \end{smallmatrix}& \mathbf{C}
    & \mathbf{D} }.
  \]
  The column spaces of $M[I_i]$ and $M[\overline{I_i}]$ do not change if we replace $\mathbf A$, $\mathbf{D}$ with $0$ and therefore for the computation of $V_i'$, we may assume that $\mathbf{A}=0$ and $\mathbf{D}=0$. Then $V_i'=V_i\cap \overline{V_i}$ is equal to the column space of the matrix
  $
  \left(\begin{smallmatrix}
      0 & \mathbf B \\
      \mathbf C&0
    \end{smallmatrix}\right)$. As we take the column basis $X_i$ of $\mathbf B$ and $Y_i$ of $\mathbf C$, the column space of $M'_i$ is equal to $V_i'$.
  Since $\abs{X_i}+\abs{Y_i}=\dim V_i'$, we can find $i$ with $\dim V_i\cap \spn{\V-\{V_i\}}>\theta$ if $\abs{X_i}+\abs{Y_i}>\theta$.

  Thus our algorithm can output $M'$ with $\mathcal I'=\{I_1',I_2',\ldots,I_n'\}$
  or find $i$ such that $\dim V_i'>\theta$.
  Let us now estimate its running time.
  It takes $O(\theta \abs{Z\cap I_i} \cdot\abs{\overline{I_i}-Z})$ time to compute $X_i$
  and $O(\theta \abs{Z-I_i}\cdot\abs{I_i-Z})$ time to compute $Y_i$
  or verify that $\abs{X_i}+\abs{Y_i}>\theta$
  by applying the elementary row operations. (Note that we only need to apply at most $\theta$ pivots because if the rank is big, then $\dim V_i'$ is big.)
  Since
  \[
    \theta\sum_{i=1}^n\left( \abs{Z\cap I_i}\cdot \abs{\overline{I_i}-Z} +\abs{Z-I_i}\cdot\abs{I_i-Z}\right)
    \le \theta\sum_{i=1}^n r (m-r)\le \theta rmn,
    \]
  the total running time of the algorithm is at most $O(\theta rmn)$.
\end{proof}

We apply the row reduction lemma first, the  column reduction lemma second,  and the row reduction lemma last to reduce the input instance $M$ and $\mathcal I$ to an equivalent instance $M'$ and $\mathcal I'$, 
or find $i$ such that $\dim V_i\cap \spn{\V-V_i}>\theta$, 
in time $O(rm^2)+O(\theta rmn)+O(rm^2)=O(rm(m+\theta n))$.
If there exists $i$ such that $\dim V_i\cap \spn{\V-V_i}>\theta$, then 
we verify that the branch-width of $\V$ is larger than $\theta$.
Thus from now on, we may assume the following for the input instance $M$ and $\mathcal I$.
\begin{itemize}
\item $\abs{I_i}\le \theta$ and the columns of $M[I_i]$ are linearly independent for each $i$.
\item The number $r$ of rows of $M$ is at most $\theta n$.
\end{itemize}
As usual, we let $V_i\subseteq \F^r$ be the column space of $M[I_i]$ and $\V=\{V_1,V_2,\ldots,V_n\}$ be the subspace arrangement given by $M$ and $\mathcal I$. Then by the above assumptions, we have the following.
\begin{itemize}
\item $\dim V_i\le \theta$ for all $i$.
\item $m\le \theta n$.
\item $r\le \theta n$.
\end{itemize}
If $\dim V_i=0$ for all $i$, then the path-width of $\V$ is $0$ and therefore we may assume that $\dim V_i>0$ for some $i$ and the path-width of $\V$ is at least $1$. 
Therefore, we may further assume that $k>0$, $I_i\neq \emptyset$ and $\dim V_i>0$ for all $i$ by discarding such a subspace from $\V$. 

We are going to solve the problem for the subspace arrangement $\V$ given by $M$ and $\mathcal{I}$, which are preprocessed.
Once we find a linear layout $\sigma$ of $\V$ of width at most $k$, it is straightforward to obtain
a linear layout of the original input of width at most $k$.

\subsection{Preparing bases of boundary spaces}\label{subsec:prepare}
Given a branch-decomposition $(T,\L)$ of width at most $\theta$, 
in order to compute the full set $\FS(\V_v,B_v)$ for every node $v$ of $T$,
we first need to figure out $B_v$.
By Proposition~\ref{prop:preprocess-alg},
we can compute bases of the boundary space at each node 
and the sum of two boundary spaces of children at each internal node of $T$ 
in time $O(\theta rmn)$.

\begin{PROP}\label{prop:preprocess-alg}
  Let $\F$ be a fixed finite field and $\theta$ be a nonnegative integer.
  let $M$ be an $r\times m$ matrix of the standard form over $\F$ given with a partition $\mathcal I=\{I_1,I_2,\ldots,I_n\}$ of $\{1,2,\ldots,m\}$. 
  Let $V_i$ be the column space of $M[I_i]$ and $\V=\{V_i: i=1,2,\ldots,n\}$ be a subspace arrangement.
  Given a branch-decomposition $(T,\L)$ of $\V$ having width $\theta$, we can compute the following for all nodes $v$ of $T$, in time $O(\theta r m n)$.
\begin{itemize}
\item A basis $\B_v$ of the boundary space $B_v$.
\item If $v$ has two children $w_1$ and $w_2$, then
  \begin{itemize}
  \item a basis $\B_v'$ of $B_{w_1}+B_{w_2}$ such that $\B_v\subseteq \B_{v}'$,
  \item for each $i\in \{1,2\}$, a transition matrix $T_{w_i}$ from $\B_{w_i}$ to $\B_v'$ such that $T_{w_{i}}\cdot [x]_{\B_{w_i}}=[x]_{\B_v'}$ for all $x\in B_{w_i}$.
  \end{itemize}
\end{itemize}
\end{PROP}
\begin{proof}
  Let $Z=\{1,2,\ldots,r\}$.
  For each node $v$ of $T$, we need to compute the basis $\B_v$ of $B_v$. Let $I_v$ be the union of all $I_i$ such that $\L^{-1}(V_i)$ is $v$ or the descendants of $v$ in $T$. Let $\overline{I_v}=\{1,2,\ldots,m\}-I_v$.
  Then $B_v$ is the intersection of the column spaces of $M[I_v]$ and $M[\overline{I_v}]$.
  Now observe the following representation of $M$ after permuting rows and columns, similar to the proof of Lemma~\ref{lem:columnreduction}:
  \[
  M=\bordermatrix{
    & Z\cap I_v &Z-I_v &  I_v- Z & \overline{I_v}-Z \cr
    Z\cap I_v &
    \begin{smallmatrix}
      1 &\\
       & 1\\
       &&\ddots \\
       &&& 1
    \end{smallmatrix}
    &0& \mathbf{A} &  \mathbf{B} \cr
    Z-I_v& 0 &
    \begin{smallmatrix}
      1 &\\
       & 1\\
       &&\ddots \\
       &&& 1
    \end{smallmatrix}& \mathbf{C}
    & \mathbf{D} }.
  \]
  The column space of $M[I_v]$ is independent of $\mathbf A$ and so we may assume $\mathbf A=0$. Similarly we may assume $\mathbf D=0$ as it does not affect the column space of $M[\overline{I_v}]$. Then $B_v$ is indeed equal to the column space of $
  \left(\begin{smallmatrix}
      0 & \mathbf B\\
      \mathbf C & 0
    \end{smallmatrix}\right)$.
  Furthermore $\dim B_v\le \theta$ and therefore a basis $\mathcal B_v$ of $B_v$ can be found in time $O(\theta \abs{Z\cap I_v}\cdot\abs{\overline{I_v}-Z}+\theta \abs{Z-I_v}\cdot\abs{I_v-Z})$. As $\abs{Z\cap I_v}, \abs{Z-I_v}\le r$, 
 and $\abs{\overline{I_v}-Z}+\abs{I_v-Z}\le m-r\le m$, 
 we deduce that in time $O(\theta r m n)$, we can compute $\B_v$ for all nodes $v$ of $T$.

  If $v$ is an internal node of $T$ with two children $w_1$ and $w_2$, we need to compute a basis $\B'_v$ of $B_{w_1}+B_{w_2}$ such that $\B_v\subseteq \B_v'$.
  Let $M$ be a matrix of $r$ rows and $\abs{\B_{v}}+\abs{\B_{w_1}}+\abs{\B_{w_2}}(\le 3\theta)$ columns of the form
  \[
  M=
  \bordermatrix{
    & P & Q & R \cr
    & \B_v & \B_{w_1} & \B_{w_2}
    }
  ,
  \]
  where each of $\B_v$, $\B_{w_1}$, and $\B_{w_2}$ is considered as a matrix whose column vectors are the vectors in each of those sets.
  By applying the elementary row operations to $M$, we can find, in time $O(\theta r^2)\le O(\theta r m)$,  a column basis $X$ of $M$. 
  Note that $r\le m$ because $M$ is of a standard form.
  As $B_v\subseteq B_{w_1}+B_{w_2}$,  we can choose $X$ such that $P\subseteq X$.
  Then the column vectors of $M[X]$ form a basis $\B'_v$ of $B_{w_1}+B_{w_2}$.
  Since we do this for each internal node $v$, it takes the time $O(\theta rmn)$ 
  to compute $\B_v'$ for all internal nodes $v$ of $T$.

  It remains to compute the transition matrices $T_{w_1}$ and $T_{w_2}$ for each internal node $v$ with two children $w_1$ and $w_2$. Let $i\in \{1,2\}$. Our goal is to find a matrix $T_{w_i}$ such that \[T_{w_i}\cdot [x]_{\B_{w_i}} = [x]_{\B_v'}\] for all $x\in \B_{w_i}$.
  Let $\B_{w_i}=\{b_1,b_2,\ldots,b_\ell\}\subseteq \F^r$. Then if $x=b_j$, then $[x]_{\B_{w_i}}=e_j\in \F^\ell$. Thus the $j$-th column vector of $T_{w_i}$ is equal to the coordinate of $b_j$ with respect to $\B_v'$. In other words, we have
  \[
  (\B_v' )\cdot  T_{w_i} = ( \B_{w_i} )
  \]
  where $T_{w_i}$ is a $\abs{\B_v'}\times \abs{\B_{w_i}}$ matrix.
  This matrix equation can be solved in time $O(\theta r^2 )\le O(\theta rm )$. Thus we can compute $T_{w_1}$ and $T_{w_2}$ for all $v$ in time $O(\theta rmn)$.
\end{proof}

By Proposition~\ref{prop:preprocess-alg}, 
the following information can be computed in time $O(\theta rmn)$.
Therefore, %
we assume that the following are given from Subsection~\ref{subs:ds} to Subsection~\ref{subs:summary}.

\begin{itemize}
\item A branch-decomposition $(T, \L)$ of $\V$ of width at most $\theta$.
\item A basis $\mathcal{B}_v$ of the boundary space $B_v$ at every node $v$ of $T$.
\item A basis $\mathcal{B}_v'$ of $B_v'=B_{w_1}+B_{w_2}$ extending $\mathcal{B}_v$ at every internal node $v$ of $T$ having two children $w_1$ and $w_2$.
\item For each $i\in \{1,2\}$, a transition matrix $T_{w_i}$ from $\mathcal{B}_{w_i}$ to $\mathcal{B}_v'$ such that
\[T_{w_{i}}\cdot [x]_{\mathcal{B}_{w_i}}=[x]_{\mathcal{B}_v'} \text{ for all } x\in B_{w_i}\]  at every internal node $v$ of $T$ having two children $w_1$ and $w_2$.
\end{itemize}

\subsection{Data structure for the full sets}\label{subs:ds}
Before describing the algorithm, we present our data structure to store a $B$-trajectory in the full set $\FS(\V,B)$ when we have a precomputed basis $\mathcal{B}$ of $B$.
For a $B$-trajectory $\Gamma=a_1,a_2,\ldots,a_m$ with $a_i=(L_i,R_i,\lambda_i)$, we need to store subspaces $L_i$ and $R_i$ of $B$.  Let $d=\dim B$. 

If we want to represent a $d'$-dimensional  subspace $S$ of $B$ such as $L_i$ and $R_i$, a naive method is to pick a basis and make a matrix. As our vector space is a subspace of $\F^r$, one might use an $r\times d'$ matrix to represent $S$. However, our $r$ depends on the input and may grow very large even if $\dim B$ is bounded.

In Subsection~\ref{subsec:linear},
we say that a $d\times d'$ matrix $M_{B,S}$ \emph{represents} a $d'$-dimensional subspace $S$ of $B$
if the range of $M_B M_{B,S}$ 
is exactly $S$.
We store subspaces $L_i$ and $R_i$ using a $d\times(\dim L_i)$ matrix $M_{B,L_i}$ and a $d\times(\dim R_i)$ matrix $M_{B,R_i}$. %
This will ensure that each compact $B$-trajectory of width at most $k$ can be stored in  bounded amount of space when $\dim B$ is at most $2\theta$.
Also, by Lemmas~\ref{lem:subspace1} and~\ref{lem:subspace2},
many common operations on subspaces $L_i$ and $R_i$
can be done in time $O(d^3)$ when $M_{B,L_i}$ and $M_{B,R_i}$ are given.
See Figure~\ref{fig:fullset} for an illustration of a full set.

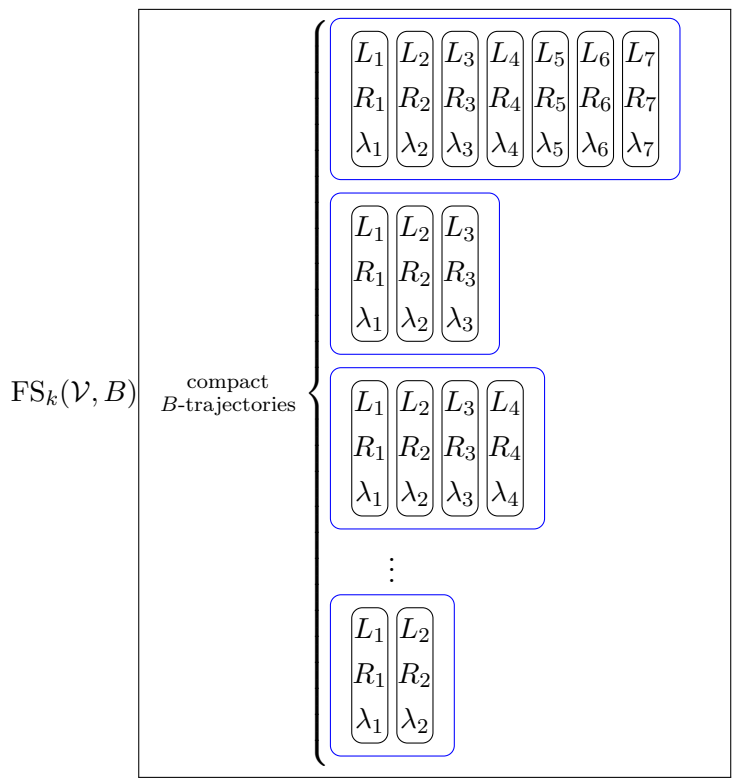
\begin{figure}
  \centering
  \[
    \FS(\V,B)\fbox{
      $\genfrac{}{}{0pt}{1}{\text{compact}}{\text{$B$-trajectories}}
      \begin{cases}
        \begin{tikzpicture}[scale=.6,framed,background rectangle/.style={rounded corners,draw=blue}]
          \foreach \i in {1,...,7} {
            \node at (\i,2) {$L_{\i}$};
            \node at (\i,1) {$R_{\i}$};
            \node at (\i,0) {$\lambda_{\i}$};
            \draw [rounded corners](\i-.4,-.5)--(\i-.4,2.5)--(\i+0.4,2.5)--(\i+0.4,-.5)--cycle;
          }
        \end{tikzpicture}\\
        \begin{tikzpicture}[scale=.6,framed,background rectangle/.style={rounded corners,draw=blue}]
          \foreach \i in {1,...,3} {
            \node at (\i,2) {$L_{\i}$};
            \node at (\i,1) {$R_{\i}$};
            \node at (\i,0) {$\lambda_{\i}$};
            \draw [rounded corners](\i-.4,-.5)--(\i-.4,2.5)--(\i+0.4,2.5)--(\i+0.4,-.5)--cycle;
          }
        \end{tikzpicture}\\
        \begin{tikzpicture}[scale=.6,framed,background rectangle/.style={rounded corners,draw=blue}]
          \foreach \i in {1,...,4} {
            \node at (\i,2) {$L_{\i}$};
            \node at (\i,1) {$R_{\i}$};
            \node at (\i,0) {$\lambda_{\i}$};
            \draw [rounded corners](\i-.4,-.5)--(\i-.4,2.5)--(\i+0.4,2.5)--(\i+0.4,-.5)--cycle;
          }
        \end{tikzpicture}\\
        \quad\quad\vdots\\
        \begin{tikzpicture}[scale=.6,framed,background rectangle/.style={rounded corners,draw=blue}]
          \foreach \i in {1,...,2} {
            \node at (\i,2) {$L_{\i}$};
            \node at (\i,1) {$R_{\i}$};
            \node at (\i,0) {$\lambda_{\i}$};
            \draw [rounded corners](\i-.4,-.5)--(\i-.4,2.5)--(\i+0.4,2.5)--(\i+0.4,-.5)--cycle;
          }
        \end{tikzpicture}
      \end{cases}$
      }
  \]
  \caption{An illustration of a full set: 
    $\lambda_i$ are nonnegative integers and 
    $L_i$, $R_i$ are represented by matrices with $(\dim B)$ rows
    and linearly independent columns, each representing a vector in $\F^r$ by using a basis $\B$ of $B$. }
  \label{fig:fullset}
\end{figure}

\subsection{Computing the full sets}

The procedure \textsc{full-set}$(\V,k,(T,\L))$ aims to construct the full set $\FS(\V,\{0\})$ at the root node of $T$. We compute a set $\FF_v$ at each node $v$, which will be shown to be $\FS(\V_v,B_v)$ later. To this end, \textsc{full-set}$(\V,k,(T,\L))$  recursively chooses a node $v$ of $T$, farthest from the root (ties broken arbitrarily), such that $\FF_v$ is not computed yet.

The core of \textsc{full-set}$(\V,k,(T,\L))$ is the \emph{join} operation at an internal node $v$ which combines the two full sets obtained at its children $w_1$ and $w_2$. For this, the two full sets, which are represented with respect to the individual boundary spaces $B_{w_1}$ and $B_{w_2}$ of $w_1$ and $w_2$, respectively, are adjusted via the \emph{expand} operation. After that, two full sets which are represented with respect to the same subspace $B_{w_1}+B_{w_2}$ become amenable for join. After the join operation, the full set is again adjusted via the \emph{shrink} operation and now represented with respect to the boundary space $B_v$ at node $v$. In this way, all subspaces under consideration have dimension at most $2\theta$.

\begin{algorithm}
  \caption{Constructing the full sets}      \label{alg:fullset}
\begin{algorithmic}[1]
\Procedure{full-set}{$\V,k,(T,\L)$}
\Repeat
\State choose an unmarked node $v$ farthest from the root
\If{$v$ is a leaf} \Comment{initialization}
\State set $\FF_v$ as in Proposition~\ref{prop:single} \label{line:init}
\ElsIf{$v$ is internal with two children $w_1$,$w_2$}
\State recall $B_v'= B_{w_1}+B_{w_2}$
\State $\FF_v^{(i)}\gets \up(\FF_{w_i},B_v')$ for $i=1,2$ \Comment{expand} \label{line:expand}
\State $\FF_v'\gets \up(\FF_v^{(1)}\fplus \FF_v^{(2)},B_v')$ \Comment{join} \label{line:join}
\State $\FF_v\gets \up(\FF_v'|_{B_v},B_v)$ \Comment{shrink}  \label{line:shrink}
\EndIf
\State mark $v$
\Until{all nodes in $T$ are marked}
\EndProcedure
\end{algorithmic}
\end{algorithm}

Recall that $U_k(B)$ is the set of all compact $B$-trajectories of width at most $k$ and that $\abs{U_k(B)}$ is at most
$2^{9d+2}\abs{\F}^{d(d-1)}2^{2(2d+1)k}$
by Lemma~\ref{lem:numtra}, where $d=\dim(B)$. Clearly, $\FS(\V,B) \subseteq U_k(B)$ and therefore we consider the elements of $U_k(B)$ as candidates for $\FS(\V,B)$. To evaluate the number of steps carried out by \textsc{full-set}$(\V,k,(T,\L))$, we need the following lemmas.

\begin{LEM}\label{lem:numpaths}
There are at most $3^{m+n-2}$ distinct lattice paths with diagonal steps from $(1,1)$ to $(m,n)$.
\end{LEM}
\begin{proof}
For every lattice path with diagonal steps $v_1,v_2,\cdots ,v_{\ell}$ from $(1,1)$ to $(m,n)$, we have $\ell\leq m+n-1$. At each $i$, we have three possible choices.
\end{proof}

\begin{LEM}\label{lem:algcomp}
  Let $B$ and $B'$ be subspaces of dimension at most $2\theta$ given with bases $\mathcal B$ and $\mathcal B'$, respectively, such that $B\subseteq B'$. Let $T$ be the transition matrix from $\mathcal B$ to $\mathcal B'$ so that $T\cdot [x]_{\mathcal B}=[x]_{\mathcal B'}$ for all $x\in B$. 
  Then, for a $B$-trajectory $\Delta$ represented with respect to $\mathcal B$
and a $B'$-trajectory $\Gamma$ represented with respect to $\mathcal B'$,
  we can decide $\Delta\tle \Gamma$ in time $\theta^3 (\abs{\Delta}+\abs{\Gamma})3^{\abs{\Delta}+\abs{\Gamma}}\cdot O(1)$.
\end{LEM}
\begin{proof}
  First we transform the representation of $\Delta$ into a representation with respect to $\mathcal B'$. This can be done by multiplying $T$ to each of the matrices representing bases of subspaces appearing in $\Delta$. Since both the number of rows and the number of columns are at most $2\theta$ and each subspace in $\Delta$ has dimension at most $2\theta$, we can transform the representation of $\Delta$ in time $\abs{\Delta} \theta^3 \cdot O(1)$.

  Now we are prepared with $\Delta$ and $\Gamma$, both represented with respect to $\mathcal B'$. In order to test $\Delta\tle\Gamma$, we need to consider all possible lattice paths from $(1,1)$ to $(\abs{\Delta},\abs{\Gamma})$. By Lemma~\ref{lem:numpaths}, there are at most $3^{\abs{\Delta}+\abs{\Gamma}-2}$ such paths. For each path, we need to compare an extension of $\Delta$ with an extension of $\Gamma$. Comparing two subspaces of $B'$ takes time $\theta^3\cdot O(1)$ and there are  $2(\abs{\Delta}+\abs{\Gamma}-1)$ comparisons of subspaces to make for each path. Thus, we can decide $\Delta\tle \Gamma$ in time $\abs{\Delta} \theta^3 \cdot O(1) +\theta^3 (\abs{\Delta}+\abs{\Gamma})3^{\abs{\Delta}+\abs{\Gamma}} \cdot O(1)= \theta^3 (\abs{\Delta}+\abs{\Gamma})3^{\abs{\Delta}+\abs{\Gamma}}\cdot O(1)$.
\end{proof}

\begin{LEM}\label{lem:numcomp}
Let $k$ be an integer.
Let $B$, $B'$ be subspaces such that $B\subseteq B'$ and let $\mathcal{R}$ be a set of $B$-trajectories.
To compute $\up(\mathcal{R},B')$, it is enough to perform at most $\abs{U_k(B')} \cdot \abs{\mathcal{R}}$ comparisons, each deciding whether $\Delta \tle \Gamma$ for some $\Delta\in\mathcal{R}$ and $\Gamma\in U_k(B')$.
\end{LEM}
\begin{proof}
Recall that $\Gamma \in \up(\mathcal{R},B)$ if and only if $\Gamma\in U_k(B)$ and there exists $\Delta\in \mathcal{R}$ such that $\Delta\tle \Gamma$. Therefore, for each pair $\Gamma \in U_k(B)$ and $\Delta \in \mathcal{R}$ we verify whether $\Delta \tle \Gamma$.
\end{proof}

Below, we explain how to perform each operation in \textsc{full-set}$(\V,k,(T,\L))$ 
and argue (in the proof of Proposition~\ref{prop:FSruntime}) 
that each operation takes at most 
$\poly(\theta,\abs{\F},k) \cdot 2^{151 \theta k} \abs{\F}^{12\theta^2}$ steps.

\smallskip

\noindent \textbf{Initialization:} At each leaf $v$ of $T$, which is handled at line~\ref{line:init}, %
Proposition~\ref{prop:single}
 states that the set $\FF_v$ is equal to  
 $\up(\{\Delta\},B_v)$ where $\Delta=(\{0\},B_v,0),(B_v,\{0\},0)$.

By Lemma~\ref{lem:numcomp},
computing $\FF_v$ requires at most 
$\abs{U_k (B_v)}$
comparisons for $\Delta\tle \Gamma$ between $\Delta$ and $\Gamma\in U_k(B_v)$.

\smallskip

\noindent\textbf{Expand:} 
At line~\ref{line:expand}, for $i=1,2$, we compute the set $\up(\FF_{w_i},B_v')$. By Lemma~\ref{lem:numcomp}, this requires at most $\abs{\FF_{w_i}}\cdot \abs{U_k(B_v')}\leq  \abs{U_k(B_{w_i})}\cdot \abs{U_k(B_v')}$ comparisons between $\Delta \in \FF_{w_i}\subseteq U_k(B_{w_i})$ and $\Gamma \in U_k(B_v')$.

\smallskip

\noindent\textbf{Join:} At line~\ref{line:join}, computing $\FF_v'=\up(\FF_v^{(1)} \fplus \FF_v^{(2)},B_v')$, where $B_v'=B_{w_1}+B_{w_2}$, consists of two steps: (a) to compute $\FF_v^{(1)} \fplus \FF_v^{(2)}$, and (b) to take $\up(\FF_v^{(1)} \fplus \FF_v^{(2)},B_v')$.
Note that, for (b), 
the set $\up(\FF_v^{(1)} \fplus \FF_v^{(2)},B_v')$ can be computed in 
at most 
$\abs{U_k(B_v')}^3$ comparisons by Lemma~\ref{lem:numcomp}
since $\abs{\FF_v^{(1)} \fplus \FF_v^{(2)}}\leq \abs{U_k(B_v')}^2$.

\smallskip

\noindent\textbf{Shrink:} 
At line~\ref{line:shrink}, we compute the set $\up(\FF_v'|_{B_v},B_v)$ in two steps: (a) to compute $\FF_v'|_{B_v}$, and (b) to compute $\up(\FF_v'|_{B_v},B_v)$. 
Here, by Lemma~\ref{lem:numcomp}, the set $\up(\FF_v'|_{B_v},B_v)$ can be computed by comparing at most $\abs{U_k(B_v)}\cdot\abs{U_k(B_v)}$ pairs $\Delta$ and $\Gamma$, where $\Delta \in \FF_v'|_{B_v}\subseteq U_k(B_v)$ and $\Gamma \in U_k(B_v)$.

\smallskip

\begin{PROP}\label{prop:FSruntime}
Let $\theta\ge 1$.
The procedure \textsc{full-set}$(\V,k,(T,\L))$ runs in 
$O(\poly(\theta,\abs{\F},k) \cdot 2^{151 \theta k} \abs{\F}^{12\theta^2}\cdot n)$ 
steps, when $\theta$ is the width of $(T,\L)$.
\end{PROP}

\begin{proof}
In each operation, when we compare two trajectories $\Delta$ and $\Gamma$ to decide $\Delta\tle\Gamma$, 
we know, by Lemma~\ref{lem:lentraj},
$\abs{\Delta}\le 2(4\theta+1)(2k+1)\le 28\theta k +2
$ and $\abs{\Gamma}\le (4\theta+1)(2k+1)\le 14\theta k+1$; 
the worst case happens at the join operation. 
Thus, by Lemma~\ref{lem:algcomp}, one comparison can be done in time
$
\theta^3 (\abs{\Delta}+\abs{\Gamma}) 3^{\abs{\Delta}+\abs{\Gamma}} \cdot O(1)
\le
\theta^3 (42\theta k + 3) \cdot 3^{42\theta k +3} \cdot O(1)
\le
\poly(\theta, k)\cdot 3^{42\theta k}
<
\poly(\theta, k) 2^{67 \theta k}.  
$
Note that $3^{42}<2^{67}$.

Recall that $\dim(B_v)\le\theta$, $\dim(B_v')\le 2\theta$, and $\dim(B_{w_i})\le \theta$ for $i=1,2$. 
By Lemma~\ref{lem:numtra}, 
\begin{align*}
\abs{U_k(B_v)}&\le 
2^{9\theta+2}\abs{\F}^{\theta(\theta-1)}2^{2(2\theta+1)k}
\le 2^{15\theta k +2}\abs{\F}^{\theta^2},\\
\abs{U_k(B_v')}&\le 
2^{18\theta+2}\abs{\F}^{2\theta(2\theta-1)}2^{2(4\theta+1)k}
\le 2^{28\theta k +2}\abs{\F}^{4\theta^2}, \text{  and }\\
\abs{U_k(B_{w_i})}&\le 
2^{9\theta+2} \abs{\F}^{\theta(\theta-1)} 2^{2(2\theta+1)k}
\le 2^{15\theta k +2}\abs{\F}^{\theta^2} \text{ for $i=1,2$.}
\end{align*}
When $B$ is a vector space with $\dim (B) \le 2\theta$, Lemma~\ref{lem:numtra} implies that $U_k(B)$ can be generated in 
$\poly(\theta, \abs{\F}, k)\cdot \abs{U_k(B)}
\le \poly(\theta, \abs{\F}, k) \cdot 2^{28\theta k +2}\abs{\F}^{4\theta^2}
$ steps.

Let us analyze the time complexity for each operation.
\smallskip

\noindent\textbf{Initialization:}
Since computing $\FF_v$ requires at most $\abs{U_k(B_v)}$ comparisons, it can be computed in time 
$
\poly(\theta,k) 2^{67\theta k} \cdot 2^{15\theta k +2}\abs{\F}^{\theta^2}
\leq \poly(\theta,k) \cdot  2^{82 \theta k} \abs{\F}^{\theta^2}.
$

\smallskip

\noindent\textbf{Expand:}
The set $\up(\FF_{w_i},B_v')$ can be computed in time
$
\poly(\theta,k)  2^{67 \theta k} \cdot \abs{U_k(B_{w_i})}\cdot \abs{U_k(B_v')}
\le \poly(\theta,k)  2^{67 \theta k} \cdot 2^{15\theta k +2}\abs{\F}^{\theta^2}
\cdot 2^{28\theta k +2}\abs{\F}^{4\theta^2}
\le \poly(\theta,k) \cdot 2^{110 \theta k} \abs{\F}^{5\theta^2}.
$

\smallskip

\noindent\textbf{Join:}
  Recall that $\FF_v^{(1)} \fplus \FF_v^{(2)}=\bigcup_{\Delta_1\in \FF_v^{(1)},\Delta_2 \in \FF_v^{(2)}}\Delta_1\tplus \Delta_2$. Since $\FF_v^{(i)}\subseteq U_k(B_v')$ for $i=1,2$ after an expanding operation, there are at most $\abs{U_k(B_v')}^2$ pairs $\Delta_1$, $\Delta_2$ to consider. 
For each pair $\Delta_1$ and $\Delta_2$, computing $\Delta_1\tplus \Delta_2$ requires considering all lattice paths from $(1,1)$ to $(\abs{\Delta_1},\abs{\Delta_2})$. 
The number of such paths is at most $2^{\abs{\Delta_1}+\abs{\Delta_2}}$ and thus, (a) takes at most $\abs{U_k(B_v')}^2\cdot 2^{\abs{\Delta_1}+\abs{\Delta_2}}$ steps.
By Lemma~\ref{lem:lentraj},  
$\abs{U_k(B_v')}^2\cdot 2^{\abs{\Delta_1}+\abs{\Delta_2}}\le
(2^{28\theta k +2}\abs{\F}^{4\theta^2})^2 \cdot 
2^{2(4\theta+1)(2k+1)}
\le 2^{84 \theta k+6} \abs{\F}^{8\theta^2}.
$

For (b), the set $\up(\FF_v^{(1)} \fplus \FF_v^{(2)},B_v')$ can be constructed in time
$
\poly(\theta,k) 2^{67 \theta k} \cdot \abs{U_k(B_v')}^3
\leq \poly(\theta,k) 2^{67 \theta k}
\cdot (2^{28\theta k +2}\abs{\F}^{4\theta^2})^3
\leq \poly(\theta, k) \cdot 2^{151 \theta k} \abs{\F}^{12\theta^2}.
$ Therefore, the join operation can be performed in $\poly(\theta, k) \cdot 2^{151 \theta k} \abs{\F}^{12\theta^2}$ steps.

\smallskip

\noindent\textbf{Shrink:}
For (a), we consider every $\Delta \in \FF_v'$ and take $\Delta|_{B_v}$. 
For each $\Delta \in \FF_v'$, computing $\Delta|_{B_v}$ can be done in time 
$\theta^3 \abs{\Delta}\cdot O(1)$
(see Subsection~\ref{subsec:linear}). Since $\FF_v'\subseteq U_k(B_v')$, there are at most $\abs{U_k(B_v')}$ elements to consider. 
By Lemma~\ref{lem:lentraj}, $\abs{\Delta}\le(4\theta+1)(2k+1)$.
It follows that (a) can be done in 
$\theta^3 (4\theta+1)(2k+1) \cdot \abs{U_k(B_v')}\cdot O(1)\le \poly(\theta,k) \cdot 2^{28\theta k}\abs{\F}^{4\theta^2}
$ steps. 

To compute $\up(\FF_v'|_{B_v},B_v)$, we need to compare two trajectories at most 
$\abs{U_k(B_v)}\cdot\abs{U_k(B_v)}$ times.
Thus, (b) takes at most 
$
\poly(\theta,k) 2^{67\theta k} \cdot \abs{U_k(B_v)}\cdot\abs{U_k(B_v)}
\leq \poly(\theta,k) 2^{67\theta k} 
\cdot (2^{15\theta k +2}\abs{\F}^{\theta^2})^2
\leq \poly(\theta,k) \cdot 2^{97\theta k}\abs{\F}^{2\theta^2}
$
steps. Hence the shrink operation can be done in $\poly(\theta,k) \cdot 2^{97\theta k}\abs{\F}^{4\theta^2}$ steps.

\smallskip

Therefore, the running time of one iteration of the repeat-loop is
$
\poly(\theta, \abs{\F}, k) \cdot 2^{28\theta k +2}\abs{\F}^{4\theta^2}
+
\poly(\theta,k) \cdot 2^{82 \theta k} \abs{\F}^{\theta^2}
+
\poly(\theta,k) \cdot 2^{110 \theta k} \abs{\F}^{5\theta^2}
+
\poly(\theta, k) \cdot 2^{151 \theta k} \abs{\F}^{12\theta^2}
+
\poly(\theta,k) \cdot 2^{97\theta k}\abs{\F}^{4\theta^2}
\le 
\poly(\theta,\abs{\F},k) \cdot 2^{151 \theta k} \abs{\F}^{12\theta^2}
$.
Notice that the tree $T$ contains $n$ leaf nodes and $n-1$ internal nodes. 
This completes the proof.
\end{proof}

\begin{PROP}\label{prop:FSroot}
The path-width of $\V$ is at most $k$ if and only if $\FF_{root}\neq \emptyset$ at the root node of $T$.
\end{PROP}
\begin{proof}
First, we claim that for each node $v$ of $T$, we have $\FF_v=\FS(\V_v,B_v)$, and furthermore $\FF_v'=\FS(\V_v,B_v')$ and $\FF_v^{(i)}=\FS(\V_{w_i},B_v')$ for $i=1,2$ when $v$ has two children $w_1$ and $w_2$. We prove this by induction on the number of steps executed by \textsc{full-set}$(\V,k,(T,\L))$. The claim holds for  $\FF_v$ produced at line~\ref{line:init}
by Proposition~\ref{prop:single}.
At line~\ref{line:expand}, Lemma~\ref{lem:subspace} implies $\spn{\V_{v}}\cap (B_{w_1}+B_{w_2})=B_{v}$. Since the condition of Proposition~\ref{prop:expand} holds, it follows that $\FF_v^{(i)}=\FS(\V_{w_i},B_v')$ for $i=1,2$. For $\FF_v'$ produced at line~\ref{line:join}, notice that $(\spn{\V_{w_1}}+B_{w_1}+B_{w_2})\cap (\spn{\V_{w_2}}+B_{w_1}+B_{w_2})=B_{w_1}+B_{w_2}$ by Lemma~\ref{lem:subspace}. Hence, the condition of Proposition~\ref{prop:join} is satisfied. 
By the induction hypothesis and Proposition~\ref{prop:join}, 
we have $\FF_v'=\FS(\V_v,B_v')$. At line~\ref{line:shrink}, observe that $B_v\subseteq B_{w_1}+B_{w_2}$ by Lemma~\ref{lem:subspace}. The condition of Proposition~\ref{prop:shrink} holds, which implies $\FF_v=\FS(\V_v,B_v)$. This completes the proof of our claim.

From the above claim, it immediately follows that the forward implication holds. For the opposite direction, let $\Gamma \in \FF_{root}=\FS(\V_{root},B_{root})$ be a compact $B_{root}$-trajectory of width at most $k$. Note that $\V_{root}=\V$ and $B_{root}=\{0\}$. By definition, there exists a $\{0\}$-trajectory $\Delta$ realizable in $\V$ such that $\Delta \tle \Gamma$. 
This means that there exists a linear layout $\sigma=V_1, V_2,\ldots, V_m$ of $\V$
whose canonical $\{0\}$-trajectory $a_1,a_2,\ldots , a_{m+1}$ is an extension of $\Delta$. Since $\Gamma$ is of width at most $k$ and $\Delta \tle \Gamma$, Lemma~\ref{lem:width} implies that $\Delta$ is indeed of width at most $k$. Recall that $L(a_i)$ and $R(a_i)$ are subspaces of $\{0\}$, and we have
\begin{align*}
  \lambda(a_i)&=\dim (V_1+V_2+\cdots+V_{i-1})\cap (V_{i}+V_{i+1}+\cdots+V_m)\\
  &\quad\quad-\dim L(a_i)\cap R(a_i)\\
  &= \dim (V_1+V_2+\cdots+V_{i-1})\cap (V_{i}+V_{i+1}+\cdots+V_m)
\end{align*}
for every $i=1,\ldots, m+1$. Therefore, 
$\sigma$ is a linear layout of $\V$ of width at most $k$.
\end{proof}

\subsection{Backtracking to construct a linear layout}\label{subsec:backtracking}
In this subsection, we illustrate how to construct an optimal linear layout of $\V$ if  $\FS(\V,\{0\})$ is nonempty.
First let $T'$ be a tree obtained from $T$ as follows:
\begin{itemize}
\item The vertex set of $T'$ is given as the disjoint union of
$\{ v^1, v^2, v^J, v^{U}, v^{S},v\}$ for each internal node $v$ of $T$
and $\{v^L,v\}$ for each leaf node $v$ of $T$.
\item For each internal node $v$ of $T$, $v^1v^J$, $v^2v^J$, $v^Jv^U$, $v^Uv^S$, $v^Sv$ are edges of $T'$ induced on $\{ v^1, v^2, v^J, v^{U}, v^{S},v\}$.
\item For each leaf node $v$ of $T$, $v^Lv$ is an edge of $T'$.
\item For each internal node $v$ of $T$ and its two children $w_1$ and $w_2$, $T'$ has edges $w_1v^1$ and $w_2v^2$.
\end{itemize}
This tree $T'$ is called the \emph{composition tree} of $T$.
See Figure~\ref{fig:compositiontree}.
    \begin{figure} %
      \center
      \includegraphics[scale=0.88]{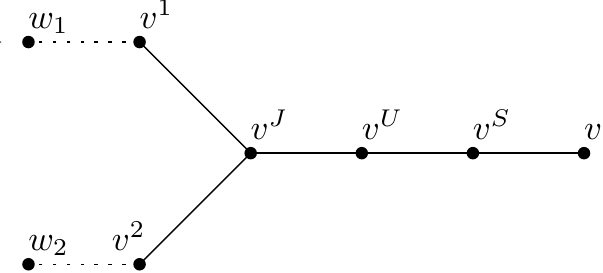}
      \caption{Subtrees of $T'$ for an internal node $v$ with two children $w_1$, $w_2$.}
      \label{fig:compositiontree}
    \end{figure}
There are four kinds of nodes in $T'$.
\begin{itemize}
\item
  Vertices  $v^U$, $v^1$, $v^2$ of $T'$ for all internal nodes $v$ of $T$
and vertices $v$ of $T'$ for all nodes $v$ of $T$ are \emph{up nodes}.
\item
  Vertices $v^J$ of $T'$ for all internal nodes $v$ of $T$ are \emph{join nodes}.
\item Vertices $v^L$ of $T'$ for all leaf nodes $v$ of $T$ are  \emph{leaf nodes}.
\item Vertices $v^S$ of $T'$ for all internal nodes $v$ of $T$ are called \emph{shrink nodes}.
\end{itemize}

The composition tree $T'$ describes how our dynamic programming works as follows. 
The nodes of $T'$ can be grouped so that 
each group corresponds to the computational steps 
for a full set
in Algorithm~\ref{alg:fullset}.
For a leaf node $v$ of $T$, $v^L$ and $v$ correspond to Line~\ref{line:init}.
For a internal node $v$ of $T$,
$v^1$ and $v^2$ correspond to Line~\ref{line:expand},
$v^J$ and $v^U$ correspond to Line~\ref{line:join}, and 
$v^S$ and $v$ correspond to Line~\ref{line:shrink}.

We choose $\Gamma\in \FS(\V,\{0\})$ with the minimum width. 
The width of $\Gamma$ is going to be equal to the path-width of $\V$.
By backtracking how $\Gamma$ is placed in $\FS(\V,\{0\})$ in our algorithm, we can label each node $x$ of $T'$ by a $B_x$-trajectory $\Gamma_x$ for  some $B_x$ and other necessary information satisfying the following.
\begin{itemize}
\item $B_x=B_v$ if $x=v$, $x=v^S$, or $x=v^L$, and $B_x=B_v'$ if $x=v^U$,  $x=v^J$, $x=v^1$, or $x=v^2$ for some node  $v$ of $T$.
\item If $x$ is a shrink node with a child $y$, then $\Gamma_x=\Gamma_y|_{B_x}$.
\item If $x$ is a up node with a child $y$, then $\Gamma_y\tle \Gamma_x$ and there is  a lattice path $P_x$ with diagonal steps as given in Lemma~\ref{lem:comparison}. Then there is a sequence $1= x_1\le x_2\le \cdots\le x_{\abs{\Gamma_x}}=\abs{\Gamma_y}$ of integers such that $P_x$ goes through points $(x_j,j)$ for all $j$.  This sequence $x_1,x_2,\ldots,x_{\abs{\Gamma_x}}$ is also stored at $x$.
\item If $x$ is a join node with two children $y$ and $z$, then $\Gamma_x=\Gamma_y +_{P_x} \Gamma_z$ for some lattice path $P_x$. This lattice path $P_x$ is stored at the node $x$.
\item If $x$ is a leaf node $v^L$ for a leaf node $v$ of $T$, then
  \[\Gamma_x=
  \begin{cases}
(\{0\},B,0),(B,\{0\},0) & \text{ if }B_v\neq \{0\},\\
(\{0\},\{0\},0)&\text{ if }B_v=\{0\}. \end{cases}\]

\item If $x$ is the root node of $T'$, then $\Gamma_x=\Gamma\in \FS(\V,\{0\})$.
\end{itemize}
We modify \textsc{full-set}$(\V,k,(T,\L))$ so that whenever an operation is carried out, for every $\Gamma$ which is placed in $\FF_v$, $\FF_v'$, $\FF_v^{(1)}$, or $\FF_v^{(2)}$, a \emph{certificate} for $\Gamma$ is stored so that later when we find $\Gamma\in \FS(\V,\{0\})$, we can construct the \emph{labeled composition tree} $T'$ by backtracking.

\begin{algorithm}
  \caption{Print a linear layout of width at most $k$}      \label{alg:order}
\begin{algorithmic}[1]
\Procedure{printorder}{$x,i$}
\If{$x$ is a leaf and $i=1$}
\State{print the space $V$ in $\V$ such that $(\L^{-1}(V))^L=x$}
\ElsIf{$x$ is a join node with two children $y$ and $z$}
\State{assume that $P_x=v_1,v_2,\ldots, v_t$. Let $v_i=(x_i,y_i)$}
\If{$v_{i+1}-v_i=(1,0)$}
\State{call \textsc{printorder}($y$, $x_i$)}
\Else
\State{call \textsc{printorder}($z$, $y_i$)}
\EndIf
\ElsIf{$x$ is a up node with a child $y$}
\State{recall the sequence $1= x_1\le x_2\le \cdots\le x_{\abs{\Gamma_x}}=\abs{\Gamma_y}$ so that the lattice path $P_x$ with diagonal steps goes through points $(x_j,j)$}
\State{call \textsc{printorder}($y$, $j$) for all $j=x_i,x_i+1,\ldots,x_{i+1}-1$}
\Else
\State{for a child $y$ of $x$, call \textsc{printorder}($y$, $i$)}
\EndIf
\EndProcedure
\Procedure{order}{}
\State{print all $V\in \V$ if $B_x=\{0\}$ and $\L(x)=V$}
\State{call \textsc{printorder}($x$,$i$) for the root $x$ of $T'$ and $1\le i<\abs{\Gamma_x}$}
\EndProcedure
\end{algorithmic}
\end{algorithm}

\begin{PROP}\label{prop:timeorder}
  The algorithm \textsc{order} in Algorithm~\ref{alg:order} correctly finds a linear layout of $\V$ of width at most $k$ in time $O(\theta k n)$ when a labeled composition tree $T'$ obtained from a branch-decomposition $(T,\L)$ of width at most $\theta$ and a $\{0\}$-trajectory $\Gamma\in \FS(\V,\{0\})$ are given.
\end{PROP}
\begin{proof}
We sketch the proof for the correctness.
Consider a $B$-trajectory $\Gamma=a_1, a_2, \ldots, a_{m+1}$ realizable in a subspace arrangement $\V$.
Then there exists the canonical $B$-trajectory $\Gamma'=a_1', a_2', \ldots, a_{n+1}'\in E(\Gamma)$ 
of a linear layout $V_{\sigma(1)}$, $V_{\sigma(2)}$, $\ldots$, $V_{\sigma(n)}$ of $\V$.
By definition, we can say that $V_i$ is between $a_i'$ and $a_{i+1}'$.
Since $\Gamma'$ is an extension of $\Gamma$,
we can naturally say that $V_i$ is between $a_j$ and $a_{j+1}$ 
if $a_i'=a_j$ and $a_{i+1}'\in \{a_j, a_{j+1}\}$.
Let $X_j=\{V_i:V_i\text{ is between }a_j\text{ and }a_{j+1}\text{ for all }1\le i\le n\}$ for all $1\le j\le m$.
Then $X_1$, $X_2$, $\ldots$, $X_m$ is a partition of $\V$.
Thus, we can regard each realizable $B$-trajectory to be equipped with such an ordered partition of $\V$.

For each $\Gamma$ in the full set,
there is a realizable $\Delta$ with $\Delta\tle\Gamma$ 
and so we can also give such an ordered partition to $\Gamma$.
The only place where we mix two ordered partitions is at the join node of $T'$
and we have the complete description how to combine linear layouts by Lemma~\ref{lem:Psumreal}.
(For a leaf $x$ with $B_x=\{0\}$ and $\L(x)=V$, it does not matter where this subspace $V\in \V$ is placed in the linear layout and so we simply put it in the beginning of the linear layout.)

Now let  us discuss the running time. Each node $x$ of $T'$ is visited $\abs{\Gamma_x}-1$ times because $1\le i<\abs{\Gamma_x}$ is always ensured in \textsc{printorder}. If a node $x$ of $T'$ is neither a join node nor a shrink node, then $x$ is labeled by a compact $B_x$-trajectory and $\dim B_x\le 2\theta$ and therefore $\abs{\Gamma_x}\le (4\theta+1)(2k+1)$ by Lemma~\ref{lem:lentraj}.

If $x$ is a join node with two children $y$ and $z$, then $\Gamma_x$ might not be compact, but $\Gamma_x=\Gamma_y+_{P_x} \Gamma_z$ for a compact $B_y$-trajectory $\Gamma_y$ and a compact $B_z$-trajectory $\Gamma_z$. Since $\dim B_y\le \theta$ and $\dim B_z\le \theta$, both $\Gamma_y$ and $\Gamma_z$ have length at most $(2\theta+1)(2k+1)$. Then $\Gamma_x$ has length at most $2(2\theta+1)(2k+1)-1$.

If $x$ is a shrink node with  a child $y$, then $\Gamma_x=\Gamma_y|_{B_x}$ and $\Gamma_y$ is a $B_y$-trajectory with $\dim B_y \le 2\theta$ and so $\abs{\Gamma_x}=\abs{\Gamma_y}\le (4\theta+1)(2k+1)$.

In all cases, we deduce that $\abs{\Gamma_x}=\theta k \cdot O(1)$. Thus the algorithm visits each node of $T'$ at most $\theta k \cdot O(1)$ times and $T'$ has $O(n)$ nodes and therefore the running time of this algorithm is bounded by $O(\theta k n)$.
\end{proof}

\subsection{Summary}\label{subs:summary}
The following theorem summarizes the procedure and the time complexity for an algorithm 
that solves the problem introduced at the beginning of Section~\ref{sec:algorithm}.
\begin{THM}\label{thm:timecomplexity}
  Let $k$, $\theta$, and $r$ be nonnegative integers and let $\F$ be a fixed finite field.
  Let $\V=\{V_1,V_2,\ldots,V_n\}$ be a subspace arrangement of subspaces of $\F^r$ and $(T,\L)$ be a branch-decomposition of $\V$ of width at most $\theta$.  
  One can find in time $O(rm(m+\theta n)+\poly(\theta,\abs{\F},k) \cdot 2^{151 \theta k} \abs{\F}^{12\theta^2} n)$
  a linear layout of $\V$ of width at most $k$, 
  or confirm that no such linear layout exists,
  where each $V_i$ is given by its spanning set of $d_i$ vectors and $m=\sum_{i=1}^n d_i$.
\end{THM}
We remark that the algorithm can output a linear layout of the minimum width if there is a linear layout of width at most $k$, by choosing a $\{0\}$-trajectory of the minimum width in the full set. 
\begin{proof}
  First, apply the preprocessing steps discussed in Subsection~\ref{subsec:preprocess} using Lemmas~\ref{lem:rowreduction} and \ref{lem:columnreduction}
  in time $O(rm(m+\theta n))$ in order to simplify the input.   
  If we find $V\in \V$ with $\dim(V\cap \spn{\V-\{V\}})>k$ in Lemma~\ref{lem:columnreduction},
  then we confirm that no such linear layout exists.
  For convenience, 
  let $\V$ denote the subspace arrangement after the preprocessing and let $\abs{\V}=n$.
  If $n\le 1$, then the path-width of $\V$ is $0$ and so we can output an arbitrary linear layout when $k\ge 0$. 
  Thus we may assume that $n\ge 2$. Also, we may assume that 
  $k>0$  and $\theta>0$ as described in Subsection~\ref{subsec:preprocess}.

  Then we run the algorithm due to Proposition~\ref{prop:preprocess-alg} 
  in time $O(\theta rmn)$ so that we are ready to run Algorithm~\ref{alg:fullset}.
  
  By Proposition~\ref{prop:FSruntime}, Algorithm~\ref{alg:fullset} can be executed in $O(\poly(\theta,\abs{\F},k) \cdot 2^{151 \theta k} \abs{\F}^{12\theta^2}\cdot n)$ steps. If Algorithm~\ref{alg:fullset} outputs an empty full set for the root node, then 
no linear layout of width at most $k$ exists.
  If Algorithm~\ref{alg:fullset} outputs a nonempty full set for the root node, 
  then we construct a linear layout of width at most $k$ using Algorithm~\ref{alg:order}, which takes time $O(\theta kn)$ by Proposition~\ref{prop:timeorder}.
  When Algorithm~\ref{alg:order} gives a linear layout of width at most $k$,
  we can easily convert it into a linear layout of the input subspace arrangement (before the preprocessing)
  in linear time.

  Therefore, the total time complexity is $O(rm(m+\theta n)+\theta r m n+\poly(\theta,\abs{\F},k) \cdot 2^{151 \theta k} \abs{\F}^{12\theta^2}\cdot n + \theta k n)=O(rm(m+\theta n)+\poly(\theta,\abs{\F},k) \cdot 2^{151 \theta k} \abs{\F}^{12\theta^2}\cdot n)$.
\end{proof}

\section{How to provide an approximate branch-decomposition}\label{sec:approx}
Up to the previous section, we assumed that we are given a branch-decomposition of width at most $\theta$ 
upon which our algorithm is based.
But how do we obtain such a branch-decomposition in the beginning?
We present two methods to accomplish this goal.
We assume that an input subspace arrangement contains at least two subspaces
otherwise it has path-width $0$.

\subsection{Iterative compression}\label{subsec:compression}
Our first method is to use our algorithm iteratively. This strategy is sometimes called iterative compression; roughly speaking, for a subspace arrangement $\V=\{V_i\}_{i=1}^n$ and $2\le \ell<n$,
we will iteratively apply Algorithms~\ref{alg:fullset} and~\ref{alg:order} to a subset $\{V_i\}_{i=1}^\ell$ of $\V$ and then modify the output to obtain a branch-decomposition of small width of $\{V_i\}_{i=1}^{\ell+1}$, which is needed for the next step. 
This requires to run Algorithms~\ref{alg:fullset} and~\ref{alg:order} $O(n)$ times and has an advantage of making this paper self-contained.

For this approach, we need the following simple lemma whose trivial proof we omit.
A tree is a \emph{caterpillar tree} if it has a path $P$ such that
every vertex is either a vertex in $P$ or is adjacent to some vertex in $P$.
\begin{LEM}\label{lem:compression}
  Let $\V=\{V_i\}_{i=1}^n$ be a subspace arrangement over $\F$ with  $n\ge 3$
  and $\dim V_i\le 2k$ for all $i=1,2,\ldots,n$.
  If $V_1,V_2,\ldots,V_{n-1}$ is a linear layout of width at most $k$, then
  $V_1,V_2,\ldots,V_n$ is a linear layout of width at most $3k$.
  \end{LEM}
We argue that for an input subspace arrangement $\V_{\text{in}}$ of subspaces of $\F^r$, in time $O(r(\abs{\V_{\text{in}}}^3+m^2)+ f(k)\abs{\V_{\text{in}}}^2)$ we can either find a linear layout of width at most $k$ or confirm that no such linear layout exists, where each subspace in $\V_{\text{in}}$ is given by its spanning set of $d_i$ vectors and $m=\sum_{i=1}^{\abs{\V_{\text{in}}}} d_i$. 
Given a subspace arrangement $\V_{\text{in}}$ of subspaces of $\F^r$, 
we run the preprocessing described in Subsection~\ref{subsec:preprocess} with $\theta=2k$.
Note that the preprocessing takes time $O(rm(m+ 2k\abs{\V_{\text{in}}}))=O(rm^2)$. 
During the preprocessing, the algorithm due to Lemma~\ref{lem:columnreduction} may find a subspace $V\in\V_{\text{in}}$ such that $\dim(V\cap\spn{\V_{\text{in}}-\{V\}})>2k$.
If so, we can confirm that $\V_{\text{in}}$ has path-width larger than $k$ by Proposition~\ref{prop:pwbw}.
Otherwise, let $\V$ be the resulting subspace arrangement with $\abs{\V}\le\abs{\V_{\text{in}}}$.
If $\V$ contains at most one element, then every linear layout of $\V$ has width $0$
and thus every linear layout of $\V_{\text{in}}$ has width $0$.
Thus, we may assume that $\V=\{V_1,V_2,\ldots, V_{n}\}$ contains at least two subspaces, that is, $2\le n\le \abs{\V_{\text{in}}}$.

Now we iterate over $\ell=2,\ldots,n$ and we may assume that 
we are given a linear layout $\sigma_{\ell-1}$ of width at most $k$ for a subspace arrangement $\{V_1,\ldots, V_{\ell-1}\}$.
If this is not the case, we declare that $\V$ allows no linear layout of width at most $k$ 
and terminate.
For the base case when $\ell=2$, we have the unique linear layout $\sigma_1$ of $\{V_1\}$ 
whose width is $0$. 
By Lemma~\ref{lem:compression}, we can construct a linear layout of width at most $3k$ 
for a subspace arrangement $\{V_1,\ldots, V_{\ell}\}$.
This can be converted into a branch-decomposition $(T,\L)$ of 
$\{V_1,\ldots,V_{\ell}\}$ whose width is at most $3k$ by taking a subcubic caterpillar tree $T$ 
whose leaves are mapped to $V_1$, $V_2$, $\ldots$, $V_{\ell}$ following the order of the linear layout.
And then, 
we run the algorithm due to Proposition~\ref{prop:preprocess-alg}
to compute bases of the boundary spaces.
When we prepare a branch-decomposition $(T,\L)$ with bases of the boundary spaces,
we run Algorithm~\ref{alg:fullset} with $\{V_1,\ldots,V_{\ell}\}$, $k$, and $(T,\L)$.
If the computed $\mathcal{F}_{root}$ is empty, then by Proposition~\ref{prop:FSroot}, the path-width of $\V$ is larger than $k$ and therefore the path-width of the input $\V_{\text{in}}$ is also larger than $k$.
Otherwise we run Algorithm~\ref{alg:order} with an element of $\mathcal{F}_{root}$ having the minimum width to obtain an output linear layout $\sigma_\ell$ whose width is at most $k$.
Until $\ell=n$, we repeat this process by increasing $\ell$. 
Once $\ell$ is equal to $n$, 
if it is not terminated yet, 
then we obtain a linear layout $\sigma_n$ of $\V$ of width at most $k$, that can be converted to a linear layout of $\V_{\text{in}}$ of width at most $k$
by Lemmas~\ref{lem:rowreduction} and~\ref{lem:columnreduction}.
See Figure~\ref{fig:iterative} for a flowchart of the algorithm using an iterative compression.

\begin{figure}
\small
\centering
\tikzstyle{startstop} = [rectangle, rounded corners, minimum width=3cm, minimum height=0.5cm,text centered, draw=black, fill=red!30]
\tikzstyle{io} = [trapezium, trapezium left angle=70, trapezium right angle=110, minimum width=3cm, minimum height=1cm, text centered, text width=2.4cm, draw=black, fill=blue!30]
\tikzstyle{io3} = [trapezium, trapezium left angle=70, trapezium right angle=110, minimum width=3cm, minimum height=1cm, text centered, text width=2cm, draw=black, fill=blue!30]
\tikzstyle{io2} = [trapezium, trapezium left angle=70, trapezium right angle=110, minimum width=3cm, minimum height=1cm, text centered, text width=5.4cm, draw=black, fill=blue!30]
\tikzstyle{io1} = [trapezium, trapezium left angle=70, trapezium right angle=110, minimum width=3cm, minimum height=0.6cm, text centered, text width=8cm, draw=black, fill=blue!30]
\tikzstyle{process} = [rectangle, minimum width=3cm, minimum height=0.6cm, text centered, text width=10cm, draw=black, fill=orange!30]
\tikzstyle{process3} = [rectangle, minimum width=3cm, minimum height=1cm, text centered, text width=9.2cm, draw=black, fill=orange!30]
\tikzstyle{process4} = [rectangle, minimum width=1cm, minimum height=1cm, text centered, text width=1.5cm, draw=black, fill=orange!30]
\tikzstyle{process1} = [rectangle, minimum width=3cm, minimum height=0.6cm, text centered, text width=8.3cm, draw=black, fill=orange!30]
\tikzstyle{decision} = [diamond, aspect=5, minimum width=2cm, minimum height=0.3cm, text centered, text width=3.5cm, draw=black, fill=green!30]
\tikzstyle{decision1} = [diamond, aspect=5, minimum width=2cm, minimum height=0.5cm, text centered, text width=2cm, draw=black, fill=green!30]
\tikzstyle{decision2} = [diamond, aspect=4, minimum width=2cm, minimum height=0.5cm, text centered, text width=5.6cm, draw=black, fill=green!30]
\tikzstyle{decision3} = [diamond, aspect=5, minimum width=2cm, minimum height=0.5cm, text centered, text width=4.7cm, draw=black, fill=green!30]
\tikzstyle{arrow} = [thick,->,>=stealth]

\begin{tikzpicture}[node distance=2cm]
\node (start) [startstop] {Start};
\node (in1) [io1, below of=start, yshift=1cm] {Let $\V_{\text{in}}$ be an input subspace arrangement of $\F^r$.};
\node (pro1) [process, below of=in1, yshift=0.9cm] {Run the preprocessing described in Subsection~\ref{subsec:preprocess} with $\theta=2k$.};
\node (dec3) [decision2, below of=pro1] {Does the preprocessing confirm that the path-width is larger than $k$?};
\node (pro7) [process1, below of=dec3, yshift=-0.3cm] {
Let $\V=\{V_1,V_2,\ldots,V_n\}$ be the resulting subspace arrangement.};
\node (dec4) [decision1, below of=pro7, yshift=0.6cm] {Is $\abs{\V}\le1$?};
\node (out4) [io3, right of=dec4, xshift=5.5cm] {Output an arbitrary linear layout of $\V_{\text{in}}$.};
\node (pro2) [process1, below of=dec4,yshift=0.7cm] {
Let $\ell=2$. Let $\sigma_1$ be the unique
linear layout of $\{V_1\}$.};
\node (pro5) [process3, below of=pro2, yshift=0.25cm] {Convert the linear layout $\sigma$ of $\{V_i\}_{i=1}^{\ell-1}$ to a branch-decomposition $(T,\L)$ of $\{V_i\}_{i=1}^{\ell}$ of width at most $3k$ using Lemma~\ref{lem:compression}. 
Run the algorithm due to Proposition~\ref{prop:preprocess-alg}.
Run Algorithm~\ref{alg:fullset} with $\{V_i\}_{i=1}^{\ell}$, $k$, and $(T,\L)$.};
\node (dec1) [decision, below of=pro5] {Is $\mathcal{F}_{root}$ an empty set?};
\node (pro4) [process3, below of=dec1, yshift=0.2cm] {Run Algorithm~\ref{alg:order} with an element of $\mathcal{F}_{root}$ with the minimum width to obtain an output linear layout $\sigma_\ell$ of width at most $k$.};
\node (dec2) [decision1, below of=pro4, yshift=0.4cm] {Is $\ell=n$?};
\node (out1) [io2, below of=dec2, yshift=0.5cm] {Output the linear layout of $\V_{\text{in}}$ of width at most $k$
obtained from $\sigma_n$.};
\node (stop) [startstop, below of=out1, yshift=0.8cm] {Stop};
\node (out2) [io, right of=out1, xshift=3.3cm] {Confirm that the path-width of $\V_{\text{in}}$ is larger than $k$.};
\node (pro6) [process4, right of=dec1, xshift=4.4cm] {Increase $\ell$ by $1$.};

\draw [arrow] (start) -- (in1);
\draw [arrow] (in1) -- (pro1);
\draw [arrow] (pro1) -- (dec3);
\draw [arrow] (dec3) -- node[anchor=west] {No} (pro7);
\draw [arrow] (pro7) -- (dec4);
\draw [arrow] (dec4) -- node[anchor=west] {No} (pro2);
\draw [arrow] (dec3) -| node[anchor=west, yshift=-0.6cm] {Yes} (out2);
\draw [arrow] (pro2) -- (pro5);
\draw [arrow] (pro5) -- (dec1);
\draw [arrow] (dec1) -- node[anchor=west] {No} (pro4);
\draw [arrow] (pro4) -- (dec2);
\draw [arrow] (dec2) -- node[anchor=west] {Yes} (out1);
\draw [arrow] (out1) -- (stop);
\draw [arrow] (out4) |- (1.51,-19.05);
\draw [arrow] (dec4) -- node[anchor=south] {Yes} (out4);
\draw [arrow] (dec1) -| node[anchor=south, xshift=-1.2cm] {Yes} (5,-16.85);
\draw [arrow] (out2) |- (1.51,-18.85);
\draw [arrow] (dec2) -| node[anchor=south, xshift=-3cm] {No} (pro6);
\draw [arrow] (pro6) |- (pro5);

\end{tikzpicture}
\caption[A flowchart]
{A flowchart for an iterative compression in Subsection~\ref{subsec:compression}.}
\label{fig:iterative}
\end{figure}
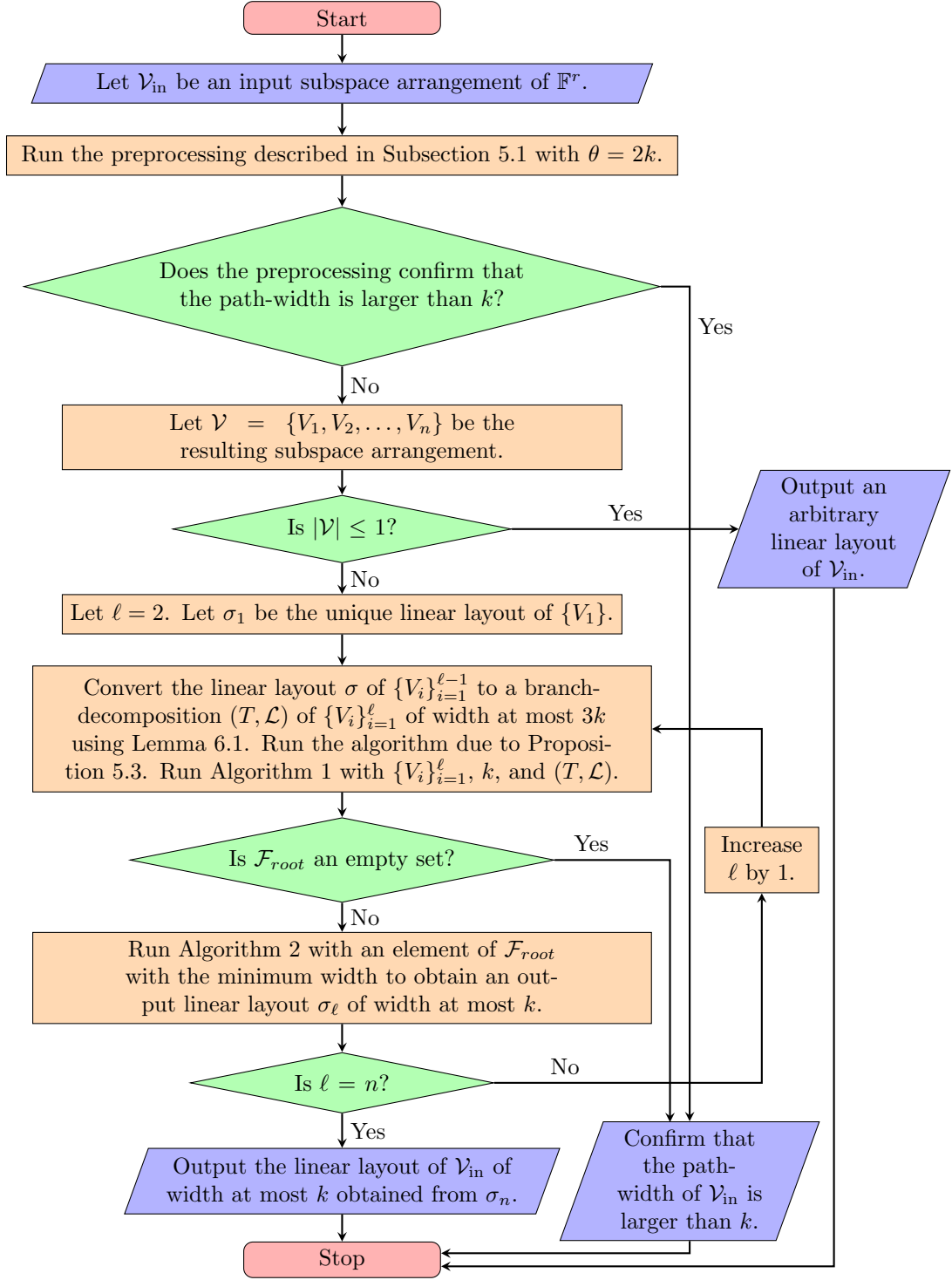

Here, we discuss the total running time.
Given $(T,\L)$ of  $\{V_1,\ldots , V_{\ell}\}$ whose width is at most $3k$, we first prepare bases of the boundary spaces by running the algorithm of Subsection~\ref{subsec:prepare}. This takes $O(\theta rm\abs{\V_{\text{in}}})=O(r\cdot k^2 \abs{\V_{\text{in}}}^2)$. Then, we run Algorithm~\ref{alg:fullset} and Algorithm~\ref{alg:order} to produce a linear layout of width at most $k$ if one exists. The second step can be done in time $\poly(\theta,\abs{\F},k) \cdot 2^{151 \theta k} \abs{\F}^{12\theta^2}\cdot n$ by Propositions~\ref{prop:FSruntime} and~\ref{prop:timeorder}. As we iterate this procedure for $\ell=2,\ldots , n$, the running time follows. 

\subsection{Using exact algorithms for partitioned matroids}\label{subsec:exact}

One can relate a subspace arrangement $\V$ with a partitioned matroid as introduced in \Hlineny{} and Oum~\cite{HO2006}. A \emph{partitioned matroid} is a pair $(M,\mathcal P)$ of a matroid $M$ and a partition $\mathcal P$ of $E(M)$;
the precise definition of matroids will be introduced in Subsection~\ref{subsec:matroid}.
\Hlineny{} and Oum~\cite{HO2006} defined the branch-width of a partitioned matroid in order to devise an algorithm to construct a branch-decomposition of an $\F$-represented matroid of width at most $k$ if it exists. Indeed, if $(M,\mathcal P)$ is the partitioned matroid and $M$ is $\F$-represented, then we can construct a subspace arrangement $\V$ by taking the span of the vectors in $P$ for each part $P\in \mathcal P$ as an element of $\V$. One can also convert $\V$ into a partitioned matroid easily by taking a basis for each $V\in \V$. 
Thus the problem of computing branch-width of $(M,\mathcal P)$ is equivalent to 
that of computing branch-width of $\V$.
Thus, we can identify $\F$-represented partitioned matroids with subspace arrangements over $\F$.

Then an algorithm by \Hlineny~and Oum for partitioned matroids naturally translates for a subspace arrangement.
\begin{THM}[{\Hlineny~and Oum~\cite{HO2006}}]\label{thm:decomposition2}
Let $\F$ be a fixed finite field.
There is, for some function $f$, an $O(f(\theta)m^3)$-time algorithm that, for an input %
subspace arrangement $\V=\{V_i\}_{i=1}^n$ of subspaces of $\F^r$ with $r\le m=\sum_{i=1}^n \dim V_i$ and an integer $\theta$,
either
\begin{itemize}
\item outputs a branch-decomposition of $\V$ of width at most $\theta$
\item or confirms the branch-width of $\V$ is larger than $\theta$.
\end{itemize}
\end{THM}
Proposition~\ref{prop:pwbw} states that if the path-width of $\V$ is at most $k$, then the branch-width of $\V$ is at most $2k$.
Thus by applying Theorem~\ref{thm:decomposition2} for $\theta=2k$, we obtain a branch-decomposition of width at most $2k$ unless $\V$ has path-width larger than $k$.
By running the algorithm in Section~\ref{sec:algorithm} once with this branch-decomposition, we deduce an $O(n^3)$-time algorithm to find a linear layout of width at most $k$ if it exists for an input subspace arrangement $\V=\{V_i\}_{i=1}^n$ over a fixed finite field $\F$. 
(Note that after the preprocessing in Subsection~\ref{subsec:preprocess}, $m\le \theta n = 2kn$ and therefore $O(m^3)=O(n^3)$.)
This completes the proof of the following theorem.
\begin{THM}\label{thm:subspace}
  Let $\F$ be a fixed finite field.
  Given, as an input, $n$ subspaces of $\F^r$ for some $r$ and a parameter $k$, in time $O(f(k)n^3+rm^2)$ for some function $f$, 
we can either find a linear layout $V_1,V_2,\ldots,V_n$ of the subspaces such that
\[
\dim (V_1+V_2+\cdots+V_i)\cap (V_{i+1}+V_{i+2}+\cdots+V_n)\le k
\]
for all $i=1,2,\ldots,n-1$, or confirm that no such linear layout exists,
where each $V_i$ is given by its spanning set of $d_i$ vectors and $m=\sum_{i=1}^n d_i$.
\end{THM}
This approach gives a better asymptotic running time in terms of $n$, but the hidden constant as a function of $k$ is bigger here. The algorithm in Theorem~\ref{thm:decomposition2} is based on the fact that forbidden minors for matroid branch-width at most $k$ has at most $(6^{k+1}-1)/5$ elements, proved by Geelen, Gerards, Robertson and Whittle~\cite{GGRW2003a}.%
\footnote{In \cite{GGRW2003a}, their definition of $\lambda_M$ is one bigger than ours and therefore their bound $(6^k-1)/5$ is translated into our bound $(6^{k+1}-1)/5$.}
 Though this bound allows the algorithm to search all forbidden $\F$-representable minors in ``constant'' time, it would take extremely long, as a function of $k$, to search all of them because there are too many $\F$-representable matroids to consider. 
Thus the approach in Subsection~\ref{subsec:compression} would be a lot easier to be implemented and used, though it gives an algorithm that runs in time $O(n^4+rm^2)$, worse than $O(n^3+rm^2)$ of Theorem~\ref{thm:subspace}. We warn that the algorithm in Subsection~\ref{subsec:compression} is still far from being practical because the upper bound on the size of a full set is gigantic and so the hidden constant is huge, but a lot smaller than that of Theorem~\ref{thm:subspace}.

\section{Application to path-width of $\F$-represented matroids}
\label{sec:matroid}

So far our paper focused on path-width of a subspace arrangement
and presented the algorithm in a self-contained manner.
In this section %
we will discuss how to avoid using Theorem~\ref{thm:decomposition2} of \Hlineny{} and Oum by adapting a simpler algorithm of \Hlineny{} while keeping the running time to be still $O(n^3)$ for the path-width of $n$-element matroids.

\subsection{Path-width and branch-width of a connectivity function}\label{subsec:conn}
We will first present the definition of path-width and branch-width 
of a connectivity function 
so that we can use in the next section.
An integer-valued function $f$ on the set $2^E$ of all subsets of a finite set $E$
is a \emph{connectivity function}  on $E$ if
\begin{enumerate}
\item $f(X)=f(E-X)$ for all $X\subseteq E$ (\emph{symmetric}),
\item $f(X)+f(Y)\ge f(X\cap Y)+f(X\cup Y)$ for all $X, Y\subseteq E$ (\emph{submodular}),
\item $f(\emptyset)=0$.
\end{enumerate}

Let $f$ be a connectivity function on $E$.
A \emph{linear layout} of $E$ is a permutation $\sigma=e_1,e_2,\ldots,e_n$ of $E$.
The \emph{width} of a linear layout $\sigma=e_1,e_2,\ldots,e_n$ of $E$ with respect to $f$ is
$\max_{1\le i\le n-1}f(\{e_1,e_2,\ldots,e_i\})$.
The \emph{path-width} of $f$ is the minimum width of all possible linear layouts of $E$.
(If $\abs{E}\le 1$, then the path-width is $0$.)

A \emph{branch-decomposition} of $E$ is a pair $(T, \L)$ of a subcubic tree $T$ and
a bijection $\L$ from all leaves of $T$ to $E$.
Every edge $e$ of $T$ induces a partition $(A_e,B_e)$ of $E$ given by components of $T-e$ and $\L$.
The \emph{width} of $e$ is defined as $f(A_e)$.
The \emph{width} of a branch-decomposition $(T,\L)$ is the maximum width of all edges of $T$.
The \emph{branch-width} of $f$ is the minimum width of all possible branch-decompositions of $E$.
(If $\abs{E}\le 1$, then there is no branch-decomposition and we say $f$ has branch-width $0$.)

\subsection{Matroids and vectors}\label{subsec:matroid}

A \emph{matroid} $M$ is a pair $(E,\mathcal{I})$ of a finite set $E$ and a set $\mathcal{I}$ of subsets of $E$
satisfying the following.
\begin{itemize}
\item The empty set is in $\mathcal{I}$.
\item If $X\in \mathcal{I}$, then every subset of $X$ is contained in $\mathcal{I}$.
\item If $X,Y \in \mathcal{I}$ and $\abs{X}>\abs{Y}$, then there exists an element $e\in X\setminus Y$
such that $Y\cup\{e\}\in \mathcal{I}$.
\end{itemize}
We write $E(M):=E$.
For a subset $X$ of $E$, the \emph{rank} $r_M(X)$ of $X$ is the maximum size of an independent subset of $X$.
Let $\lambda_M(X)=r_M(X)+r_M(E-X)-r_M(E)$ for all subsets $X$ of $E$.
It is well known that $\lambda_M$ is a connectivity function on $E$~{\cite[Lemma 8.2.9]{Oxley2011}}.
The \emph{path-width} of a matroid $M$ is the path-width of $\lambda_M$ and a \emph{path-decomposition} of width $k$ is a linear layout of $E(M)$ whose width with respect to $\lambda_M$ is $k$.
Similarly \emph{branch-width} and \emph{branch-decompositions} of a matroid $M$ are branch-width and branch-decompositions, respectively, of $\lambda_M$.

For a matrix $A$ over a field $\F$, let $E(A)$ be the set of all column vectors of $A$ and $\mathcal{I}(A)$ be the set of all linearly independent subsets of $E(A)$.
Then $M(A)=(E(A),\mathcal{I}(A))$ is a matroid, called the \emph{vector matroid}.
If a matroid $M$ admits a matrix $A$ over $\F$ such that $M=M(A)$, then
we say $M$ is \emph{$\F$-representable} or $M$ is \emph{representable} over $\F$ and such a matrix $A$ is called an \emph{$\F$-representation} of $M$ or a \emph{matrix representation} of $M$ over $\F$. A matroid is \emph{binary} if it is representable over the field of two elements.
When $M=M(A)$, then it is easy to observe \[\lambda_M(X)=\dim(\spn{X}\cap \spn{E(A)-X}).\]

We remark that in general we cannot hope to find an $\F$-representation of a matroid efficiently, even if we assume such a representation exists~\cite{SW1981,Truemper1982}, unless $\F=GF(2)$. So we will need that a representation of a matroid, or equivalently a multiset of vectors, is given as an input. Let us say that a matroid is \emph{$\F$-represented} if it is given with an $\F$-representation.

\subsection{A better algorithm for matroid path-width}
Recall that we give an algorithm that,
for a parameter $k$ and an input subspace arrangement $\V$ with its branch-decomposition of small width,
decides in time $O(n^3)$ whether its path-width is at most $k$ and
if so, outputs a path-decomposition of width at most $k$.
In Subsection~\ref{subsec:exact}, we used 
Theorem~\ref{thm:decomposition2} that gives
an algorithm by \Hlineny{} and Oum 
to obtain a branch-decomposition of bounded width.
However, that algorithm by \Hlineny{} and Oum \cite{HO2006} uses the huge but finite list of forbidden minors for the class of matroids of branch-width at most $k$ and so in practice it would be too slow to implement it. 
We give an alternative in Subsection~\ref{subsec:compression} yielding more direct and implementable algorithm 
but slower, only giving the time $O(n^4)$.

For the path-width of matroids, we can instead use %
the following $O(n^3)$-time algorithm by \Hlineny{},
which does not depend on the existence of finitely many forbidden minors,
than the algorithm due to Theorem~\ref{thm:decomposition2}.
  \begin{THM}[{\Hlineny~\cite{Hlineny2002}}]\label{thm:decomposition}
 Let $k$ be a fixed constant, $\F$ be a finite field.
 Let $\V$ be a subspace arrangement of $n$ $1$-dimensional subspaces of $\F^r$ with $r\le n$.
 There exists, for some function $f$, an $O(f(k)n^3)$-time algorithm that, %
 for a given $\V$, either
 \begin{itemize}
 \item outputs a branch-decomposition of $\V$ of width at most $3k$
 \item or confirms the branch-width of $\V$ is larger than $k$.
 \end{itemize}
 \end{THM}

We associate a matroid $M$ represented by vectors $v_1,v_2,\ldots,v_n$ in a vector space over a finite field $\F$
with a subspace arrangement $\V=\{\spn{v_1},\spn{v_2},\ldots,\spn{v_n}\}$. Then a linear layout of $\V$ of width $k$ is precisely a path-decomposition of $M$ having width $k$. This relation allows us to deduce the following.
In addition, we do not need to apply the preprocessing step discussed in Subsection~\ref{subsec:preprocess}, as long as we remove all \emph{loops} and \emph{coloops}\footnote{In $\F$-representable matroids, a loop corresponds to a zero vector and a coloop corresponds to a vector that is not spanned by the other vectors.
} from $M$.
\begin{THM}\label{thm:mainthm}
  Let $\F$ be a fixed finite field.
  There is an algorithm that, for an input $n$-element matroid given by its matrix representation over $\F$ having at most $n$ rows and  a parameter $k$, decides in time
  $O(f(k)n^3)$ for some function $f$ whether its path-width is at most $k$ 
and if so, outputs a path-decomposition of width at most $k$.
\end{THM}

\section{Application to linear rank-width and linear clique-width of graphs}
\label{sec:linearrankwidth}
As it is discussed in Section~\ref{sec:intro}, our full theorem on subspace arrangements provides a nice corollary to linear rank-width of graphs when applied to subspaces of dimension at most two. As linear rank-width is closely related to linear clique-width, this will also give a fixed-parameter tractable approximation algorithm for linear clique-width as well. We will review the definitions and discuss how to obtain the desired results from Theorem~\ref{thm:subspace}.

\subsection{Linear rank-width and rank-width of graphs}
Let $G=(V,E)$ be a graph and $A_G$ be its adjacency matrix,
which is the $V(G)\times V(G)$ matrix over the binary field
whose $(i,j)$-entry is $1$ if and only if $i$ and $j$ are adjacent in $G$.
For an $X\times Y$ matrix $M=(m_{ij})_{i\in X, j\in Y}$ and subsets $A\subseteq X$ and $B\subseteq Y$,
let $M[A,B]$ be the $A\times B$ submatrix  $(m_{ij})_{i\in A, j\in B}$ of $M$.
The \emph{cut-rank function} of $G$, denoted by $\rho_G$,
is defined to be \[\rho_G(S)=\rank(A_G[S,V(G)\setminus S])\]
for all subsets $S$ of $V(G)$.
Oum and Seymour~\cite{OS2005} showed that the cut-rank function is a connectivity function on $V(G)$
(which is defined in Subsection~\ref{subsec:conn}).
In Subsection~\ref{subsec:conn}, we discussed how to define path-width and branch-width of a connectivity function.
We define \emph{linear rank-width} and \emph{rank-width} of a graph $G$ to be path-width and branch-width, respectively, of the cut-rank function $\rho_G$.
For a graph $G$, a \emph{linear rank-decomposition} of width $k$ is a linear layout of $V(G)$ whose  width with respect to $\rho_G$ is $k$.
For a graph $G$, a \emph{rank-decomposition} of width $k$ is  a branch-decomposition of $\rho_G$ having width $k$.

\subsection{From linear rank-width of graphs to path-width of subspace arrangements}\label{subs:lrwpw}
We now explain how to relate the linear rank-width of a graph with
the path-width of a subspace arrangement.

Let $G$ be a graph with $V(G)=\{1,2,\ldots,n\}$.
Let $e_1,e_2,\ldots,e_n$ be the standard basis of $GF(2)^n$.
Each vertex $i$ is associated with a vector $v_i\in GF(2)^n$ such that
\[v_i=\sum_{\text{$j$ is adjacent to $i$ in $G$}} e_j.\]
Let $V_i=\spn{e_i,v_i}$ for each $i$.
Let $\V_G=\{V_1,V_2,\ldots,V_n\}$ be \emph{the subspace arrangement associated with} $G$.
The following lemma is equivalent to \cite[Lemma 7.1]{HO2006} and easily implies that
the path-width of $\V_G$ is precisely twice the linear rank-width of $G$.

\begin{LEM}\label{lem:lrw}
For $X\subseteq V(G)$,
$	\dim\left(\sum_{i\in X} V_i\right)\cap\left(\sum_{j\in V(G)-X} V_j\right)
    =2\rho_G(X)
$.
\end{LEM}

\begin{proof}
Observe that
$\dim\left(\sum_{i\in X} V_i\right)\cap\left(\sum_{j\in V(G)-X} V_j\right)= \dim \sum_{i\in X} V_i + \dim \sum_{i\in V(G)-X} V_i - n $. Since $\dim \sum_{i\in X}V_i=\abs{X}+ \rank(A_G[V(G)-X,X])$ and $\dim \sum_{i\in V(G)-X} V_i= n-\abs{X} +\rank(A_G[X,V(G)-X])$, the conclusion follows.
\end{proof}

Therefore, a linear rank-decomposition of width $k$ of a graph $G$ precisely corresponds to a linear layout of $\V_G$ of width $2k$
and a rank-decomposition of width $k$ of a graph $G$ is exactly a branch-decomposition of width $2k$ of $\V_G$.

\subsection{An algorithm for linear rank-width}
In order to run the algorithm in Theorem~\ref{thm:subspace}, we need a branch-decomposition of small width to be given. Instead of using methods in Section~\ref{sec:approx}, we may use the following algorithm by Oum~\cite{Oum2006} for rank-width.
\begin{THM}[Oum~{\cite{Oum2006}}]\label{thm:rwapprox}
 For fixed $k$, there is an $O(f(k)n^3)$-algorithm that, for an input $n$-vertex graph $G$, either obtains a rank-decomposition of $G$ of width at most $3k-1$ or confirms that the rank-width of $G$ is larger than~$k$.
\end{THM}
We will describe an algorithm that finds a linear rank-decomposition of width at most $k$ of an input graph if it exists.
For an input graph $G$ with $V(G)=\{1,2,\ldots,n\}$, we first run the algorithm due to Theorem~\ref{thm:rwapprox} in time $O(f(k)n^3)$.
If the algorithm confirms that the rank-width of $G$ is larger than $k$, then we deduce that the linear rank-width of $G$ is also larger than $k$.
Otherwise, we obtain a rank-decomposition $(T,\L)$ of $G$ of width at most $3k-1$.
Let $\V_G=\{V_1,V_2,\ldots,V_n\}$ be the subspace arrangement associated with $G$, 
which is defined in Subsection~\ref{subs:lrwpw}.
Let $\L'$ be a bijection from the set of all leaves of $T$ to $\V_G$
such that 
$\L'(v)=V_{\L(v)}$ for every leaf $v$ of $T$.
Then $(T,\L')$ is a branch-decomposition of a subspace arrangement $\V_G$.
By Lemma~\ref{lem:lrw}, the width of $(T,\L')$ is at most $6k-2$.
We then apply the algorithm proven in Theorem~\ref{thm:timecomplexity} with, as an input,
a subspace arrangement $\V_G$, a parameter $2k$, a branch-decomposition $(T,\L')$ of width at most $6k-2$.
If $\mathcal{F}_{root}$ is empty, then by Proposition~\ref{prop:FSroot} and Lemma~\ref{lem:lrw}, 
we can conclude that the linear rank-width of $G$ is larger than $k$.
Otherwise, we run Algorithm~\ref{alg:order} by choosing an element with the minimum width 
in $\mathcal{F}_{root}$.
Note that once we obtain a linear layout of $\V_G$ whose path-width at most $2k$,
then it is easy to convert a linear rank-decomposition of $G$ of width at most $k$.
We remark that there is no need to apply the preprocessing discussed in Subsection~\ref{subsec:preprocess} since the dimension of $V_i$ is at most $2$ for each $i$.
See Figure~\ref{fig:lrw} for a flowchart of this algorithm.

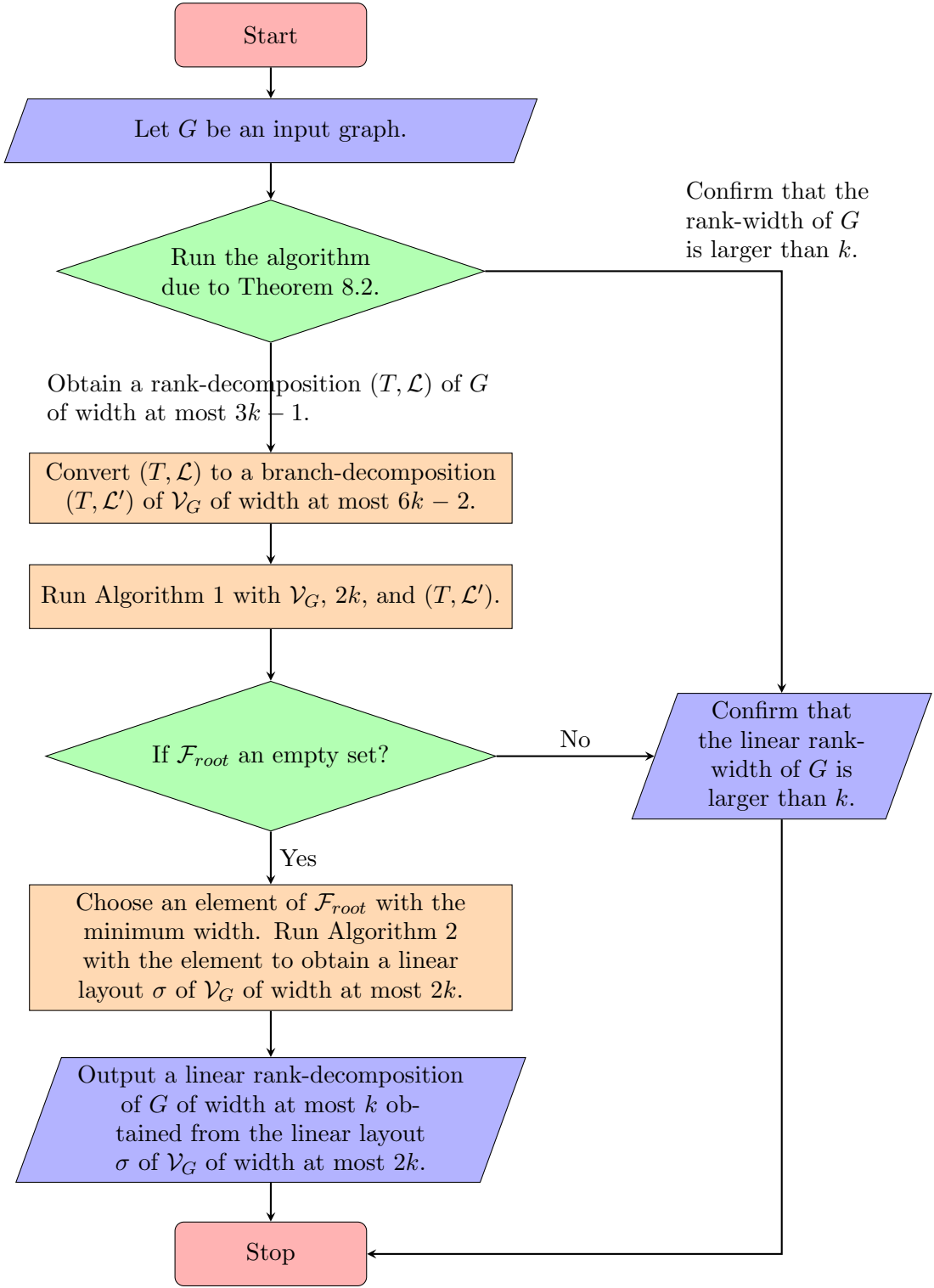
\begin{figure}
\centering
\tikzstyle{startstop} = [rectangle, rounded corners, minimum width=3cm, minimum height=1cm,text centered, draw=black, fill=red!30]
\tikzstyle{io} = [trapezium, trapezium left angle=70, trapezium right angle=110, minimum width=3cm, minimum height=1cm, text centered, text width=4.3cm, draw=black, fill=blue!30]
\tikzstyle{ioi} = [trapezium, trapezium left angle=70, trapezium right angle=110, minimum width=3cm, minimum height=1cm, text centered, text width=6.3cm, draw=black, fill=blue!30]
\tikzstyle{oo} = [trapezium, trapezium left angle=70, trapezium right angle=110, minimum width=3cm, minimum height=1cm, text centered, text width=3cm, draw=black, fill=blue!30]
\tikzstyle{process} = [rectangle, minimum width=3cm, minimum height=1cm, text centered, text width=7.3cm, draw=black, fill=orange!30]
\tikzstyle{decision1} = [diamond, aspect=3, minimum width=2cm, minimum height=0.5cm, text centered, text width=3.5cm, draw=black, fill=green!30]
\tikzstyle{decision} = [diamond, aspect=3, minimum width=2cm, minimum height=0.5cm, text centered, text width=5cm, draw=black, fill=green!30]
\tikzstyle{arrow} = [thick,->,>=stealth]

\begin{tikzpicture}[node distance=2cm]
\node (start) [startstop] {Start};
\node (in1) [io, below of=start, yshift=0.5cm] {Let $G$ be an input graph.};
\node (dec1) [decision1, below of=in1, yshift=-0.2cm] {Run the algorithm due to Theorem~\ref{thm:rwapprox}.};
\node (pro2) [process, below of=dec1, yshift=-1.4cm] {Convert $(T,\L)$ to a branch-decomposition $(T,\L')$ of $\V_G$ of width at most $6k-2$.};
\node (pro3) [process, below of=pro2, yshift=0.3cm] {Run Algorithm~\ref{alg:fullset} with $\V_G$, $2k$, and $(T,\L')$.};
\node (dec2) [decision, below of=pro3, yshift=-0.5cm] {If $\mathcal{F}_{root}$ an empty set?};
\node (pro4) [process, below of=dec2, yshift=-1cm] {Choose an element of $\mathcal{F}_{root}$ with the minimum width. Run Algorithm~\ref{alg:order} with the element to obtain a linear layout $\sigma$ of $\V_G$ of width at most $2k$.};
\node (out2) [ioi, below of=dec2, yshift=-3.7cm] {Output a linear rank-decomposition of $G$ of width at most $k$ obtained from the linear layout $\sigma$ of $\V_G$ of width at most $2k$.};
\node (out1) [oo, right of=dec2, xshift=6cm] {Confirm that the linear rank-width of $G$ is larger than $k$.};
\node (stop) [startstop, below of=out2, yshift=-.1cm] {Stop};

\draw [arrow] (start) -- (in1);
\draw [arrow] (in1) -- (dec1);
\draw [arrow] (dec1) -- node[text width=7cm] {Obtain a rank-decomposition $(T,\L)$ of $G$ of width at most $3k-1$.} (pro2);
\draw [arrow] (pro2) -- (pro3);
\draw [arrow] (pro3) -- (dec2);
\draw [arrow] (dec2) -- node[anchor=south] {No} (out1);
\draw [arrow] (dec2) -- node[anchor=west] {Yes} (pro4);
\draw [arrow] (pro4) -- (out2);
\draw [arrow] (out2) -- (stop);
\draw [arrow] (dec1) -| node[anchor=south,text width=3cm] {Confirm that the rank-width of $G$ is larger than $k$.} (out1);
\draw [arrow] (out1) |- (stop);
\end{tikzpicture}
\caption[A flowchart of Theorem~\ref{thm:lrw}.]
{A flowchart of the algorithm due to Theorem~\ref{thm:lrw}.}
\label{fig:lrw}
\end{figure}

\begin{THM}\label{thm:lrw}
For an input $n$-vertex graph and a parameter $k$,
we can 
decide in time $O(f(k)n^3)$ for some function $f$ whether its linear rank-width is at most $k$
and if so, find a linear rank-decomposition of width at most~$k$.
\end{THM}

\subsection{Application to linear clique-width}\label{subsec:linearcwd}
For an integer $k$, a \emph{$k$-expression} is an algebraic expression that consists of the following operations on graphs whose vertices are labeled by integers in $\{1,2,\ldots,k\}$:
\begin{itemize}
 \item $\cdot_i$ (a graph with a single vertex labeled by $i\in \{1,2,\ldots,k\}$).
 \item $G_1 \oplus G_2$ (the disjoint union of two vertex-labeled graphs $G_1$ and $G_2$).
 \item $\eta_{i,j}(G)$ with $i\neq j$ (adding an edge from each pair of a vertex of label $i$ and a vertex of label $j$).
 \item $\rho_{i\to j}(G)$ (relabeling all vertices of label $i$ to label $j$).
\end{itemize}
A $k$-expression is \emph{linear} if one operand of each disjoint union $\oplus$ has at most one vertex.
The \emph{clique-width} of a graph $G$ is the minimum $k$ such that there is a $k$-expression representing $G$ after ignoring labels of the vertices \cite{CO2000}.
Similarly the \emph{linear clique-width} of a graph $G$ is defined as the minimum $k$ such that there exists a linear $k$-expression representing $G$ after ignoring labels of the vertices.

By modifying the proof of Oum and Seymour~{\cite[Proposition 6.3]{OS2004}}, it is straightforward to prove the following lemma. We omit its easy proof.
\begin{LEM}\label{lem:lrwlcw}
  If a graph $G$ has linear rank-width $k$, then the linear clique-width of $G$ is at least $k$ and at most $2^k+1$.
  Furthermore, for each fixed $k$, a linear rank-decomposition of $G$ of width $k$ can be 
  converted to a linear $(2^k+1)$-expression representing $G$ in linear time.
\end{LEM}
Now we provide
an algorithm that makes an approximate linear expression of an input graph, which was unknown.
For an input $n$-vertex graph $G$, 
we run the algorithm due to Theorem~\ref{thm:rwapprox}.
If the linear rank-width of $G$ is larger than $k$, then by Lemma~\ref{lem:lrwlcw},
the linear clique-width of $G$ is also larger than $k$.
If the linear rank-width of $G$ is at most $k$ and thus the algorithm outputs 
a linear rank-decomposition of width at most $k$, then by Lemma~\ref{lem:lrwlcw},
we obtain a linear $(2^k+1)$-expression of $G$.
\begin{COR}\label{cor:cwd}
For an input $n$-vertex graph $G$ and a parameter $k$, 
we can find a linear $(2^k+1)$-expression of $G$ confirming that $G$ has linear clique-width at most $2^k+1$
or certify that $G$ has linear clique-width larger than $k$ in time $O(f(k)n^3)$ for some function $f$.
\end{COR}

\section{Computing path-width exactly when the input has bounded branch-width}\label{sec:exact}

In this section, we present, for each fixed integer $\theta$, a polynomial-time algorithm to compute path-width exactly with an optimal linear layout when the input subspace arrangement has branch-width at most $\theta$.
Our idea is similar to the idea of Bodlaender and Kloks~\cite{BK1996} who showed a polynomial-time algorithm to compute path-width of graphs %
of tree-width at most $\theta$.

The running time of this algorithm is polynomial 
because our algorithm in Theorem~\ref{thm:subspace} has the running time $O(c^k n^3)$ for some $c$ depending only on the underlying field $\F$ and $\theta$. Therefore if we can make sure that $k\le \theta\lfloor \log_2 n\rfloor$, then 
we can run the algorithm with $k=\theta\lfloor \log_2 n\rfloor$ 
and deduce the exact path-width from the full set.

The following proposition shows why the path-width is at most $\theta \lfloor \log_2 n\rfloor$.
Its proof follows proofs of a theorem of Kwon~\cite{Kwon2015} (also in Adler, Kant\'e, and Kwon~\cite{AMK2016}) on an inequality between rank-width and linear rank-width of graphs
and a theorem of Bodlaender, Gilbert, Hafsteinsson, and Kloks~\cite{BGHK1995} on an inequality between tree-width and path-width of graphs.

\begin{PROP}\label{prop:logn}
  Let $\V=\{V_1,V_2,\ldots,V_n\}$ be a subspace arrangement.
  If the branch-width of $\V$ is $k$, then the path-width of $\V$ is at most $k\lfloor \log_2 n\rfloor$.
\end{PROP}
\begin{proof}
  Let $(T,\L)$ be a branch-decomposition of width $k$.
  By picking an edge $e$ of $T$ and subdividing $e$ to add a root $r$ of degree $2$, we may make $T$ a rooted binary tree. 
  For each internal node $v$ of $T$, let $A_v$ be the subset of $\V$ that corresponds to all leaves below the left child of $v$ by $\L$
  and let $B_v$ the subset of $\V$ that corresponds to all leaves below the right child of $v$ by $\L$.
  We may assume that $\abs{A_v}\ge \abs{B_v}$ for all internal nodes $v$ of $T$.

  By permuting labels, we may assume that for all $1\le i<j\le n$, if $v$ is the common ancestor of $V_i$ and $V_j$ with the maximum distance from $r$, then $V_i$ is a descendant of the left child of $v$ and $V_j$ is a descendant of the right child of $v$.
  We claim that this linear layout $V_1,V_2,\ldots,V_n$ has width at most $k\lfloor \log_2 n\rfloor$.

  To see this, we need to prove that \[\dim (V_1+\cdots+V_{i-1})\cap (V_i+V_{i+1}+\cdots+V_n)\le k\lfloor \log_2 n\rfloor\] for each $i\in\{2,3,\ldots,n\}$.
  Let $P$ be the unique path $w_1=r,w_2,\ldots, w_t$ in $T$ from $r$ to a leaf $w_t$ with $\L(w_t)=V_i$.

  Now observe that   $\abs{A_{w_j}}\ge \abs{B_{w_j}}$ and therefore $\abs{B_{w_j}}\le \frac{1}{2} (\abs{A_{w_j}}+\abs{B_{w_j}})$.
  Thus if $w_{j+1}$ is the right child of $w_j$, then 
  \[
  \abs{A_{w_{j+1}}}+\abs{B_{w_{j+1}}}
  =\abs{B_{w_j}}\le \frac{\abs{A_{w_j}}+\abs{B_{w_j}}}{2}.\]
  If $w_{j+1}$ is the left child of $w_j$, then trivially 
  \[\abs{A_{w_{j+1}}}+\abs{B_{w_{j+1}}}
  = \abs{A_{{w_j}}}< \abs{A_{w_j}}+\abs{B_{w_j}}.\]
  Note that $\abs{A_{w_{t-1}}}+\abs{B_{w_{t-1}}}\ge 1$ because $V_i\in A_{w_{t-1}}\cup B_{w_{t-1}}$.

  Let $P=\{j\in\{1,2,\ldots,t-1\}:  \text{$w_{j+1}$ is the right child of $w_j$}\}$. 
  Then \[1\le \abs{A_{w_{t-1}}}+\abs{B_{w_{t-1}}}\le \frac1{2^{\abs{P}}} (\abs{A_1}+\abs{B_1})= \frac1{2^{\abs{P}}}{n}.\] 
  This implies that $\abs{P}\le \lfloor\log_2 n\rfloor$.
  Let $X=\bigcup_{j\in P}A_{v_{j}}$. Then by the definition of $P$, $X=\{V_1,V_2,\ldots,V_{i-1}\}$. Since $\abs{P}\le \lfloor \log_2 n\rfloor$ and $\dim \spn{A_{v_j}}\cap \spn{\V-A_{v_j}}\le k$, by Proposition~\ref{prop:submodular}, we deduce that $\dim \spn{X}\cap \spn {\V-X} \le \abs{P} k\le k\lfloor \log_2n\rfloor$.
This completes the proof.
\end{proof}

By using this inequality, we can prove the following main theorem.
\begin{THM}\label{thm:exact}
Let $\theta$ be a constant and $\F$ be a fixed finite field.
We can compute  the path-width of a given subspace arrangement $\V$ over $\F$
if the branch-width of $\V$ is at most $\theta$ in polynomial time.
Furthermore, we can find an optimal linear layout of~$\V$.
\end{THM}

\begin{proof}
  Let $\V=\{V_1,V_2,\ldots,V_n\}$.
  Since the branch-width of $\V$ is at most $\theta$, 
  using an algorithm by \Hlineny{} and Oum in Theorem~\ref{thm:decomposition2}, 
  we obtain a branch-decomposition of $\V$ of width at most $\theta$ in polynomial time.

  We will run an algorithm in Theorem~\ref{thm:timecomplexity} for $k=\theta\lfloor\log_2{n}\rfloor$ whose time complexity is $O(rm(m+\theta n)+\poly(\theta,\abs{\F},k) \cdot 2^{151 \theta k} \abs{\F}^{12\theta^2}\cdot n)$.
  This is a polynomial in the size of the input.

  By Proposition~\ref{prop:logn},
  the path-width of $\V$ is at most $\theta\lfloor \log_2 n\rfloor$
  and therefore the full set $\FS(\V,\{0\})$ is non-empty.
  Find a $\{0\}$-trajectory $\Gamma\in \FS(\V,\{0\})$ 
  with minimum width $w$. Then $w$ is the path-width of $\V$
  and we can find a corresponding linear layout by using the algorithm in Algorithm~\ref{alg:order}.
\end{proof}

\subsection{Path-width of matroids having bounded branch-width}
Computing the path-width of a given $\F$-represented matroid is NP-hard, shown by Kashyap~\cite{Navin2008}. 
However, if we restrict our input matroid to a small class, then it may be possible to compute path-width in polynomial time.
For instance, Koutsonas, Thilikos, and Yamazaki~\cite{KTY2011} presented a polynomial-time algorithm to compute path-width exactly 
for the cycle matroids of outerplanar graphs.
We remark that all cycle matroids of outerplanar graphs have very small branch-width.

Theorem~\ref{thm:exact} implies that there exists a polynomial-time algorithm to compute path-width with an optimal linear layout 
if the input matroid has bounded branch-width.

\begin{COR}
Let $\theta$ be a constant and $\F$ be a fixed finite field.
We can compute the path-width of an $\F$-represented matroid $M$ having branch-width at most $\theta$
in polynomial time.
Furthermore, we can find an optimal linear layout of $M$.
\end{COR}

\subsection{Linear rank-width of graphs having bounded rank-width}
Computing the linear rank-width of a given graph is also NP-hard by the result of Kashyap~\cite{Navin2008} and Oum~\cite{Oum2004}.
Adler, Kant{\'e}, and Kwon~\cite{AMK2016} proved that
the linear rank-width of graphs of rank-width at most $1$ can be computed in polynomial time.
They also asked the following question:
\begin{quote}
Can we compute in polynomial time the linear rank-width of graphs of rank-width at most $k$, for fixed $k\ge 2$?
\end{quote}
We answer this question in the affirmative by deducing Corollary~\ref{coro:exactlrw} from Theorem~\ref{thm:exact} 
by using the method in Section~\ref{sec:linearrankwidth}.

\begin{COR}\label{coro:exactlrw}
Let $\theta$ be a constant.
We can compute the linear rank-width of a graph $G$ having rank-width at most $\theta$ in polynomial time.
Furthermore, we can find an optimal linear rank-decomposition of $G$.
\end{COR}

\begin{proof}
Let $G$ be a graph of rank-width at most $\theta$ with $V(G)=\{1,2,\ldots,n\}$.
Let $\V_G=\{V_1,V_2,\ldots,V_n\}$ be the subspace arrangement associated with $G$,
which is defined in Subsection~\ref{subs:lrwpw}.
By Theorem~\ref{thm:exact},
we can obtain a linear layout $V_{\sigma(1)},V_{\sigma(2)},\ldots,V_{\sigma(n)}$ of $\V_G$ having optimal path-width.
Then by Lemma~\ref{lem:lrw},
a linear rank-decomposition ${\sigma(1)},{\sigma(2)},\ldots,{\sigma(n)}$ of $G$ 
has optimal linear rank-width.
\end{proof}

\section{Discussions}\label{sec:disc}

Can we decide whether the path-width of an input matroid is at most
$k$ for a fixed $k$ when the input matroid is general?
Our algorithm only works if the input matroid is $\F$-represented.
First we have to clarify our question. 
Binary strings of length $\poly(n)$ cannot represent all $n$-element matroids, 
because there are too many $n$-element
matroids~\cite{PW1971}. 
To study algorithmic problems on general matroids, typically we assume that a
matroid is given by an independence oracle, which answers whether a
given set is independent in the matroid in the constant time.  

The algorithm by Nagamochi~\cite{Nagamochi2012} can still decide in
time $O(n^{ck})$ for some $c$ whether the path-width of an input
matroid is at most $k$, even if the input matroid is given by its
independence oracle. However, the exponent of $n$ depends on $k$ and
so his algorithm is not fixed-parameter tractable when $k$ is a
parameter.
\begin{QUE}
  Does there exist, for each fixed $k$, 
  an algorithm that runs in time $O(n^c)$ for an
  input $n$-element matroid given with its independence oracle whether
  its path-width is at most $k$ for some~$c$?
\end{QUE}
As a weaker question, we ask the following. A matroid is \emph{binary}
if it is representable over the field of two elements.
\begin{QUE}
  Does there exist, for each fixed $k$, 
  an algorithm that runs in time $O(n^c)$ for an
  input $n$-element matroid given with its independence oracle 
  and confirms one of the following three outcomes for some~$c$?
  \begin{itemize}
  \item The path-width is at most $k$.
  \item The path-width is larger than $k$.
  \item The input matroid is not binary.
  \end{itemize}
\end{QUE}
One may attempt to construct a binary representation of the input matroid and then apply Theorem~\ref{thm:mainthm}. However, no polynomial-time algorithm can verify whether an input matroid given by its independence oracle is binary \cite{Seymour1981a} and so that strategy is not going to work.

By using the result of this paper, one can make a slightly weaker algorithm as follows.
\begin{PROP}\label{prop:binary}
  For each fixed $k$, there exists an algorithm that runs in time $O(n^3)$ for an input $n$-element matroid given by an independence oracle and confirms one of the following.
  \begin{itemize}
  \item Its path-width is at most $k$.
  \item Either its path-width is larger than $k$ or the matroid is not binary.
  \end{itemize}
\end{PROP}
\begin{proof}
  Let $M$ be the input matroid.
  By assuming the input matroid is binary, we can construct a potential binary representation, representing a binary matroid $M'$. This representation can be
  constructed in time $O(n^2)$ by picking a base and finding fundamental circuits for each element outside of the chosen base.
  The problem is that we do not yet know whether $M=M'$. But if $M$ is binary, then $M=M'$.

  We apply our algorithm in Theorem \ref{thm:mainthm} with $M'$ to
  decide whether the path-width of $M'$ is at most $k$. If the
  path-width of $M'$ is at most $k$, then we also obtain a linear
  layout of $M'$ having width at most $k$. It is now easy to verify
  whether this linear layout also has width at most $k$ in $M$ and if
  not, then $M\neq M'$ and so $M$ is not binary. 

  If our algorithm confirms that $M'$ has path-width larger than $k$,
  then either $M$ has path-width larger than $k$ or $M$ is non-binary.
\end{proof}

Lastly, it is natural to ask whether the running time of our algorithm in Theorem~\ref{thm:mainthm} can be improved. We assume that the input matroid is given as an $n\times n$ matrix over $\F$.
\begin{QUE}
Does there exist an algorithm that, given an $n$-element $\F$-represented matroid and a nonnegative integer $k\ge 2$, decides whether its path-width is at most $k$ in time $O(n^{3-\epsilon})$ for some $\epsilon >0$?
\end{QUE}

It can be readily verified that the preprocessing step of Section~\ref{subsec:preprocess} can be implemented with $O(n^{\omega})$ field operations using fast matrix multiplication~\cite{CheungKL13}, where $\omega < 2.38$. Recall that the algorithm of Theorem~\ref{thm:mainthm} uses the approximation algorithm by \Hlineny\ (see Theorem~\ref{thm:decomposition}, also~\cite{Hlineny2002}) for computing branch-decomposition. The latter algorithm runs in time $O(n^3)$ and it appears that there is an inherent bottleneck in the approach of~\cite{Hlineny2002} for obtaining a truly subcubic algorithm. 

For $k=0$ and $1$, the path-width of a $\F$-represented matroid $M$ can be computed in $O(n^{\omega})$-time. When $k=0$, the path-width of $M$ equals 0 if and only if every element of $M$ is either a loop or a coloop. This is true if and only if the number of non-zero columns in the matrix representing $M$ equals the rank of $M$, which can be checked in $O(n^{\omega})$-time. 

We consider the case $k=1$. For a fixed basis of $M$, we construct a potential binary representation $M'$ as in Proposition~\ref{prop:binary} in time $O(n^{\omega})$. Notice that $M'$ is isomorphic to $M$, written as $M'\cong M$, if $M$ is binary. It is known that the path-width of $M'$ equals the linear rank-width of the fundamental graph $G(M')$ of $M'$, whose vertices correspond to the columns of $M'$, see Proposition~3.1 in~\cite{Oum2004}. One can test whether $G(M')$ has linear rank-width at most one in time $O(n+m)$, where $n$ and $m$ are the number of vertices and edges of $G(M')$ respectively~\cite{Kante0KP15}. We argue the following.
\begin{itemize}
\item If $G(M')$ has linear rank-width at least two, then $M$ has path-width at least two regardless of whether $M\cong M'$ or not. This is because $M$ is a binary matroid if and only if $M$ does not contain the uniform matroid $U_{2,4}$ as a minor. As $U_{2,4}$ has path-width 2,  $M \ncong M'$ implies that $M$ contains $U_{2,4}$ as a minor, thus the path-width of $M$ is at least two. In case $M\cong M'$, the path-width of $M$ equals the path-width of $M'$, which is in turn equals the linear rank-width of $G(M')$.

\item If $G(M')$ has linear rank-width at most one, then consider the corresponding layout $\sigma$ of $M$. If the width of $\sigma$ on $M$ is at most one, we are done. If not, observe that  $M\ncong M'$; indeed $M \cong M'$ implies that an identical layout on $M$ and $M'$ results in the same width on both $M$ and $M'$. Especially, $M$ is non-binary, and thus has path-width at least two.
\end{itemize}


\begin{thebibliography}{10}

\bibitem{AMK2016}
I.~Adler, M.~M. Kant{\'e}, and O.-j. Kwon.
\newblock Linear rank-width of distance-hereditary graphs {I}. {A}
  polynomial-time algorithm.
\newblock {\em Algorithmica}, pages 1--36, 2016.

\bibitem{BBMNS2013}
P.~Berthom{\'e}, T.~Bouvier, F.~Mazoit, N.~Nisse, and R.~P. Soares.
\newblock An unified {FPT} algorithm for width of partition functions.
\newblock [Technical Report] RR-8372, Unpublished,
  https://hal.inria.fr/hal-00865575v2/document, 2013.

\bibitem{BFT2009}
H.~L. Bodlaender, M.~R. Fellows, and D.~M. Thilikos.
\newblock Derivation of algorithms for cutwidth and related graph layout
  parameters.
\newblock {\em J. Comput. System Sci.}, 75(4):231--244, 2009.

\bibitem{BGHK1995}
H.~L. Bodlaender, J.~R. Gilbert, H.~Hafsteinsson, and T.~Kloks.
\newblock Approximating treewidth, pathwidth, frontsize, and shortest
  elimination tree.
\newblock {\em J. Algorithms}, 18(2):238--255, 1995.

\bibitem{BK1996}
H.~L. Bodlaender and T.~Kloks.
\newblock Efficient and constructive algorithms for the pathwidth and treewidth
  of graphs.
\newblock {\em J. Algorithms}, 21(2):358--402, 1996.

\bibitem{BT2004}
H.~L. Bodlaender and D.~M. Thilikos.
\newblock Computing small search numbers in linear time.
\newblock In R.~Downey, M.~Fellows, and F.~Dehne, editors, {\em Parameterized
  and Exact Computation. Proc. First International Workshop, IWPEC 2004,
  Bergen, Norway, September 14-17, 2004.}, volume 3162 of {\em Lecture Notes in
  Comput. Sci.}, pages 37--48. Springer, 2004.

\bibitem{CheungKL13}
H.~Y. Cheung, T.~C. Kwok, and L.~C. Lau.
\newblock Fast matrix rank algorithms and applications.
\newblock {\em J. {ACM}}, 60(5):31:1--31:25, 2013.

\bibitem{COU1990}
B.~Courcelle.
\newblock The monadic second-order logic of graphs. i. recognizable sets of
  finite graphs.
\newblock {\em Information and Computation}, 85(1):12 -- 75, 1990.

\bibitem{CO2000}
B.~Courcelle and S.~Olariu.
\newblock Upper bounds to the clique width of graphs.
\newblock {\em Discrete Appl. Math.}, 101(1-3):77--114, 2000.

\bibitem{CO2004}
B.~Courcelle and S.~Oum.
\newblock Vertex-minors, monadic second-order logic, and a conjecture by
  {S}eese.
\newblock {\em J. Combin. Theory Ser. B}, 97(1):91--126, 2007.

\bibitem{EST1983}
J.~A. Ellis, I.~H. Sudborough, and J.~S. Turner.
\newblock Graph separation and search number.
\newblock In {\em Twenty-first annual {A}llerton conference on communication,
  control, and computing ({M}onticello, {I}ll., 1983)}, pages 224--233. Univ.
  Illinois, Urbana, IL, 1983.

\bibitem{FL1994}
M.~R. Fellows and M.~A. Langston.
\newblock On search, decision, and the efficiency of polynomial-time
  algorithms.
\newblock {\em J. Comput. System Sci.}, 49(3):769--779, 1994.

\bibitem{FRRS2009}
M.~R. Fellows, F.~A. Rosamond, U.~Rotics, and S.~Szeider.
\newblock Clique-width is {NP}-complete.
\newblock {\em SIAM J. Discrete Math.}, 23(2):909--939, 2009.

\bibitem{GGW2014}
J.~Geelen, B.~Gerards, and G.~Whittle.
\newblock Solving {R}ota's conjecture.
\newblock {\em Notices Amer. Math. Soc.}, 61(7):736--743, 2014.

\bibitem{GGRW2003a}
J.~F. Geelen, A.~M.~H. Gerards, N.~Robertson, and G.~Whittle.
\newblock On the excluded minors for the matroids of branch-width {$k$}.
\newblock {\em J. Combin. Theory Ser. B}, 88(2):261--265, 2003.

\bibitem{GGW2002}
J.~F. Geelen, A.~M.~H. Gerards, and G.~Whittle.
\newblock Branch-width and well-quasi-ordering in matroids and graphs.
\newblock {\em J. Combin. Theory Ser. B}, 84(2):270--290, 2002.

\bibitem{GGW2002cor}
J.~F. Geelen, A.~M.~H. Gerards, and G.~Whittle.
\newblock A correction to our paper ``{B}ranch-width and well-quasi-ordering in
  matroids and graphs''.
\newblock Manuscript, 2006.

\bibitem{HOS2007}
R.~Hall, J.~Oxley, and C.~Semple.
\newblock The structure of 3-connected matroids of path width three.
\newblock {\em European J. Combin.}, 28(3):964--989, 2007.

\bibitem{Hlineny2002}
P.~Hlin{\v e}n{\'y}.
\newblock A parametrized algorithm for matroid branch-width.
\newblock {\em SIAM J. Comput.}, 35(2):259--277, loose erratum (electronic),
  2005.

\bibitem{Hlineny2004}
P.~Hlin{\v e}n{\'y}.
\newblock Branch-width, parse trees, and monadic second-order logic for
  matroids.
\newblock {\em J. Combin. Theory Ser. B}, 96(3):325--351, 2006.

\bibitem{Hlineny2016}
P.~Hlin{\v e}n{\'y}.
\newblock Simpler self-reduction algorithm for matroid path-width.
\newblock Submitted. arXiv:1605.09520, 2016.

\bibitem{HO2006}
P.~Hlin{\v{e}}n{\'y} and S.~Oum.
\newblock Finding branch-decompositions and rank-decompositions.
\newblock {\em SIAM J. Comput.}, 38(3):1012--1032, 2008.

\bibitem{HK1996}
G.~B. Horn and F.~R. Kschischang.
\newblock On the intractability of permuting a block code to minimize trellis
  complexity.
\newblock {\em IEEE Trans. Inform. Theory}, 42(6, part 1):2042--2048, 1996.
\newblock Codes and complexity.

\bibitem{JMV1998}
K.~Jain, I.~M{\u{a}}ndoiu, and V.~V. Vazirani.
\newblock The ``art of trellis decoding'' is computationally hard---for large
  fields.
\newblock {\em IEEE Trans. Inform. Theory}, 44(3):1211--1214, 1998.

\bibitem{JKO2016a}
J.~Jeong, E.~J. Kim, and S.~Oum.
\newblock Constructive algorithm for path-width of matroids.
\newblock In {\em Proceedings of the Twenty-Seventh Annual ACM-SIAM Symposium
  on Discrete Algorithms (SODA 2016)}, pages 1695--1704, Philadelphia, PA, USA,
  2016. Society for Industrial and Applied Mathematics.

\bibitem{JKO2014}
J.~Jeong, O.-j. Kwon, and S.~Oum.
\newblock Excluded vertex-minors for graphs of linear rank-width at most {$k$}.
\newblock {\em European J. Combin.}, 41:242--257, 2014.

\bibitem{Kante0KP15}
M.~M. Kant{\'{e}}, E.~J. Kim, O.~Kwon, and C.~Paul.
\newblock {FPT} algorithm and polynomial kernel for linear rank-width one
  vertex deletion.
\newblock Submitted. arXiv:1504.05905, 2015.

\bibitem{Navin2008}
N.~Kashyap.
\newblock Matroid pathwidth and code trellis complexity.
\newblock {\em SIAM J. Discrete Math.}, 22(1):256--272, 2008.

\bibitem{KTY2011}
A.~Koutsonas, D.~M. Thilikos, and K.~Yamazaki.
\newblock Outerplanar obstructions for matroid pathwidth.
\newblock {\em Discrete Math.}, 315--316:95--101, 2014.

\bibitem{Kwon2015}
O.-j. Kwon.
\newblock {\em On the structural and algorithmic properties of linear
  rank-width}.
\newblock PhD thesis, KAIST, 2015.

\bibitem{MZ2009}
A.~Maheshwari and N.~Zeh.
\newblock I/{O}-efficient algorithms for graphs of bounded treewidth.
\newblock {\em Algorithmica}, 54(3):413--469, 2009.

\bibitem{Massey1978}
J.~L. Massey.
\newblock Foundation and methods of channel encoding.
\newblock In {\em Proc. Int. Conf. Information Theory and Systems}, volume~65,
  pages 148--157, 1978.

\bibitem{Nagamochi2012}
H.~Nagamochi.
\newblock Linear layouts in submodular systems.
\newblock In K.-M. Chao, T.-S. Hsu, and D.-T. Lee, editors, {\em ISAAC '12},
  volume 7676 of {\em Lecture Notes in Comput. Sci.}, pages 475--484. Springer
  Berlin Heidelberg, 2012.

\bibitem{Oum2004}
S.~Oum.
\newblock Rank-width and vertex-minors.
\newblock {\em J. Combin. Theory Ser. B}, 95(1):79--100, 2005.

\bibitem{Oum2006}
S.~Oum.
\newblock Approximating rank-width and clique-width quickly.
\newblock {\em ACM Trans. Algorithms}, 5(1):Art. 10, 20, 2008.

\bibitem{Oum2004a}
S.~Oum.
\newblock Rank-width and well-quasi-ordering.
\newblock {\em SIAM J. Discrete Math.}, 22(2):666--682, 2008.

\bibitem{OS2004}
S.~Oum and P.~Seymour.
\newblock Approximating clique-width and branch-width.
\newblock {\em J. Combin. Theory Ser. B}, 96(4):514--528, 2006.

\bibitem{OS2005}
S.~Oum and P.~Seymour.
\newblock Testing branch-width.
\newblock {\em J. Combin. Theory Ser. B}, 97(3):385--393, 2007.

\bibitem{Oxley2011}
J.~Oxley.
\newblock {\em Matroid theory}, volume~21 of {\em Oxford Graduate Texts in
  Mathematics}.
\newblock Oxford University Press, Oxford, second edition, 2011.

\bibitem{PW1971}
M.~J. Piff and D.~J.~A. Welsh.
\newblock The number of combinatorial geometries.
\newblock {\em Bull. London Math. Soc.}, 3:55--56, 1971.

\bibitem{RS1983}
N.~Robertson and P.~Seymour.
\newblock Graph minors. {I}. {E}xcluding a forest.
\newblock {\em J. Combin. Theory Ser. B}, 35(1):39--61, 1983.

\bibitem{RS1991}
N.~Robertson and P.~Seymour.
\newblock Graph minors. {X}. {O}bstructions to tree-decomposition.
\newblock {\em J. Combin. Theory Ser. B}, 52(2):153--190, 1991.

\bibitem{RS1995}
N.~Robertson and P.~Seymour.
\newblock Graph minors. {XIII}. {T}he disjoint paths problem.
\newblock {\em J. Combin. Theory Ser. B}, 63(1):65--110, 1995.

\bibitem{Seymour1981a}
P.~D. Seymour.
\newblock Recognizing graphic matroids.
\newblock {\em Combinatorica}, 1(1):75--78, 1981.

\bibitem{SW1981}
P.~D. Seymour and P.~N. Walton.
\newblock Detecting matroid minors.
\newblock {\em J. London Math. Soc. (2)}, 23(2):193--203, 1981.

\bibitem{Soares2013}
R.~Soares.
\newblock {\em Pursuit-{E}vasion, decompositions and convexity on graphs}.
\newblock PhD thesis, Universit{\'e} de Nice - Sophia-Antipolis, 2013.

\bibitem{Thilikos2000}
D.~M. Thilikos.
\newblock Algorithms and obstructions for linear-width and related search
  parameters.
\newblock {\em Discrete Appl. Math.}, 105(1-3):239--271, 2000.

\bibitem{TSB2005}
D.~M. Thilikos, M.~Serna, and H.~L. Bodlaender.
\newblock Cutwidth. {I}. {A} linear time fixed parameter algorithm.
\newblock {\em J. Algorithms}, 56(1):1--24, 2005.

\bibitem{TSB2005a}
D.~M. Thilikos, M.~Serna, and H.~L. Bodlaender.
\newblock Cutwidth. {II}. {A}lgorithms for partial {$w$}-trees of bounded
  degree.
\newblock {\em J. Algorithms}, 56(1):25--49, 2005.

\bibitem{Truemper1982}
K.~Truemper.
\newblock On the efficiency of representability tests for matroids.
\newblock {\em European J. Combin.}, 3(3):275--291, 1982.

\bibitem{BDFGR2012}
R.~van Bevern, R.~G. Downey, M.~R. Fellows, S.~Gaspers, and F.~A. Rosamond.
\newblock How applying {M}yhill-{N}erode methods to hypergraphs helps mastering
  the art of trellis decoding.
\newblock arXiv:1211.1299v1, 2012.

\bibitem{BDFGR2015}
R.~van Bevern, R.~G. Downey, M.~R. Fellows, S.~Gaspers, and F.~A. Rosamond.
\newblock {M}yhill-{N}erode methods for hypergraphs.
\newblock {\em Algorithmica}, 73(4):696--729, 2015.

\bibitem{Vardy1998}
A.~Vardy.
\newblock Trellis structure of codes.
\newblock In {\em Handbook of coding theory, {V}ol. {I}, {II}}, pages
  1989--2117. North-Holland, Amsterdam, 1998.

\bibitem{Viterbi67}
A.~J. Viterbi.
\newblock Error bounds for convolutional codes and an asymptotically optimum
  decoding algorithm.
\newblock {\em IEEE Trans. Information Theory}, 13(2):260--269, 1967.

\end{thebibliography}
\end{document}